\documentclass[11pt]{article}
\usepackage{fullpage}
\usepackage{amsmath,amsfonts,amsthm,mathrsfs,xspace,graphicx}
\usepackage[backref,colorlinks,citecolor=blue,bookmarks=true, hyperindex,breaklinks]{hyperref}
\usepackage{mathpazo}
\usepackage{color}
\usepackage{float}
\usepackage{amssymb,latexsym,amsthm}
\usepackage{braket}
\usepackage{comment}
\usepackage{xcolor}
\usepackage{enumitem}
\usepackage[affil-it]{authblk}

\makeatletter
\renewcommand{\paragraph}{%
  \@startsection{paragraph}{4}%
  {\z@}{1ex \@plus 1ex \@minus .2ex}{-1em}%
  {\normalfont\normalsize\bfseries}%
}
\makeatother

\setlist{itemsep=0.2\parsep,topsep=0.5\parsep}

\newtheorem{theorem}{Theorem}[section]

\newtheorem{definition}[theorem]{Definition}
\newtheorem{lem}[theorem]{Lemma}

\newtheorem{claim}[theorem]{Claim}

\usepackage[scr=boondoxo,scrscaled=1.05]{mathalfa}

\newcommand{\beq}{\begin{eqnarray}}
\newcommand{\eeq}{\end{eqnarray}}

\newcommand\proj[1]{\ket{#1}\!\bra{#1}}

\let\epsilon=\varepsilon %
\newcommand{\eps}{\epsilon} %

\newcommand{\class}[1]{\textup{#1}\xspace} %

\newcommand{\NP}{\class{NP}} %
\newcommand{\QMA}{\class{QMA}} %
\newtheorem{remark}[theorem]{Remark} %

\newcommand{\setft}[1]{\mathrm{#1}} %
\newcommand{\Density}{\setft{D}} %

\DeclareMathOperator{\poly}{poly}
\DeclareMathOperator{\negl}{negl}

\newcommand{\cX}{\mathcal{X}}
\newcommand{\cY}{\mathcal{Y}}

\newcommand{\N}{\mathbb{N}}
\newcommand{\C}{\mathbb{C}}
\newcommand{\Tr}{\mbox{\rm Tr}}
\newcommand{\crs}{\mathrm{crs}}

\newcommand{\anote}[1]{}
\newcommand{\tnote}[1]{}
\newcommand{\znote}[1]{}

\usepackage{subfiles}

\usepackage{tocloft}

\begin{document}

\title{Non-interactive zero-knowledge arguments for QMA, with preprocessing}

\date{}

\author[1,2]{Andrea Coladangelo\footnote{Email: acoladan@caltech.edu}}
\author[1,2]{Thomas Vidick\footnote{Email: vidick@caltech.edu}}
\author[3]{Tina Zhang\footnote{Email: tinazhang@caltech.edu}}
\affil[1]{Department of Computing and Mathematical Sciences, California Institute of Technology}
\affil[2]{Institute for Quantum Information and Matter, California Institute of Technology}
\affil[3]{Division of Physics, Mathematics and Astronomy, California Institute of Technology}

\maketitle

\begin{abstract}
A non-interactive zero-knowledge (NIZK) proof system for a language $L\in\NP$ allows a prover (who is provided with an instance $x \in L$, and a witness $w$ for $x$) to compute a \emph{classical certificate $\pi$} for the claim that $x\in L$ such that $\pi$ has the following properties: 1) $\pi$ can be verified efficiently, and 2) $\pi$ does not reveal any information about $w$, besides the fact that it exists (i.e. that $x \in L$). NIZK proof systems have recently been shown to exist for all languages in $\NP$ in the common reference string (CRS) model and under the learning with errors (LWE) assumption. 

We initiate the study of NIZK \emph{arguments} for languages in $\QMA$. An argument system differs from a proof system in that the honest prover must be efficient, and that it is only sound against (quantum) polynomial-time provers. Our first main result is the following: if LWE is hard for quantum computers, then any language in $\QMA$ has an \emph{NIZK argument} \emph{with preprocessing}. The preprocessing in our argument system consists of (i) the generation of a CRS and (ii) a \emph{single (instance-independent) quantum message} from verifier to prover. The instance-dependent phase of our argument system, meanwhile, involves only a single \emph{classical} message from prover to verifier. Importantly, verification in our protocol is entirely classical, and the verifier needs not have quantum memory; its only quantum actions are in the preprocessing phase. 
NIZK proofs of (classical) knowledge are widely used in the construction of more advanced cryptographic protocols, and we expect the quantum analogue to likewise find a broad range of applications. In this respect, the fact that our protocol has an entirely classical verification phase is particularly appealing. 

Our second contribution is to extend the notion of a classical \emph{proof of knowledge} to the quantum setting. We introduce the notions of \emph{arguments} and \emph{proofs of quantum knowledge} (AoQK/PoQK), and we show that our non-interactive argument system satisfies the definition of an AoQK, which extends its domain of usefulness with respect to cryptographic applications. In particular, we explicitly construct an extractor which can recover a quantum witness from any prover who is successful in our protocol. We also show that any language in QMA has an (interactive) \emph{proof of quantum knowledge}, again by exhibiting a particular proof system for all languages in QMA and constructing an extractor for it.

\end{abstract}

\renewcommand{\baselinestretch}{0.75}\normalsize
{
  \hypersetup{linkcolor=black}
  \tableofcontents
}
\renewcommand{\baselinestretch}{1.0}\normalsize

\newpage

\section{Introduction}

The paradigm of the interactive proof system is commonly studied in cryptography and in complexity theory. Intuitively speaking, an interactive proof system is a protocol in which an \emph{unbounded} prover attempts to convince an \emph{efficient} verifier that some problem instance $x$ is in some language $L$. The verifier represents an entity less computationally powerful or less informed than the prover; the prover holds some knowledge that the verifier does not (namely, that $x \in L$), and the prover attempts to convince the verifier of this knowledge. We say that there is an interactive proof system \emph{for a language $L$} if the following two conditions are satisfied. Firstly, for any $x \in L$, there must exist a prover (the `honest' prover) which causes the (honest) verifier to accept in the protocol with high probability; and secondly, for any $x \notin L$, there is no prover which can cause the honest verifier to accept, except with some small probability. These two conditions are commonly referred to as the `completeness' and `soundness' conditions. We can also consider a relaxed soundness condition where, when $x \notin L$, we require only that it be computationally intractable (rather than impossible) to cause the verifier to accept. A protocol satisfying this relaxed soundness condition, and which has an efficient honest prover, is known as an interactive \emph{argument} system.

Some interactive proof and argument systems satisfy a third property known as \emph{zero-knowledge} \cite{goldwasser1985knowledge}, which captures the informal notion that the verifier (even a dishonest verifier) `learns no new information' from an interaction with the honest prover, except for the information that $x \in L$. This idea is formalised through a \emph{simulator}, which has the same computational powers as the verifier $V$ does, and can output transcripts that (for $x$ such that $x \in L$) are indistinguishable from transcripts arising from interactions between $V$ and the honest prover. As such, $V$ intuitively `learns nothing', because whatever it might have learned from a transcript it could equally have generated by itself. The property of zero-knowledge can be \emph{perfect} (PZK), \emph{statistical} (SZK) or \emph{computational} (CZK). The difference between these three definitions is the extent to which simulated transcripts are indistinguishable from real ones. In a PZK protocol, the simulator's output distribution is \emph{identical} to the distribution of transcripts that the honest prover and (potentially dishonest) verifier generate when $x \in L$. In SZK, the two distributions have negligible statistical distance, and in CZK, they are computationally indistinguishable. In this work we will primarily be concerned with CZK.

A \emph{non-interactive} proof system (or argument system) is a protocol in which the prover and the verifier exchange only a single message that depends on the problem instance $x$. (In general, an instance-independent setup phase may be allowed in which the prover and verifier communicate, with each other or with a trusted third party, in order to establish shared state that is used during the protocol execution proper. We discuss this setup phase in more detail in the following paragraph.) Non-interactive zero-knowledge (NIZK) proofs and arguments have seen widespread application in classical cryptography, often in venues where their interactive counterparts would be impracticable---including, notably, in CCA-secure public-key cryptosystems \cite{naor1990cca, sahai1999nmcca}, digital signature schemes \cite{bellare1990nizksig, chaum1992wallet, bellare2003groupsig}, verifiable delegated computation~\cite{parno2013pinocchio} and, recently, a number of blockchain constructions \cite{gennaro2013quadratic,zcash14, coda17}. A particularly attractive feature of classical NIZK systems is that they can be amplified \emph{in parallel} to achieve better security parameters \cite{blum1991nizk}, which is in general not true of their interactive (private-coin) counterparts. 

It is known \cite{goldreich1994imposs} that NIZK proofs and arguments in the \emph{standard model} (namely, the model where the only assumption is that adversarial entities are computationally efficient) exist only for languages in BPP. As such, in order to construct NIZK protocols for more interesting languages, it is customary to consider \emph{extended} cryptographic models. Examples of these include the \emph{common reference string} (CRS) model, in which the verifier and the prover are assumed to begin the protocol sharing access to a common string sampled from a specified distribution; and the \emph{random oracle} (RO) model, in which prover and verifier have access to an efficiently evaluable function that behaves like a function sampled uniformly at random from the set of possible functions with some specified, and finite, domain and range. In these extended models, and under certain computational hardness assumptions, non-interactive computational zero-knowledge proof systems for all languages in NP are known. For instance, Blum, Santis, Micali and Persiano \cite{blum1991nizk} showed in 1990 that NIZK proofs for all languages in NP exist in the CRS model, assuming that the problem of quadratic residuosity is computationally intractable.

At this point, a natural question arises: what happens in the \emph{quantum} setting? Ever since Shor's algorithm for factoring \cite{shor1995factoring} was published in 1995, it has been understood that the introduction of quantum computers would render a wide range of cryptographic protocols insecure. For example, quadratic residuosity is known to be solvable in polynomial time by quantum computers. Given that this is so, it is natural to ask the following question: in the presence of quantum adversaries, is it still possible to obtain proof systems for all languages in NP that are complete and sound, and if it is, in which extended models is it feasible? This question has been studied in recent years. For example, Unruh showed in \cite{unruh2015nizk} that quantum-resistant NIZK proof systems for all languages in NP exist in the quantum random oracle (QRO) model, a quantum generalisation of the random oracle model. More recently, Peikert and Shiehian \cite{peikert2019noninteractive} achieved a more direct analogue of Blum et al.'s result, by showing that NIZK proofs for all languages in NP exist in the CRS model, assuming that learning with errors (LWE)---a problem believed to be difficult for quantum computers---is computationally intractable. \footnote{Peikert and Shiehan construct, based on LWE, a NI(C)ZK proof system in the common \emph{reference} string model, and a NI(S)ZK argument system in the common \emph{random} string model. They do not explicitly consider the applications of either result to the quantum setting. We show, however, for our own purposes, that the latter of these results generalises to quantum adversaries. In other words, we show (in section \ref{sec:nizk-np}) that the Peikert-Shiehan NIZK \emph{argument} system in the common \emph{random} string model is adaptively sound against quantum adversaries and adaptively (quantum computational) zero-knowledge.}

However, the advent of large-scale quantum computers would not only render some cryptosystems insecure; it would also provide us with computational powers that extend those of our current classical machines, and give rise to new cryptographic tasks that were never considered in the classical literature. A second natural question which arises in the presence of quantum computers is the following: in which models is it possible to obtain a NIZK proof or argument system not only for all languages in NP, but for all languages in `quantum NP' (i.e. QMA)? Loosely speaking, NIZK protocols for NP languages allow the prover to prove any statement that can be checked efficiently by a classical verifier who is given a classical witness. A NIZK protocol for QMA languages would, analogously, allow the prover to prove to the verifier (in a non-interactive, zero-knowledge way) the veracity of statements that require a quantum witness and quantum computing power to check. To our knowledge, the question of achieving NIZK protocols for QMA has not yet been studied. In 2016, Broadbent, Ji, Song and Watrous \cite{broadbent2016zero} exhibited a zero-knowledge proof system for QMA with an efficient honest prover, but their protocol requires both quantum and classical interaction.

In this work, our first contribution is to propose a non-interactive (computational) zero-knowledge argument system for all languages in QMA, based on the hardness of LWE, in which both verifier and prover are quantum polynomial time. The model we consider is the CRS (common reference string) model, augmented by a single message of (quantum) preprocessing. (The preprocessing consists of an instance-independent quantum message from the verifier to the prover.) The post-setup single message that the prover sends to the verifier, after it receives the witness, is classical; the post-setup verifier is also entirely classical; and, if we allow the prover and verifier to share EPR pairs \emph{a priori}, as in a model previously considered by Kobayashi \cite{kobayashi2002nizk}, we can also make the verifier's preprocessing message classical. Like classical NIZK protocols, our protocol shows itself to be receptive to parallel repetition (see section \ref{sec:parallel-ampl}), which allows us to amplify soundness concurrently without affecting zero-knowledge. Our model and our assumptions are relatively standard ones which can be fruitfully compared with those which have been studied in the classical setting. As such, this result provides an early benchmark of the kinds of assumptions under which NIZK can be achieved for languages in QMA.

Our second contribution is to show that our protocol also satisfies a notion of \emph{argument of quantum knowledge}. In the classical setting, some proof systems and argument systems for NP languages satisfy a stronger notion of soundness
 wherein a witness can be \emph{extracted} from any prover $P$ who convinces the verifier to accept with high probability. More formally, in such a setting, there is an \emph{extractor} machine which---given black-box access to any $P$ who convinces the verifier to accept with high probability (on the input $x$)---is able to efficiently compute a witness $w$ that testifies that the problem instance $x$ is in the language $L$. Such protocols are known as \emph{proofs} and \emph{arguments of knowledge} (PoK and AoK). Intuitively speaking, the notion of PoK/AoK is a framework for describing situations where the prover is not necessarily more powerful, but only \emph{better informed}, than the verifier. In these situations, the prover possesses knowledge (the witness $w$, which could represent a password or some other form of private information) that the verifier does not; and the prover wishes to convince the verifier, possibly in a zero-knowledge way (i.e. without revealing sensitive information), that it indeed `knows' or `possesses' the witness $w$ (so that it might, for example, be granted access to its password-protected files, or cash a quantum cheque). The idea of a machine `knowing' some witness $w$ is formalised by the existence of the extractor. 

Until now, the witness $w$ has always been classical, and the notion of a proof of \emph{quantum} knowledge (PoQK) has not been formally defined or studied. In this paper, we formulate a definition for a PoQK that is analogous to the classical definition of a PoK,\footnote{This definition is joint work with Broadbent and Grilo.} and we exhibit a protocol that is an (interactive) PoQK for any language in QMA.\footnote{This result is also obtained in independent and concurrent work by Broadbent and Grilo~\cite{broadbent2019zeroknowledge}.} We also introduce the notion of an \emph{argument of quantum knowledge} (AoQK), and we prove that our NIZK protocol for QMA is (under this definition) a zero-knowledge argument of quantum knowledge. We present our definitions of PoQK and AoQK in section~\ref{sec: poqk}.

There are two main difficulties in extending the classical notion of a PoK to the quantum setting. The first is that we must precisely specify how the extractor should be permitted to interact with the successful (quantum) prover. For this, we borrow the formalism of quantum interactive machines that Unruh~\cite{unruh2012quantum} uses in defining quantum proofs of \emph{classical} knowledge. The second difficulty is to give an appropriate definition of success for the extractor. In the classical setting, the NP relation $R$ which defines the set of witnesses $w$ for a problem instance $x$ is binary: a string $w$ is either a witness or it is not. In the quantum setting, on the other hand---unlike in the classical case, in which any witness is as good as any other---different witnesses might be accepted with different probabilities by some verification circuit $Q$ under consideration. In other words, some witnesses may be of better `quality' than others. In addition, because QMA is a probabilistic class, the choice of $Q$ (which is analogous to the choice of the NP relation $R$) is more obviously ambiguous than it is in the classical case. Different (and equally valid) choices of verifiers $Q$ for a particular language $L \in \mathrm{QMA}$ might have different probabilities of accepting a candidate witness $\rho$ on a particular instance $x$. In our definition, we define a `QMA relation' with respect to a fixed choice of verifying circuit (family) $Q$; we define the `quality' of a candidate witness $\rho$ for $x$ to be the probability that $Q$ accepts $(x, \rho)$; and we require that the successful extractor returns a witness whose quality lies strictly above the soundness parameter for the $\QMA$ relation.

Since we show that our protocol is a NIZK argument of quantum knowledge under our new definition, it can be used in settings where a prover wishes to prove it `knows' or `possesses' a quantum state. For example, this state could be a quantum money state~\cite{wiesner1983conjugate}. Suppose that the prover participated in the appropriate setup phase with the bank when the quantum money state was minted. Later on, using our protocol, the prover could---by sending a \emph{single classical} message, which can be processed \emph{classically}---demonstrate to the bank that it still possesses that quantum banknote, allowing the prover to cash a quantum cheque or transfer quantum funds.\footnote{For this to be possible, the prover needs to be provided with a homomorphic encryption of a classical description of the quantum money state. In the case of Wiesner's scheme~\cite{wiesner1983conjugate}, this is a Hamiltonian that is already in the Clifford form required for our protocol.} The fact that our protocol is an \emph{argument of quantum knowledge} means that, when it is used in this way (i.e. in the place of a more traditional quantum money verification protocol), it preserves (computationally speaking) the anti-counterfeiting security properties that are central to quantum money. 

Another example of an application in which the unique properties of our protocol might be useful is the setting of \emph{verifiable delegated computation}, in which a prover (who is generally a server to whom a client, the verifier, has delegated a quantum task) wishes to prove to the verifier a statement about a history state representing a certain computation. Suppose  that the prover and the verifier complete the setup phase of our protocol when the delegation occurs. After the setup phase is complete, \emph{the verifier does not need to preserve any quantum information}, meaning that it could perform the setup phase using borrowed quantum resources, and thereafter return to the classical world. When it receives the prover's single-message zero-knowledge proof, the verifier can verify its delegated computation without performing any additional quantum operations---a property that our protocol shares with protocols that have purely classical verification, such as Mahadev's classical-verifier argument system for QMA \cite{mahadev2018classical}. An additional advantage of our protocol, however, is that the server can free the quantum memory associated with the verifier's computation \emph{immediately} after the computation terminates, rather than holding the history state until the verifier is available to perform the verification.

\subsection*{The interactive protocol from~\cite{broadbent2016zero}}

Our protocol is inspired by the protocol exhibited in \cite{broadbent2016zero}, which gives a zero-knowledge (interactive) proof system for any language in $\QMA$. The \cite{broadbent2016zero} protocol can be summarized as follows. (For a more detailed exposition, see section~\ref{sec:bjsw}.)
\begin{enumerate}
\item The verifier and the prover begin with an instance $x$ of some interesting problem, the latter of which is represented by a (promise) language $L = (L_{yes},L_{no})\in \QMA$. The prover wishes to prove to the verifier that $x \in L_{yes}$. The first step is to map $x$ to an instance $H$ of the QMA-complete \emph{local Clifford Hamiltonian problem}. In the case that $x$ is a yes instance, i.e.\ $x\in L_{yes}$, the prover, who receives a witness state $\ket{\Phi}$ for $x$ as auxiliary input, performs the efficient transformation that turns the witness $\ket{\Phi}$ for $x$ into a witness $\ket{\Psi}$ for $H$. (The chief property that witnesses $\ket{\Psi}$ for $H$ have is that $\bra{\Psi} H \ket{\Psi}$ is \emph{small}---smaller than a certain threshold---which, rephrased in physics terminology, means that $\ket{\Psi}$ has \emph{low energy with respect to $H$}.) The prover then sends an \emph{encoding} of $\ket{\Psi}$ to the verifier (under a specified quantum authentication code which doubly functions as an encryption scheme). The prover also \emph{commits} to the secret key of the authentication code.
\item The Clifford Hamiltonian $H$ to which $x$ has been mapped can be written as a sum of polynomially many terms of the form $C^* \ket{0^k} \bra{0^k} C$, where $C$ is a Clifford unitary. (This is the origin of the name `Clifford Hamiltonian'.) The verifier chooses a string $r$ uniformly at random. $r$ plays a role analogous to that of the verifier's choice of edge to check in the 3-colouring zero-knowledge protocol introduced by \cite{goldwasser1985knowledge}: intuitively, $r$ determines the verifier's challenge to the prover. Each $r$ corresponds to one of the terms $C_r^*\ket{0^k}\bra{0^k}C_r$ of the Clifford Hamiltonian.

The verifier then measures the term $C_r^*\ket{0^k}\bra{0^k}C_r$ on the encoded witness (this can be done `homomorphically' through the encoding). The outcome $z$ obtained by the verifier can be thought of as an encoding of the true measurement outcome, the latter of which should be \emph{small} (i.e. correspond to low energy) if $\ket{\Psi}$ is a true witness. The verifier sends $z$ (its measurement outcomes) and $r$ (its choice of Hamiltonian term) back the prover. 
\item Finally, using a zero-knowledge NP proof system,\footnote{It is known that there are quantumly sound and quantumly zero-knowledge proof systems for NP: see \cite{watrous2009zero}.} the prover provides an (interactive) ZK proof for the following NP statement: there \emph{exists} an opening to its earlier (perfectly binding) commitment such that, if the verifier had the opened encoding keys, it \emph{would} accept. This is an NP statement because the witness string is the encoding keys. Proving that the verifier `would accept' amounts to proving that the verifier's measurement outcomes $z$, decoded under the keys which were committed to earlier, would correspond to a low-energy outcome. Because the proof that the prover provides is zero-knowledge, the verifier learns nothing substantial from this exchange, but it becomes convinced that it should accept.
\end{enumerate}

In the protocol from~\cite{broadbent2016zero}, it is critical to soundness that the prover sends the encoding of the witness to the verifier \textit{before} the verifier chooses $r$. The zero-knowledge property holds because the encoding that the prover applies to the witness state functions like an authenticated encryption scheme: its encryption-like properties prevent the verifier from learning anything substantial about the witness while handling the encoded state, and its authentication code--like properties ensure that the verifier cannot deviate very far from its honest behaviour. 

\subsection*{Our non-interactive protocol}

We wish to make the protocol from~\cite{broadbent2016zero} \emph{non-interactive}. To start with, we can replace the prover's proof in step 3 with a NIZK proof in the CRS model. NIZK proofs for all languages in NP have recently been shown to exist~\cite{canetti2019fiat,peikert2019noninteractive} based on the hardness of LWE only, and we prove that the Peikert-Shiehian construction from \cite{peikert2019noninteractive} remains secure (i.e. quantum computationally sound and zero-knowledge) against quantum adversaries, assuming that LWE is quantum computationally intractable. 
 However, the more substantial obstacle to making the \cite{broadbent2016zero} protocol non-interactive is the following: in order to do away with the verifier's message in step 2, it seems that the prover would have to somehow \textit{predict} $z$ (the verifier's measurement outcomes) and send a NIZK proof corresponding to this $z$. Unfortunately, in order for the authentication code to work, the number of possible outcomes $z$ has to be exponentially large (and thus the prover cannot provide a NIZK proof of consistency for each possible outcome). Even allowing for an instance-independent preprocessing step between the verifier and the prover, it is unclear how this impasse could be resolved. 

Our first main idea is to use \emph{quantum teleportation}. We add an instance-independent preprocessing step in which the verifier creates a number of EPR pairs and sends half of each to the prover. We then have the verifier (prematurely) make her measurement from step 2 \emph{during the preprocessing step} (and hence \textit{independently of the instance!}), and send the measurement outcomes $z$ to the prover. Once $x$ is revealed, the prover \textit{teleports} the encoded witness to the verifier, and sends the verifier the teleportation outcomes $d$, along with a commitment to his encoding keys. The prover then provides an NIZK proof of an opening to the committed keys such that $d, z$ and the encoding keys are consistent with a low-energy outcome. The hope is that, because the prover's and the verifier's actions commute (at least when the prover is honest), this protocol will be, in some sense, equivalent to one where the prover firstly teleports the witness, \emph{then} the verifier makes the measurements, and finally the prover sends an NIZK proof. This latter protocol would be essentially equivalent to the \cite{broadbent2016zero} protocol.

There are three main issues with this strategy:
\begin{enumerate}
\item In the preprocessing step, the verifier does not yet know what the instance $x$ (and hence what the Clifford Hamiltonian) is. Thus, she cannot measure the term $C_r^* \ket{0^k} \bra{0^k} C_r$, as she would have done in what we have called step 2 of the protocol from \cite{broadbent2016zero}.
\item The second issue is that the verifier cannot communicate her choice of $r$ in the preprocessing step in the clear. If she does, the prover will easily be able to cheat by teleporting a state that passes the check for the $r$th Hamiltonian term, but that would not pass the check for any other term.
\item The third issue is a bit more subtle. If the prover knows the verifier's measurement outcomes $z$ before he teleports the witness state to the verifier, he can misreport the teleportation outcomes $d$, and make a clever choice of $d$ such that $d,z$ and the committed keys are consistent with a low-energy outcome even when he does not possess a genuine witness.
\end{enumerate}

The first issue is resolved by considering the (instance-independent) verifying circuit $Q$ \emph{for the $\QMA$ language $L$} (recall that $Q$ takes as input both an instance $x$ and a witness state), and mapping $Q$ itself to a Clifford Hamiltonian $H(Q)$. (For comparison, in the protocol from \cite{broadbent2016zero}, it is the circuit $Q(x, \cdot)$ which is mapped to a Clifford Hamiltonian.) In the instance-dependent step, the prover will be asked to teleport a ``history state'' corresponding to the execution of the circuit $Q$ on input $(x, \ket{\Psi})$, where $\ket{\Psi}$ is a witness for the instance $x$. In the preprocessing step, the verifier will measure a uniformly random term from $H(Q)$, and will also perform a special measurement (with some probability) which is meant to certify that the prover put the correct instance $x$ into $Q$ when it was creating the history state. Of course, the verifier does not know $x$ at the time of this measurement, but she will know $x$ at the point where she needs to verify the prover's NIZK proof.

Our second main idea, which addresses the second and the third issues above (at the price of downgrading our proof system to an argument system), is to have the prover \emph{compute his NIZK proof homomorphically}. During the preprocessing step, we have the verifier send the prover a (computationally hiding) commitment $\sigma$ to her choice of $r$; and, in addition, we ask the verifier to send the prover a \textit{homomorphic encryption} of $r$, of the randomness $s$ used to commit to $\sigma$, and of her measurement outcomes $z$. At the beginning of the instance-dependent step, the prover receives a witness $\ket{\Psi}$ for the instance $x$. During the instance-dependent step, and after having received the verifier's ciphertexts in the preprocessing step, we ask the prover firstly to commit to some choice of encoding keys, and then to teleport to the verifier (an encoding of) the history state corresponding to the execution of $Q$ on input $(x, \ket{\Psi})$. Let $d$ be the outcome of the teleportation measurements. After the prover has committed to his encoding keys, we ask the prover to homomorphically encrypt $d$ and his encoding keys, and homomorphically run the following circuit: check that $r,s$ is a valid opening to $\sigma$, and (using the properties of the authentication code) check also that the verifier performed the honest measurement during preprocessing. If all the checks pass, then the prover \emph{homomorphically} computes an NIZK proof that there exist encoding keys consistent with his commitment such that these keys, together with $r,z,d$, indicate that the verifier's measurement result was a low-energy outcome. The homomorphic encryption safeguards the verifier against a malicious prover who may attempt to take advantage of knowing $r$, or of the freedom to cleverly choose $d$, in order to pass in the protocol without holding a genuine witness.

In summary, the structure of our protocol is as follows. Let $Q$ be a QMA verification circuit, and let $H(Q)$ be the Clifford Hamiltonian obtained from $Q$ by performing a circuit-to-Clifford-Hamiltonian reduction. 

\begin{enumerate}
\item \textit{(preprocessing step)} The verifier creates a (sufficiently large) number of EPR pairs, and divides them into `her halves' and `the prover's halves'. She interprets her halves as the qubits making up (an encoding of) a history state generated from an evaluation of the circuit $Q$. Then, the verifier samples $r$ (her `challenge') uniformly at random, and according to its value, does one of two things: either she measures a uniformly random term of $H(Q)$ on `her halves' of the EPR pairs, or she makes a special measurement (on her halves of the EPR pairs) whose results will allow her later to verify that the circuit $Q$ was evaluated on the correct instance $x$. Following this, the verifier samples a public-key, secret-key pair $(pk,sk)$ for a homomorphic encryption scheme. She sends the prover: 
\begin{enumerate}
\item $pk$;
\item the `prover's halves' of the EPR pairs;
\item a commitment to her choice of challenge $r$;
\item homomorphic encryptions of
\begin{enumerate}
\item $r$,
\item the randomness $s$ used in the commitment, and
\item the measurement outcomes $z$.
\end{enumerate}
\end{enumerate}

\item \textit{(instance-dependent step)} Upon receiving $x$, and a witness $\ket{\Psi}$, the prover computes the appropriate history state, and samples encoding keys. Then, he teleports an encoding of the history state to the verifier using the half EPR pairs that he previously received from her. Let $d$ be the teleportation measurement outcomes. The prover sends to the verifier:
\begin{enumerate}
\item $d$;
\item a commitment $\sigma$ to his encoding keys;
\item a homomorphic encryption of a NIZK proof (homomorphically computed) of the existence of an opening to $\sigma$ such that the opened keys, together with $d,z,r$, are consistent with a low-energy outcome.
\end{enumerate}
\end{enumerate}

\subsection*{Analysis}

Our protocol is a non-interactive, zero-knowledge argument system in the CRS model with a one-message preprocessing step. It is straightforward to see that the protocol satisfies completeness.

Intuitively, soundness follows from the fact that the encryptions the prover receives in the preprocessing step should be indistinguishable (assuming the prover is computationally bounded) from encryptions of the zero string. As such, the encryptions of $z,r,s$ (and the commitment to $r$) cannot possibly be helping the prover in guessing $r$ or in selecting a false teleportation measurement outcome $d'$ which makes $z,r,d'$ and the authentication keys consistent with a low-energy outcome. Soundness then essentially reduces to soundness of the protocol in~\cite{broadbent2016zero}.


The zero-knowledge property follows largely from the properties of the protocol in \cite{broadbent2016zero} that allowed Broadbent, Ji, Song and Watrous to achieve zero-knowledge. One key difference is that, in order to avoid rewinding the (quantum) verifier, the authors of \cite{broadbent2016zero} use the properties of an \emph{interactive coin-flipping protocol} to allow the efficient simulator to recover the string $r$ (recall that $r$ determines the verifier's challenge) with probability 1. (The traditional alternative to this strategy is to have the simulator guess $r$, and rewind the verifier if it guessed incorrectly in order to guess again. This is typical in classical proofs of zero-knowledge \cite{goldwasser1985knowledge}. However, because quantum rewinding \cite{watrous2009zero} is more delicate, the authors of \cite{broadbent2016zero} avoid it for simplicity.) As our protocol is non-interactive, we are unable to take the same approach. Instead, we ask the verifier to choose $r$ and commit to it using a commitment scheme with a property we call \emph{extractability}. Intuitively, extractability means that the commitment scheme \emph{takes a public key determined by the CRS}. We then show that the simulator can efficiently recover $r$ from the verifier's commitment by taking advantage of the CRS. For a candidate LWE-based extractable commitment scheme, see section \ref{sec: commitment}.

Another subtlety, unique to homomorphic encryption, is that the verifier may learn something about the homomorphic computations performed by the prover (and hence possibly about the encoding keys) by looking at the \emph{encryption randomness} in the encryption (of an NIZK proof) that the prover sends the verifier. (Recall that the verifier possesses the decryption key $sk$ for the homomorphic encryption scheme.) This leads us to require the use of a fully homomorphic encryption scheme which satisfies the property of \emph{circuit privacy}. For a definition of this property, see section \ref{sec:hom-enc}.

\subsection*{A non-interactive argument of quantum knowledge}
One desirable feature of our non-interactive argument system is that it is also an \textit{argument of quantum knowledge}. As we mentioned earlier, one of our contributions is to generalize the definitions of PoKs and AoKs for NP-relations to definitions of PoKs and AoKs for \emph{QMA relations}. In the latter setting, the prover wishes to convince the verifier that he `knows' or `possesses' the quantum witness for an instance of a $\QMA$ problem. In order to show that our protocol satisfies this additional property, we need to exhibit an extractor that, for any yes instance $x$, and given quantum oracle access to any prover that is accepted with high probability in our protocol, outputs a quantum state which is a witness for $x$. In section \ref{sec: aoqk}, we explicitly construct such an extractor $K$ for our non-interactive protocol. The intuition is the following. $K$ (the extractor) has oracle access to a prover $P^*$, and it simulates an execution of the protocol between $P^*$ and the honest verifier $V$. We show that, if $P^*$ is accepted in our protocol with sufficiently high probability, then it must teleport to $V$ (and hence to $K$) the \textit{encoding} $\tilde{\rho}$ of a witness state, and a commitment $\sigma$ to the encoding keys. If $K$ knew the encoding keys, it would be able to decode $\tilde{\rho}$, but it is not clear \emph{a priori} how $K$ could obtain such keys. Crucially, the same feature of our protocol that allows the \emph{zero-knowledge} simulator to extract $r$ from the verifier's commitment to $r$ also plays in $K$'s favour: when $K$ simulates an execution of the protocol, it samples a common reference string which is given to both $V$ and $P^*$, and in our protocol, the CRS contains a public key which $P^*$ uses to make his commitment. As such, in order to extract a witness from $P^*$, the extractor samples a CRS containing a public key $pk$ for which it knows the corresponding secret key $sk$, and provides this particular CRS as input to $P^*$. Then, when $K$ receives $\tilde{\rho}$ and $\sigma$ from $P^*$, it is able to extract the committed keys from $\sigma$, and use these to decode $\tilde{\rho}$.

\subsection*{An interactive proof of quantum knowledge}
Our non-interactive protocol is an \emph{argument system}, which means that it is sound only against computationally bounded provers. In Section~\ref{sec:bjsw-poqk}, we introduce a separate but complementary result to our NIZK argument (of knowledge) for QMA by showing that the zero-knowledge proof system for QMA exhibited in~\cite{broadbent2016zero} (with some minor modifications) is also a \emph{proof of quantum knowledge}. 

\subsection*{Related work}

A number of other recent works have explored questions that bear some degree of resemblance to those which we have considered. We compare our results with those obtained in these works.

\begin{enumerate}
\item In \cite{broadbent2019zeroknowledge}, Broadbent and Grilo present the definition of proofs of quantum knowledge that we arrived at in collaboration with them, and also exhibit both a \emph{proof of quantum knowledge for QMA} and a \emph{non-interactive proof system (with preprocessing) for QMA}. The Broadbent-Grilo non-interactive protocol has the advantages that it is both a proof system and statistically zero-knowledge (while our protocol is an argument system and computationally zero-knowledge), and that in addition it is a proof system not only for QMA but also for QAM. However, their extended model is the \emph{secret parameters} model, in which a trusted third party generates two classical strings $r_V$ and $r_P$, which are in general correlated, and shares $r_V$ with the verifier only and $r_P$ with the prover only. Our protocol is secure in the \emph{common reference string model with one message of preprocessing}, which is closer to models considered standard in the classical NIZK literature, and which (since the single message of preprocessing can be generated by the verifier and sent to the prover, so that the only role the trusted third party plays is to generate the CRS) lends itself more easily to cryptographic applications. Our protocol also has the feature that the verifier's quantum actions are entirely in the preprocessing phase, and are thus instance-independent, which again makes our protocol more attractive for practical applications while quantum resources remain considerably more expensive than classical ones.
\item In \cite{bitansky2019postquantum}, Bitansky and Shmueli present (as a corollary to their main result) a \emph{constant-round} computationally zero-knowledge argument system for QMA, based on secure quantum fully homomorphic encryption and post-quantumly secure LWE. (The former can be constructed from post-quantum LWE together with a circular security assumption.) The key ideas in their approach are somewhat similar to those in ours. Bitansky and Shmueli perform their analysis in the standard model, meaning that their result does not require a CRS. However, their protocol is constant-round rather than non-interactive, and the security of their protocol requires stronger cryptographic assumptions. (We require only that LWE is secure against quantum computers, because we use \emph{classical} and not quantum fully homomorphic encryption.)
\item In \cite{alagic2019twomessage}, Alagic, Childs, and Hung construct (based on the classical-verifier protocol of \cite{mahadev2018classical}) a non-interactive \emph{classical-verifier} argument system for QMA, which, like ours, requires one message of preprocessing. Their verifier is entirely classical, but they consider neither zero-knowledge nor proof-of-knowledge properties. Their result is also obtained in a different extended model from ours (namely the \emph{quantum random oracle} model).
\end{enumerate}

\subsection*{Open questions}

One of the intriguing properties of our protocol is that it is both computationally sound and computationally zero-knowledge. It would be interesting to obtain a non-interactive protocol with a similarly simple preprocessing phase for which at least one of these properties (soundness or zero-knowledge) can be shown to hold information-theoretically. (This is the case for the NIZK proof for NP languages from~\cite{peikert2019noninteractive}: the authors show that it is possible to obtain both NI(S)ZK arguments and NI(C)ZK proofs from LWE assumptions.)

It is also interesting to ask how far our preprocessing phase could be weakened or modified. Our protocol is non-interactive in the \emph{CRS model with a single message of quantum preprocessing}. In the classical literature, a broad variety of models for achieving NIZK exists, and it seems plausible (especially given the array of recent work we have outlined in the previous section) that the potential for variety in the quantum setting is even higher. It might be fruitful to ask: in which \emph{other} models can NIZK for QMA be achieved? In particular, is it possible to weaken our preprocessing phase (while minimally perturbing the other desirable properties of our scheme, including the fact that the instance-dependent phase only involves classical communication and classical verification) so that \emph{all} the information generated during the preprocessing phase is---like the CRS---generated by a trusted entity, and shared by this trusted entity with both prover and verifier? (Currently, the verifier in our protocol generates some private information during the preprocessing stage, including a secret key for the homomorphic encryption scheme.)

The technique we proposed to remove interaction from the protocol of \cite{broadbent2016zero} is based on two main ingredients: the use of quantum teleportation, which allows the verifier to \textit{anticipate} her measurements of the state she receives from the prover in the instance-dependent step, and the use of classical homomorphic encryption to allow the prover to demonstrate (homomorphically) that he has performed a certain computation correctly. These two ingredients work in tandem to ensure that the soundness and the zero-knowledge property of the \cite{broadbent2016zero} protocol are preserved. We believe that this technique could find use more broadly. In particular, it may be applicable as a general (soundness and zero-knowledge preserving) transformation to any interactive proof system for QMA with an efficient honest prover. We leave a more thorough investigation of this as a direction for future work. 

\paragraph{Acknowledgements.} We thank Anne Broadbent and Alex Grilo for useful discussions on the definition of proofs and arguments of quantum knowledge. We thank Zvika Brakerski and Vinod Vaikuntanathan for useful correspondence. 

Andrea Coladangelo is supported by AFOSR YIP award number FA9550-16-1-0495.
Thomas Vidick is supported by NSF CAREER Grant CCF-1553477, AFOSR YIP award number FA9550-16-1-0495, a CIFAR Azrieli Global Scholar award, MURI Grant FA9550-18-1-0161 and the IQIM, an NSF Physics Frontiers Center (NSF Grant PHY-1125565) with support of the Gordon and Betty Moore Foundation (GBMF-12500028).

\section{Preliminaries}

\subsection{Notation}

For an integer $\ell\geq 1$, $[\ell]$ denotes the set $\{1,\ldots,\ell\}$. We use $\poly(n)$ and $\negl(n)$ to denote an arbitrary polynomial and negligible function of $n$ respectively (a negligible function $f$ is any computable function such that $f(n)q(n)\to_{n\to \infty} 0$ for all polynomials $q$). 
For an integer $d\geq 1$, 
$\Density(\C^d)$ denotes the set of density matrices on $\C^d$, i.e.\ positive semidefinite $\rho$ on $\C^d$ such that $\Tr(\rho)=1$. For a set $S$ and an element $s \in S$, we write $s \xleftarrow{\$} S$ to mean that $s$ is sampled uniformly at random from $S$. For an integer $l$, we denote by $\{0,1\}^{\leq l}$ the set of binary strings of length at most $l$. We use the notation $S_N$ to denote the set of all permutations of a set of $N$ elements.

We use the terminology PPT for \emph{probabilistic polynomial time} and QPT for \emph{quantum polynomial time} to describe algorithms; formally, a PPT (resp. QPT) procedure $A$ is a uniformly generated family $\{A_n\}$ of classical (resp. quantum) polynomial-size circuits  such that $A_n$ takes classical (resp. quantum) inputs of length $\poly(n)$.

\subsection{The local Hamiltonian problem}

The \emph{local Hamiltonian problem} is a problem that is known to be complete for the complexity class \QMA, meaning that any problem in \QMA can be reduced efficiently to an instance of the local Hamiltonian problem. Informally stated, Kitaev's proof of $\QMA$-completeness of the local Hamiltonian problem~\cite{kitaev2002classical} shows that, given as input the classical description of an $(n+m)$-qubit circuit $U = U_T \cdots U_1$ acting on $n$ input qubits and $m$ ancilla qubits, where each $U_i$ is a $2$-qubit gate,\footnote{For convenience we assume that each gate is taken from the universal gate set $\Lambda(P),H\otimes H\}$ specified in~\cite{broadbent2016zero}. Here $\Lambda(P)$ is a controlled phase gate, and $H$ is the Hadamard gate.}  it is possible to efficiently compute the description of a $5$-local Hamiltonian 
\begin{equation}\label{eq:h-cliff}
H = H_{\textsf{in}} + H_{\textsf{out}} + H_{\textsf{prop}} + H_{\textsf{clock}}
\end{equation}
acting on $(n+m+T)$ qubits with the following properties:
\begin{enumerate}
\item whenever there is an input state of the form $\ket{\psi}_{\textsf{in}}\ket{0}_{\textsf{anc}}$ such that $U$ returns the computational basis outcome `$1$' with high probability on input $\ket{\psi}_{\textsf{in}}\ket{0}_{\textsf{anc}}$, then the smallest eigenvalue of $H$ is less than some threshold $a$,
\item whenever $U$ does not return the outcome `$1$' with high probability on any state of the form $\ket{\psi}_{\textsf{in}}\ket{0}_{\textsf{anc}}$, then all eigenvalues of $H$ are larger than a certain threshold $b$,
\item $b-a$ scales at least inverse polynomially with $T$,
\item $H$ has the following structural properties. Call the first $n$ qubits on which $H$ acts the \emph{input} qubits (labeled $(\textsf{in},i)$ for $i\in\{1,\ldots,n\}$), the next $m$ the \emph{ancilla} qubits (labeled $(\textsf{anc},i)$ for $i\in\{1,\ldots,m\}$), and the last $T$ the \emph{clock} qubits (labelled $(\textsf{clock},i)$ for $i\in\{1,\ldots,T\}$).
\begin{enumerate}
\item The term $H_{\textsf{out}}$ is a single $2$-local projection $\proj{0}_{\textsf{in},1} \otimes \proj{1}_{\textsf{clock},T}$ (we assume without loss of generality that the first input qubit is also the designated output qubit).
\item The term $H_{\textsf{clock}}$ is the sum of $2$-local projections $\proj{0}_{\textsf{clock},t} \otimes \proj{1}_{\textsf{clock},t+1}$ for $t\in\{1,\ldots,T-1\}$.
\item The term $H_{\textsf{prop}}$ is the sum of $5$-local terms, each of which acts on $3$ clock qubits and $2$ input or ancilla qubits and has the form $C^* \proj{0} C$ where $C$ is a $5$-qubit unitary.
\end{enumerate}
\end{enumerate}

In \cite{broadbent2016zero}, the authors show that it is possible to strengthen the reduction from~\cite{ksv} by showing an efficient reduction to a special class of $5$-local Hamiltonians, namely $5$-local \emph{Clifford} Hamiltonians. A $5$-local Clifford Hamiltonian is such that the terms in $H_{\textsf{prop}}$ have the form $C^* \proj{0} C$ where $C$ is a $5$-qubit \emph{Clifford} unitary. The following lemma expands on the result from \cite{broadbent2016zero}. While the first two items in the lemma are standard, the third is implicit in the proofs. We make it explicit for the application to proofs of quantum knowledge. 

\begin{lem}\label{lem:5local}
There is a polynomial $p$ such that the following hold. There is an efficient classical reduction that takes as input a $T$-gate quantum circuit $U$ such that $U$ has $n$ input qubits, $m$ ancilla qubits and one designated output qubit and returns a $5$-local Clifford Hamiltonian $H$ of the form~\eqref{eq:h-cliff} such that the following hold:
\begin{enumerate}
\item If there exists an $n$-qubit state $\ket{\psi}$ such that $U$ returns `$1$' on input $\ket{\psi}_{\textsf{in}}\ket{0^m}_{\textsf{anc}}$ with probability $1-\eps$ for some $0\leq\eps\leq 1$, then $H$ has an eigenvalue at most $\eps$;
\item If for all $\ket{\psi}$, $U$ returns `$1$' on input $\ket{\psi}_{\textsf{in}}\ket{0^m}_{\textsf{anc}}$ with probability at most $\frac{1}{2}$, then all eigenvalues of $H$ are at least $\frac{1}{p(T)}$;
\item There exists an efficiently computable quantum map $\tau$ (depending on $U$ only) from $(n+m+T)$ to $n$ qubits such that for any $\rho$ such that $\Tr(H\rho)\leq \delta$ it holds that a measurement of the output qubit of $U$ in the computational basis returns `$1$' on input $\tau(\rho)_{\textsf{in}} \otimes \proj{0^m}_{\textsf{anc}}$ with probability at least $1-\delta^{1/4}\poly(T,m)$.
\end{enumerate}
\end{lem}

\begin{proof}
We refer to the proof of~\cite[Theorem 2]{broadbent2016zero} for details on the construction of $H$ from $U$. The first two properties claimed in the lemma are shown explicitly in~\cite{broadbent2016zero}. To show the last, we first define the map $\tau$.  Given an $(n+m+T)$-qubit state $\rho$ the map first measures the $T$ clock qubits in the computational basis to obtain an outcome $t$. If $t \notin \{1^i 0^{T-i} : \, i\in \{0,\ldots,T\}\}$ then $\tau$ returns the state $\proj{0^n}$. Otherwise, let $\rho_t$ be the $(n+m)$-qubit post-measurement state. Then $\tau$ returns 
\[ \Tr_{\textsf{anc}} \big(U_1^\dagger \cdots U_t^\dagger (\rho_t) U_t \cdots U_1 \big)\;,\]
where $\Tr_{\textsf{anc}}$ traces out the $m$ ancilla qubits. 

We show that $\tau$ has the desired properties. Let $\rho$ be a state such that $\Tr(H \rho) \leq \delta$. Assume without loss of generality that $\rho = \proj{\psi}$ is a pure state. By definition $H_{\textsf{clock}}$ commutes with $H_{\textsf{in}}$, $H_{\textsf{out}}$, and $H_{\textsf{prop}}$. Moreover, $H_{\textsf{clock}}$ has smallest nonzero eigenvalue $1$, and as shown in Eq.~(14.16) in~\cite{ksv} $H_{\textsf{prop}}$ has smallest nonzero eigenvalue $\Omega(1/T^2)$. Using that $\bra{\psi} H_{\textsf{clock}} + H_{\textsf{prop}} \ket{\psi} \leq \delta$ it follows that if $\ket{\psi'}$ is the (sub-normalized) projection of $\ket{\psi}$ on the ground space (i.e. the zero-eigenvalue space) of $H_{\textsf{clock}} + H_{\textsf{prop}}$ then 
\begin{equation}\label{eq:overlap-1}
|\bra{\psi}\psi'\rangle| \,\geq\, 1-O\big(\sqrt{\delta} T\big)\;. 
\end{equation}

 Since $\ket{\psi'}$ is in the ground space of $H_{\textsf{clock}} + H_{\textsf{prop}}$, after re-normalization it can be written as 
\[ \ket{\psi''} = \ket{\psi'}_{normalised} = \frac{1}{\sqrt{T+1}}\sum_{i=0}^T \big( U_i\cdots U_1 \ket{\psi_0}_{\textsf{in},anc} \big)\otimes \ket{1^i 0^{T-i}}_{\textsf{clock}}\;,\]
where $\ket{\psi_0}_{\textsf{in},anc}$ is an arbitrary $(n+m)$-qubit state. Note that $|\braket{\psi | \psi'}| = |\braket{\psi | \psi''}|^2$. Using that $\bra{\psi} H_{\textsf{in}}+H_{\textsf{out}} \ket{\psi} \leq \delta$ and~\eqref{eq:overlap-1} it follows that 
\[ \bra{\psi''} H_{\textsf{in}}+H_{\textsf{out}} \ket{\psi''} \,=\, O\big(\sqrt{\delta} T\big)m\;,\]

where the extra factor of $m$ comes from $\|H_{\textsf{in}} + H_{\textsf{out}} \|\leq m+1$. From the definition of $H_{\textsf{in}}$ it follows that 
\begin{equation}\label{eq:overlap-2}
 \big\| \ket{\psi_0} - (I_{\textsf{in}} \otimes \proj{0^m}_{\textsf{anc}})\ket{\psi_0} \big\|^2 \,=\, O\big(\sqrt{\delta}T\big) Tm^2 \;,
\end{equation}

where the extra $m$ is due to there being $m$ ancilla qubits, and the extra $T$  is due to the fact that the first clock qubit in $\ket{\psi'}$ is a $0$ with probability only $1/(T+1)$. Similarly,  
\begin{equation}\label{eq:overlap-3}
\big|(\bra{1}_{\textsf{out}} \otimes I) U \ket{\psi_0} \big| \geq 1 -  O\big(\sqrt{\delta} T\big) Tm \;.
\end{equation}
As such, there exists a $\delta' =O({\delta}^{1/4} (mT)^c)$ for some universal $c>0$ and a  state 
\[\ket{\phi} = \frac{(I \otimes \bra{0^m}_{\textsf{anc}}) \ket{\psi_0}}{\| (I \otimes \bra{0^m}_{\textsf{anc}}) \ket{\psi_0} \|}\]
such that by~\eqref{eq:overlap-2} $\ket{\psi}$ (recall that $\rho = \proj{\psi}$, where $\rho$ is a state such that $\Tr(H \rho) \leq \delta$) is within trace distance at most $\delta'$ from the history state $\ket{\phi_{hist}}$ associated with the computation of $U$ on input $\ket{\phi}_{\textsf{in}} \ket{0^m}_{\textsf{anc}}$, and furthermore, using~\eqref{eq:overlap-3}, on the input $\ket{\phi}_{\textsf{in}} \ket{0^m}_{\textsf{anc}}$ $U$ returns `$1$' with probability at least $1-\delta'$. When given as input the history state $\ket{\phi_{hist}}$ the map $\tau$ returns the state $\ket{\phi}$. By contractivity of the trace distance under quantum maps it follows that on input $\ket{\psi}$ the map $\tau$ returns a state on which $U$ returns the outcome `$1$' with probability at least $1-2\delta'$.
\end{proof}

\subsection{The \cite{broadbent2016zero} protocol}
\label{sec:bjsw}

The following exposition is taken from \cite{vz19}.

In \cite{broadbent2016zero}, Broadbent, Ji, Song and Watrous describe a protocol involving a quantum polynomial-time verifier and an unbounded prover, interacting quantumly, which constitutes a zero-knowledge proof system for languages in QMA. (Although it is sound against arbitrary provers, the system in fact only requires an honest prover who is provided with a single witness state to perform quantum polynomial-time computations.) We summarise the steps of their protocol below. For details and fuller explanations, we refer the reader to~\cite[Section 3]{broadbent2016zero}.

\emph{Notation.}
Let $L$ be any language in QMA. For a definition of the \emph{$k$-local Clifford Hamiltonian problem}, see~\cite[Section 2]{broadbent2016zero} (this is the defined analogously to the $k$-local Hamiltonian problem, except that the Hamiltonian instance consists of Clifford terms, as introduced in the previous subsection). The $k$-local Clifford Hamiltonian problem (with exponentially small ground state energy) is QMA-complete for $k=5$; therefore, for all possible inputs $x$, there exists a 5-local Clifford Hamiltonian $H$ (which can be computed efficiently from $x$) whose terms are all operators of the form $C^* \ket{0^k}\bra{0^k} C$ for some Clifford operator $C$, and such that
\begin{enumerate}
\item if $x \in L$, the ground energy of $H$ is $\leq 2^{-p}$,
\item if $x \notin L$, the ground energy of $H$ is $\geq \frac{1}{q}$, 
\end{enumerate}
for some positive integers $p$ and $q$ which are bounded above by polynomials in $|x|$.

\emph{Parties.}
The proof system involves 
\begin{enumerate}
\item A \emph{verifier}, who implements a quantum polynomial-time procedure;
\item A \emph{prover}, who is unbounded, but who is only required by the protocol to implement a quantum polynomial-time procedure.
\end{enumerate}
The verifier and the prover communicate quantumly.

\emph{Inputs.} 
\begin{enumerate}
\item Input to the verifier:
\begin{enumerate}
\item The Hamiltonian $H$.
\item A quantum computationally concealing, perfectly binding (classical) commitment protocol. In this section, we refer to the commitment algorithm from this protocol as \textsf{commit}; $\textsf{commit}(\mu, s)$ takes as input a message $\mu$ and a random string $s$ and produces a commitment string $z$.
\item A proof system for NP sound against arbitrary quantum provers.
\end{enumerate}
\item Input to the prover:
\begin{enumerate}
\item The Hamiltonian $H$.
\item The $n$-qubit quantum state $\rho$, where $\rho$ is a ground state of the Hamiltonian $H$.
\item A quantum computationally concealing, perfectly binding (classical) commitment protocol.
\item A proof system for NP sound against arbitrary quantum provers.
\end{enumerate}
\end{enumerate}

\emph{Protocol.} 
\begin{enumerate}
\item \emph{The prover's encoding step}. The prover applies the following encoding to the witness state $\rho$.
\begin{figure}[H]
\rule[1ex]{16.5cm}{0.5pt}
\raggedright

\textbf{\textsf{Auth.Enc}:}

Parameters: $N(\cdot)$, a polynomially bounded function in $|x|$. ($N$ functions as a security parameter.)

Input: An $m$-qubit state $\rho$.

The prover firstly applies a concatenated Steane code (which maps every one qubit to $N(|x|)$ qubits) to each qubit in $\rho$. (For details on the concatenated Steane code, see~\cite[Appendix A.6]{broadbent2016zero}. It will be important to Broadbent et al.'s purposes, and ours, that this code admits transversal applications of Clifford operations.) It then executes the following steps:
\begin{enumerate}[topsep=0.8\baselineskip]
\item Concatenate $N$ trap qubits to the end of each logical qubit (alternatively, to the end of each $N$-qubit block) in the result of applying the concatenated Steane code to $\rho$. Each trap qubit is initialised uniformly at random to one of $\ket{0}, \ket{+}, \ket{+_y}$. ($\ket{+_y}$ here refers to the state $\frac{1}{\sqrt{2}}(\ket{0} + i \ket{1})$.) Denote the string that records the choices of trap qubits for all $m$ logical qubits by $t = t_1, \ldots, t_N \in \{\ket{0}, \ket{+}, \ket{+_y} \}^{mN}$.
\item Permute each $2N$-tuple of qubits in the result of (a) according to a uniformly random permutation $\pi \in S_{2N}$. (Note that the same permutation $\pi$ is applied to every $2N$-tuple.)
\item Apply a Pauli one-time pad $X^a Z^b$, for uniformly random $a, b \in \{0,1\}^{2mN}$, to the entire $2mN$-qubit state.
\end{enumerate}

\rule[2ex]{16.5cm}{0.5pt}\vspace{-.5cm}
\caption{The authentication code}
\label{fig: authentication code}
\end{figure}
We refer to $t,\pi,a,b$ as `the authentication keys' or `the encoding keys'.

The prover's encoding applied to $\rho$ is denoted by $E(\rho)$, and the procedure $E$ is fully determined by the encoding key $(t, \pi, a, b)$ which the prover chose to use. At this point, the prover sends the state $E(\rho)$ to the verifier, along with a commitment (using some perfectly binding, computationally concealing classical commitment protocol) to the tuple $(\pi, a, b)$. (A commitment to the sequence of trap qubits $t$ is unnecessary because, in a sense, the trap qubits exist only to check the verifier.) Let the prover's commitment string be denoted $z$.
\item \emph{Coin-flipping protocol.} The prover and the verifier execute a coin-flipping protocol, choosing a string $r$ of fixed length uniformly at random. This random string $r$ determines a local Hamiltonian term $H_r = C^*_r \ket{0^k} \bra{0^k} C_r$ that is to be tested. (This step can be implemented \cite{damgaard2009coinflip} using the same classical commitment protocol that the prover employed in the previous step.)
\item \emph{Verifier's challenge.} The verifier applies the Clifford $C_r$ transversally to the qubits on which the $k$-local Hamiltonian term $H_r$ acts nontrivially, and measures them in the standard basis. It then sends the measurement results $u_{i_1}, \ldots, u_{i_k}$ which it obtained to the prover. (Each $u_i$ is a $2N$-bit string, and $i_1, \ldots, i_k$ are the indices of the logical qubits on which the term $H_r$ acts nontrivially.)
\item \emph{Prover's response}. The prover receives the verifier's measurement results $u$, and firstly checks whether they cause a predicate $\tilde{Q}(t, \pi, a, b, r, u)$ to be satisfied. (We will explain the predicate $\tilde{Q}$ in more detail shortly. Intuitively, $\tilde{Q}$ is satisfied if and only if both verifier \emph{and} prover behaved honestly. Note that we have used the notation $\tilde{Q}$ to represent this predicate, while the authors of \cite{broadbent2016zero} simply call it $Q$. We add the tilde in order to differentiate their predicate from our predicate $Q$, the latter of which we define in Definition \ref{def:new-Q}.) If $\tilde{Q}$ is not satisfied, the prover aborts, causing the verifier to reject. If $\tilde{Q}$ is satisfied, then the prover proves to the verifier, using an NP zero-knowledge protocol, that there exists randomness $s_P$ and an encoding key $(t, \pi, a, b)$ such that $z = \textsf{commit}((\pi, a, b), s_P)$ and $\tilde{Q}(t, \pi, a, b, r, u) = 1$.

\end{enumerate}

We now describe the predicate $\tilde{Q}$ in precise terms. For convenience, Broadbent et al.\ define a predicate $\tilde{R}_r$, which represents the prover's check after it reverses the effects on $u$ of the one-time pad $X^a Z^b$, and then proceed to define $\tilde{Q}$ in terms of $\tilde{R}_r$. (As with $Q$, we denote the predicate found in \cite{broadbent2016zero} by $\tilde{R}_r$ instead of $R_r$, in order to differentiate it from our $R_r$, which we define in Definition \ref{def:new-R_r}.) We quote below the definition of $\tilde{R}_r$ used in \cite{broadbent2016zero}. (Note that we have altered their notation slightly: the strings that they call $y_i$ and $z_i$ we have called $p_i$ and $q_i$, and the strings that they call $c$ and $d$ we have called $e$ and $f$, to avoid conflicts.)

\begin{definition}[Definition of $\tilde{R}_r$]
\label{def:bjsw-R_r}
Let $u_{i_1}, \ldots, u_{i_k}$ be the measurement results that the verifier sent to the prover in step 3. For each $i \in \{i_1, \ldots, i_k\}$, define strings $p_i, q_i$ in $\{0,1\}^N$ such that $\pi(p_i \| q_i) = u_i$ (alternatively: $\pi^{-1}(u_i) = p_i \| q_i)$. The predicate $\tilde{R}_r(t, \pi, u)$ takes the value 1 if and only if the following two conditions are met:
\begin{enumerate}
\item $p_i \in \mathcal{D}_N$ for every $i \in \{i_1, \ldots, i_k\}$, and $p_i \in \mathcal{D}^1_N$ for at least one index $i \in \{i_1, \ldots, i_k\}$. ($\mathcal{D}_N = \mathcal{D}^0_N \cup \mathcal{D}^1_N$ is the set of all valid classical $N$-bit codewords of the concatenated Steane code. We note, as an aside, that $\mathcal{D}^0_N$ and $\mathcal{D}^1_N$ are both sets of size $8^t$, where $t$ is an even positive integer such that $7^t = N$; as such, $\mathcal{D}_N$ is polynomially sized.)
\item $\bra{q_{i_1} \cdots q_{i_k}} C^{\otimes n}_r \ket{t_{i_1} \cdots t_{i_k}} \neq 0.$
\end{enumerate}
\end{definition}
Here $\ket{t_{i_1} \cdots t_{i_k}}$ is the state of $mN$ qubits obtained by tensoring $\ket{0}, \ket{+}$ and $\ket{+_y}$ in the natural way.

Now we define the predicate $\tilde{Q}(t, \pi, a, b, r, u)$ in terms of $\tilde{R}_r$:

\begin{definition}[Definition of $\tilde{Q}$]
\label{def:bjsw-Q}

Let $e_1, \ldots, e_m, \: f_1, \ldots, f_m \in \{0,1\}^{2N}$ be the unique strings such that
\begin{gather*}
C_r^{\otimes 2N}(X^{a_1} Z^{b_1} \otimes \cdots \otimes X^{a_m} Z^{b_m}) = \alpha(X^{e_1} Z^{f_1} \otimes \cdots \otimes X^{e_m} Z^{f_m}) C_r^{\otimes 2N}\;,
\end{gather*}
for some $\alpha \in \{1, i, -1, -i\}$. (It is possible to efficiently compute $e = e_1, \ldots, e_m$ and $f = f_1, \ldots, f_m$ given $a, b$ and $C_r$.) The predicate $\tilde{Q}$ is defined by
\begin{gather*}
\tilde{Q}(t, \pi, a, b, r, u) = \tilde{R}_r(t, \pi, u \oplus e_{i_1} \cdots e_{i_k})\;.	
\end{gather*}
\end{definition}

\begin{definition}

\label{def: dec auth}

For notational convenience, we define a quantum procedure $Dec_{\pi, a, b}$ on a $2nN$-qubit state $\rho$ as follows:
\begin{enumerate}[itemsep=1pt,topsep=0pt]
\item Apply $X^a Z^b$ to $\rho$, to obtain a state $\rho'$.
\item Apply $\pi^{-1}$ to each $2N$-qubit block in the state $\rho'$, to obtain a state $\rho''$.
\item Discard the last $N$ qubits of each $2N$-qubit block in $\rho''$, to obtain a state $\rho'''$.
\item Return $\rho'''$.
\end{enumerate}

Intuitively, we think of $Dec_{\pi, a, b}$ as an inverse to the prover's encoding procedure $E$. (Note that this procedure does not depend on $t$.) 

\end{definition}

\subsection{Cryptographic primitives}

\subsubsection{Homomorphic encryption (with circuit privacy)}

\label{sec:hom-enc}

The following definitions are taken (with modifications) from \cite{bitansky2019pqzk}.

A homomorphic encryption scheme $\textsf{FHE} = (\textsf{FHE.Gen}, \textsf{FHE.Enc}, \textsf{FHE.Dec}, \textsf{FHE.Eval})$ is a quadruple of PPT algorithms which operate as follows:

\begin{itemize}
\item $\textsf{FHE.Gen}(1^\lambda)$: given the security parameter $\lambda$, outputs a secret key $sk$ and a public key $pk$.
\item $\textsf{FHE.Enc}(pk, x)$: using the public key $pk$, encrypts a message $x \in \{0,1\}^*$ into a ciphertext $\mathsf{ct}$ and outputs $\mathsf{ct}$.
\item $\textsf{FHE.Dec}(sk, \mathsf{ct})$: using the secret key $sk$, decrypts a ciphertext $\mathsf{ct}$ to recover a message $x$.
\item $\textsf{FHE.Eval}(pk, C, \mathsf{ct})$: applies circuit $C$ to ciphertext $\mathsf{ct}$ and outputs another ciphertext $\hat{\mathsf{ct}}$.
\end{itemize}

For any polynomial-size family of classical circuits $\mathcal{C} = \{\mathcal{C}_\lambda\}_{\lambda \in \mathbb{N}}$, the following properties hold:

\textbf{Perfect correctness.} For any $\lambda \in \mathbb{N}$, $x \in \{0,1\}^*$, and $C \in \mathcal{C}_\lambda$,

\begin{align*}
Pr \left[
\begin{array}{c|c}
& pk, sk \leftarrow \mathsf{FHE.Gen}(1^\lambda), \\
\mathsf{FHE.Dec}(sk, \hat{\mathsf{ct}}) & \mathsf{ct} \leftarrow \mathsf{FHE.Enc}(pk, x), \\
& \hat{\mathsf{ct}} \leftarrow \mathsf{FHE.Eval}(pk, C, \mathsf{ct})
\end{array}
\right] = 1.
\end{align*}

\textbf{Semantic security against quantum adversaries.} For any polynomial $\ell(\lambda)$ and polynomial-size quantum adversary $\mathsf{A^*}=\{A^*_\lambda, \rho_\lambda\}_{\lambda \in \mathbb{N}}$, there exists a negligible function $\mu(\cdot)$ such that, for every two messages $x_0, x_1$ of length $\ell$, and for every $\lambda \in \mathbb{N}$:

\begin{align*}
Pr \left[
\begin{array}{c|c}
& b \leftarrow \{0,1\}, \\
\mathsf{A}^*_\lambda(\mathsf{ct}) = b & pk, sk \leftarrow \mathsf{FHE.Gen}(1^\lambda), \\
& \hat{\mathsf{ct}} \leftarrow \mathsf{FHE.Enc}(pk, x_b)
\end{array}
\right] \leq \frac{1}{2} + \mu(\lambda).
\end{align*}

\textbf{Circuit privacy.} For the protocol described in section \ref{sec: protocol}, we require an additional property of \textsf{FHE} known as (malicious) \emph{circuit privacy}. We formalise the property of circuit privacy in the existence of a fifth PPT algorithm, \textsf{FHE.Refresh}, which operates as follows:

\begin{itemize}
\item $\textsf{FHE.Refresh}(pk, \mathsf{ct})$: takes a ciphertext $\mathsf{ct}$, and outputs another ciphertext $\mathsf{ct}'$ such that $\mathsf{FHE.Dec}(sk, \mathsf{ct}) = \mathsf{FHE.Dec}(sk, \mathsf{ct}')$
\end{itemize}

We say that \textsf{FHE} is a scheme \emph{with malicious statistical circuit privacy} if there exist unbounded algorithms $\textsf{Enc}^*(pk, x; 1^\lambda)$ and $\textsf{Dec}^*(pk, \mathsf{ct}; 1^\lambda)$ (not to be confused with $\mathsf{FHE.Enc}$ and $\mathsf{FHE.Dec}$) such that:

\begin{enumerate}
\item For any $x \in \{0,1\}^*$, any (possibly malformed) $pk^* \in \{0,1\}^{\mathrm{poly}(\lambda)}$, and any $\mathsf{ct} \in \mathsf{FHE.Enc}(pk^*, x)$, it holds that $\mathsf{Dec^*}(pk^*, \mathsf{ct}; 1^\lambda)= x$, and
\item For any $\lambda \in \mathbb{N}$, $C \in \mathcal{C}_\lambda$, and arbitrary (possibly malformed) $pk^*, \mathsf{ct}^* \in \{0,1\}^{\mathrm{poly}(\lambda)}$, the distribution of
\begin{equation*}
\mathsf{FHE.Refresh}\big(pk^*, \mathsf{FHE.Eval}(pk^*, C, \mathsf{ct}^*)\big)
\end{equation*}
(over the random coins of $\mathsf{FHE.Eval}$) is statistically indistinguishable from the distribution of
\begin{equation*}
\mathsf{Enc^*}\left(pk^*, C\Big(\mathsf{Dec^*}\big(pk^*, \mathsf{ct}^*; 1^\lambda\big)\Big); 1^\lambda \right)
\end{equation*}
over the random coins of $\mathsf{Enc^*}$.
\end{enumerate}

Intuitively, circuit privacy guarantees that `refreshed' ciphertexts are distributed like some known distribution which can be sampled from without the knowledge of $C$. (This is formalised by the property that the distribution of refreshed ciphertexts is negligibly close to the distribution of ciphertexts which $\mathsf{Enc^*}$ produces, where $\mathsf{Enc^*}$ does not have access to a description of $C$, but only to the output of $C$ on a certain input. In \cite{bitansky2019pqzk}, $\mathsf{Enc^*}$ is known as \textsf{Sim}, and $\mathsf{Dec^*}$ is known as $\mathsf{Ext}$.) In other words, refreshed ciphertexts do not leak information about the nature of the homomorphic computations that were performed on the ciphertext prior to the refresh operation.

\begin{remark}
We refer to ciphertexts which are sampled by $\mathsf{Enc^*}$ as `fresh' ciphertexts. This makes it convenient to state that ciphertexts originating from $\mathsf{FHE.Refresh}$ are `indistinguishable from fresh ciphertexts'.
\end{remark}

Fully homomorphic encryption schemes with malicious circuit privacy are known to exist, predicated on the assumption that LWE is (quantum) computationally intractable \cite{ostrovsky2013maliciously}. The circuit-private FHE scheme exhibited in \cite{ostrovsky2013maliciously} actually has $\mathsf{FHE.Refresh} = I$, where $I$ is the identity circuit (meaning that circuit privacy is `built into' evaluation). For clarity we explicitly consider a `refreshing' algorithm in our use of the primitive.

\subsubsection{Commitment scheme}
\label{sec: commitment}

In the protocol described in section \ref{sec: protocol}, we will twice make use of a commitment scheme consisting of a tuple of PPT algorithms (\textsf{gen}, \textsf{commit}, \textsf{reveal}, \textsf{verify}) which operate as follows:

\begin{itemize}[itemsep=1pt, topsep=0pt]
\item $\mathsf{gen(1^\ell)}$ takes as input a security parameter, and generates a public key $pk$ and a secret key $sk$. (We remark that $sk$ is never used in the real commitment procedure; it exists purely for simulation purposes.)
\item $\textsf{commit}(pk, b, s)$ takes as input a public key $pk$, a bit $b \in \{0,1\}$ to which to commit, and a string $s$, and produces a commitment string $z$.
\item $\textsf{reveal}(pk, z, b, s)$ outputs the inputs it is given.
\item $\textsf{verify}(pk, z, b, s)$ takes as argument a purported public key, commitment string, committed bit and random string, and outputs either $1$ (accept) or 0 (reject).
\end{itemize}

We require the following properties of this commitment scheme:

\begin{itemize}[itemsep=1pt, topsep=0pt]
\item \emph{Perfectly binding}: If $\textsf{commit}(pk, b, s) = \textsf{commit}(pk, b', s')$, then $b=b'$.
\item \emph{(Quantum) computationally concealing}: For any public key $pk \leftarrow \textsf{gen}(1^\ell)$ and any two messages $b, b'$, the distributions over uniformly random $s$ of $\textsf{commit}(pk, b, s)$ and $\textsf{commit}(pk, b', s)$ are quantum computationally indistinguishable.
\item \emph{Extractable}: There is a PPT algorithm $\textsf{recover}(pk, sk, z)$ which, for any $z \leftarrow \textsf{commit}(pk, b, s)$, outputs the unique $b$ such that $z = \textsf{commit}(pk, b, s)$.
\end{itemize}

\textbf{Instantiation based on LWE.} There are a number of folklore results on commitment schemes based on the intractability of LWE. Below, we sketch a scheme (one sufficient for our purposes) which is inspired by the LWE-based public-key encryption scheme presented by Regev in \cite{regev2009lattices}. Before we introduce our commitment scheme, we firstly provide an overview of the \cite{regev2009lattices} public-key encryption scheme.

\begin{definition}[\cite{regev2009lattices} public-key encryption scheme]

\label{def:regev-pke}

Let $m, n$ be security parameters ($m \geq n$), and let $q = \mathrm{poly}(n)$ be a prime. The \cite{regev2009lattices} public-key encryption scheme is a tuple of algorithms $(\textsf{KeyGen}, \textsf{Enc}, \textsf{Dec})$ of the following descriptions:

\begin{itemize}
\item $\textsf{KeyGen}(1^n, 1^m, q)$ generates a uniformly random matrix $A \in \mathbb{Z}_q^{m \times n}$, along with a uniformly random vector $\vec{s} \in \mathbb{Z}_q^n$ and a `short' noise vector $e \in \mathbb{Z}_q^n$. (We say that the vector $e$ is `short' because it has bounded norm. We will not go into detail here about the distribution from which $e$ is drawn. It suffices to state that this distribution is well-defined and efficiently sampleable.) It outputs $pk = (A, A\vec{s} + e)$ and $sk = \vec{s}$.
\item $\textsf{Enc}(pk, b)$ takes a message $b \in \{0,1\}$ and chooses a vector $w \in \{0,1\}^m$ uniformly at random. It outputs a ciphertext $c = (c_0, c_1) = (w^T A, w^T (A\vec{s} + e) + b \cdot \lfloor \frac{q}{2} \rfloor).$
\item $\textsf{Dec}(sk, c)$ is 0 if $c_1 - \langle c_0, \vec{s} \rangle$ is closer to 0 than to $\lfloor \frac{q}{2} \rfloor \mathrm{\: mod \: } q$, and 1 otherwise.
\end{itemize}

\end{definition}

We refer the reader to \cite[Section 5]{regev2009lattices} for the security analysis of the cryptosystem defined in Definition \ref{def:regev-pke}.

We now provide a sketch of an LWE-based commitment scheme satisfying the properties specified at the start of this section. In addition to the \cite{regev2009lattices} cryptosystem, we use the algorithms \textsf{GenTrap} and \textsf{Invert} from \cite[Theorem 2]{micciancio2012trapdoors}. \textsf{GenTrap} is an algorithm which generates a matrix $A$ with a distribution negligibly close to the distribution of $A$ in Definition \ref{def:regev-pke}, and in addition generates a \emph{trapdoor} $R$ for $A$ that allows the algorithm \textsf{Invert} to (given $R$) efficiently recover $(\vec{s}, e)$ from any $A\vec{s} + e$.

\begin{definition}[Perfectly binding, computationally hiding, extractable commitment scheme based on LWE]

\label{def:commitment}

Let $m, n$ be security parameters ($m \geq n$), and let $q = \mathrm{poly}(n)$ be a prime.

\begin{itemize}
\item $\mathsf{gen}(1^n, 1^m, q)$ outputs $(A, R) \leftarrow \textsf{GenTrap}(1^n, 1^m, q)$ as $(pk, sk)$.
The distribution of $A$ is negligibly close to uniform over $\mathbb{Z}_q^{m \times n}$.

\item $\textsf{commit}(pk, b, s)$ uses $s$ to generate
\begin{enumerate}
\item a vector $\vec{s}(s) \in \mathbb{Z}_q^n$, whose distribution over uniformly random $s$ is uniformly random,
\item a short vector $e(s)$, whose distribution over uniformly random $s$ is equal to the appropriate distribution for error vectors, and 
\item a string $w(s) \in \{0,1\}^m$, whose distribution over uniformly random $s$ is uniformly random. 
\end{enumerate}
It then outputs $z = (z_0, z_1, z_2, z_3) = (A, A\vec{s} + e, w^T A, w^T(A\vec{s} + e) + b \cdot \lfloor \frac{q}{2} \rfloor)$.
\item $\textsf{reveal}(pk, z, b, s)$ outputs the inputs it is given.
\item $\textsf{verify}(pk, z, b, s)$ checks that $e(s)$ is sufficiently short, checks that $z_1 = z_0 \vec{s} + e(s)$, and then decrypts $(z_2, z_3)$ using $w(s)$ and $\vec{s}(s)$. It outputs 1 if and only if the result is $b$.
\end{itemize}

\end{definition}

In our setting, we are guaranteed that $A$ is generated honestly (i.e.\ that it comes from running \textsf{GenTrap}), because $A$ is embedded into a CRS. Given that this is so, the commitment scheme we have defined in Definition \ref{def:commitment} is perfectly binding because, for an $A$ originating from $\mathsf{GenTrap}$ and a sufficiently short $e$, the vector $A\vec{s} + e$ perfectly determines $\vec{s}$. Then, if $\vec{s}$ is determined, $w^T(A\vec{s} + e) + b \cdot \lfloor \frac{q}{2} \rfloor$ cannot decrypt to a value other than $b$, by correctness of the \cite{regev2009lattices} encryption scheme. That the commitment scheme is computationally hiding follows directly from the IND-CPA security of the \cite{regev2009lattices} scheme.

Finally, the commitment scheme is extractable because the trapdoor $R$ allows an algorithm to efficiently invert $A$ (as described in \cite[Theorem 12]{micciancio2012trapdoors}). Specifically, $\textsf{recover}(pk = A, sk = R, z)$ runs $\textsf{Invert}(R, A, A\vec{s} + e)$ to recover $(\vec{s}, e)$, and then recovers $b$ by decrypting $w^T(A\vec{s} + e) + b \cdot \lfloor \frac{q}{2} \rfloor$.

\begin{remark}
It is possible to strengthen the commitment scheme in Definition \ref{def:commitment} so that it satisfies the following two properties:
\begin{enumerate}
\item The perfectly binding property holds for the randomness as well as for the committed bit. In other words, if $\textsf{commit}(pk, b, s) = \textsf{commit}(pk, b', s')$, then we can conclude not only that $b = b'$ but also that $(b, s) = (b', s')$.
\item The algorithm $\textsf{recover}(pk, sk, z)$ recovers the randomness $s$ as well as the committed bit $b$ from a commitment $z$.
\end{enumerate}

In order to achieve these two properties, we simply have \textsf{gen} sample a second matrix $A'$ from \textsf{GenTrap} during key generation, along with its trapdoor $R'$, and set $pk = (A, A'), sk = (R, R')$. Then we have \textsf{commit} append the value $z_4 = A'w + e'$ to the commitment $z$ that it generates, for some appropriately sampled short vector $e'$, where $w$ is the uniformly random string in $\{0,1\}^m$ that is used to compute $z_2$ and $z_3$. We also modify \textsf{verify} appropriately, so that, in addition to the checks it performs in Definition \ref{def:commitment}, it also checks that $z_4 = A'w + e'$ and that $(z_2, z_3) = (w^T A, w^T(A\vec{s} + e) + b \cdot \lfloor \frac{q}{2} \rfloor)$. The purpose of these changes is to ensure that the value of $w$ which is used to produce a given commitment $z$ is also perfectly determined by the value of $z$ (which was not necessarily the case before, as $A$ is in general not invertible). Because $A'$ is sampled independently from $A$, the commitment scheme remains computationally hiding if LWE is computationally intractable.
\end{remark}

\subsection{Argument systems}

\subsubsection{Interactive quantum machines}
\label{sec: interactive machines}

The definitions of \textit{interactive quantum machines}, their \textit{executions} and \textit{oracle access} to an interactive quantum machine are taken largely from \cite{unruh2012quantum}.

\paragraph{Interactive quantum machines} An \textit{interactive quantum machine} is a machine $M$ with two quantum registers: a register $\textsf{S}$ for its internal state, and a register $\textsf{N}$ for sending and receiving messages (the network register). Upon activation, $M$ expects in $\textsf{N}$ a message, and in $\textsf{S}$ the state at the end of the previous activation. At the end of the current activation, $\textsf{N}$ contains the outgoing message of $M$, and $\textsf{S}$ contains the new internal state of $M$. A machine $M$ gets as input: a security parameter $\mu \in \mathbb{N}$, a classical input $x \in \{0,1\}^*$, and quantum input $\ket{\Phi}$, which is stored in $\textsf{S}$. Formally, machine $M$ is specified by a family of circuits $\{M_{\mu x}\}_{\mu \in \mathbb{N}, x \in \{0,1\}^*}$, and a family of integers $\{r_{\mu x}\}_{\mu \in \mathbb{N}, x \in \{0,1\}^*}$. $M_{\mu x}$ is the quantum circuit that $M$ performs on the registers $\textsf{S}$ and $\textsf{N}$ upon invocation. $r_{\mu x}$ determines the total number of messages/invocations. We might omit writing the security parameter when it is clear from the context. We say that $M$ is \textit{quantum-polynomial-time} (QPT for short) if the circuit $M_{\mu x}$ has polynomial size in $\mu + |x|$, the description of the circuit is computable in deterministic polynomial time in $\mu + |x|$ given $\mu$ and $x$, and $r_{\mu, x}$ is polynomially bounded in $\mu$ and $x$.

\paragraph{Execution of interactive quantum machines}
For a pair of interactive quantum machines $M$ and $M'$ with internal registers $\textsf{S},\textsf{S}'$ respectively and network register $\textsf{N}$, a security parameter $\mu$, a pair of quantum states $\ket{\Phi}, \ket{\Phi'}$ and a pair of strings $x,x'$, we define the \textit{execution} $\left (M(x, \ket{\Phi}) , M'(x',\ket{\Phi'}) \right )$ via the following process:
initialize registers $\textsf{S}, \textsf{S}', \textsf{N}$ as $\ket{\Phi}, \ket{\Phi'}$ and $\ket{0}$; alternately apply $M_{\mu x}$ to registers $\textsf{S}$ and $\textsf{N}$ and $M_{\mu x'}$ to registers $\textsf{S}'$ and $\textsf{N}$; stop applying $M_{\mu x}$ after $r_{\mu x}$ times and $M_{\mu x'}$ after $r'_{\mu x}$ times, and then measure the $\textsf{S}$ register in the computational basis. Let $\left < M(x, \ket{\Phi}) , M'(x',\ket{\Phi'}) \right >$ be the random variable for the outcome of this measurement.
For a $k$-tuple of interactive machines $M_1, \ldots, M_k$, we assume that there is a message register $\textsf{N}_{ij}$ between any two machines. Machine $M_i$ has an internal register $\textsf{S}_i$, an invocation register $\textsf{INV}_i$, and is specified by a unitary $M_{i, \mu x}$, which acts on $\textsf{S}_i$, on all invocation registers $\textsf{INV}_j$, and on all of $M_i$'s message registers ($\textsf{N}_{ij}$ for all $j$). $M_{i, \mu x}$ has the additional property that it always leaves all invocation registers in the state $\ket{0}$, except for one in the state $\ket{1}$. This determines the next machine to be invoked.

\paragraph{Oracle access to an interactive quantum machine} We say that a quantum algorithm $A$ has oracle access to an interactive quantum machine $M$ (and we write this as $A^M$, or sometimes $A^{\ket{M}}$ to emphasize that $M$ is a quantum machine and that oracle access includes the ability to apply the inverse of $M$) to mean the following. Besides the security parameter and its own classical input $x$, we allow $A$ to execute the quantum circuit $M_{\mu x}$ specifying $M$, and its inverse (recall that these act on the internal register $\textsf{S}$ and on the network register $\textsf{N}$ of $M$). Moreover, we allow $A$ to provide and read messages from $M$ (formally, we allow $A$ to act freely on the network register $\textsf{N}$). We do not allow $A$ to act on the internal register $\textsf{S}$ of $M$, except via $M_{\mu x}$ or its inverse.

\subsubsection{Argument systems with setup}
\label{sec: arguments}

First we define the kinds of relations that underlie our argument systems. 
Classically, a relation over finite sets $\cX\times \cY$ is a subset $R\subseteq \cX\times \cY$. An $\NP$ relation $R = \{(x,w): V_{|x|}(x,w)=1\}$ has the additional property that given any $x\in \cX$ and $w\in\cY$, the claim that $(x,w)\in R$ can be verified by a uniformly generated family of circuits $V = \{V_n\}$ (the ``verifier''). 

In the quantum case the ``input'' $x$ (the first argument to the relation) remains classical, but the ``witness'' $w$ (the second argument) can be a quantum state $\ket{\psi}$. Before we give our definition of a { $\QMA$ relation} we introduce some notation. Fix a uniformly generated family of  polynomial-size quantum circuits $Q=\{Q_n\}_{n\in \N}$ such that for every $n$, $Q_n$ takes as input a string $x\in\{0,1\}^n$ and a quantum state $\sigma$ on $p(n)$ qubits (for some polynomial $p(n)$) and returns a single bit as output. For any $0\leq \gamma \leq 1$ define
\[R_{Q,\gamma} = \bigcup_{n\in \N} \big\{(x,\sigma)\in \{0,1\}^n\times \Density(\C^{p(n)})\,\big|\;  \Pr(Q_{n}(x, \sigma)=1)\geq \gamma\big\}\;
\]
and
\[N_{Q,\gamma} = \bigcup_{n\in \N} \big\{x\in \{0,1\}^n\,\big|\; \forall \sigma \in  \Density(\C^{p(n)})\,,\;  \Pr(Q_{n}(x, \sigma)=1)< \gamma\big\}\;.\]
Note the presence of the parameter $\gamma$, that quantifies the expected success probability for the verifier; $\gamma$ can be thought of as a measure of the ``quality'' of a witness $\ket{\psi}$ (or mixture theoreof, as represented by the density matrix $\sigma$) that is sufficient for the witness to be acceptable with respect to the relation $R$. 


\begin{definition}[$\QMA$ relation]
  \label{def:qma-relation}
	A \emph{$\QMA$ relation} is specified by triple $(Q,\alpha,\beta)$ where $Q = \{Q_n\}_{n\in\N}$ is a uniformly generated family of quantum circuits such that for every $n$, $Q_n$ takes as input a string $x\in\{0,1\}^n$ and a quantum state $\ket{\psi}$ on $p(n)$ qubits and returns a single bit, and $\alpha,\beta:\N\to[0,1]$ are such that $ \alpha(n)- \beta(n) \geq 1/p(n)$ for some polynomial $p$ and all $n\in\N$. The $\QMA$ relation associated with $(Q,\alpha,\beta)$ is the pair of sets $R_{Q,\alpha}$ and $N_{Q,\beta}$. 

  We say that a \emph{language $L=(L_{yes},L_{no})$ is specified by a $\QMA$ relation $(Q,\alpha,\beta)$} if  
	\begin{equation}\label{eq:l-qma}
	L_{yes} \subseteq \bigcup_{n\in \N} \big\{x \in \{0,1\}^n | \exists \sigma \in \Density(\C^{p(n)}) \text{ s.t. } (x,\sigma)\in R_{Q,\alpha}\big\}\;,
	\end{equation}
	and $L_{no} \subseteq N_{Q,\beta}$.
  \end{definition}

Note that in contrast to an $\NP$ relation, we define a $\QMA$ relation using two sets: the first set, $R_{Q,\alpha}$, is the set of (instance, witness) pairs that are deemed to form part of the relation. The second set, $N_{Q,\beta}$, is the set of instances that are deemed to be such that they are in relation to no witness. Some instances may lie in neither (the projection of) $R_{Q,\alpha}$ or $N_{Q,\beta}$; this is analogous to the necessity for a ``promise'' between the completeness and soundness parameters $\alpha$ and $\beta$ in the definition of the class $\QMA$, that do not appear in the definition of $\NP$. In particular, note that, whenever $\alpha - \beta > 1/\poly(n)$, a language $L$ that is specified by $(Q,\alpha,\beta)$ lies in $\QMA$. Conversely, any language in $\QMA$ is specified by some $\QMA$ relation (of course such relation is not unique).


\begin{definition}[protocol with setup]
\label{def: protocol with setup}
A \emph{protocol with setup} is  a triple of interactive machines $(S,P,V)$ with the following properties:
\begin{enumerate}
\item $S = \{S_{\mu n}\}_{\mu \in \N}$ depends on the security parameter $\mu$ and an instance size $n$, takes no input and returns a classical output in the message registers $\textsf{N}_{SP}$ and $\textsf{N}_{SV}$. When the output in both registers is the same, we refer to it as ``common reference string''.
\item Each of $P$ and $V$ has two phases: $P=(P_1,P_2)$ and $V=(V_1,V_2)$. $P_1 = \{P_{1,\mu n}\}$ and $V_1 = \{V_{1,\mu n}\}$ are interactive machines that depend on the security parameter $\mu$ and an instance size parameter $n$, take a classical message input in register $\textsf{N}_{SP}$ and $\textsf{N}_{SV}$ respectively and return a quantum message as output in registers $\textsf{N}_{P_1P_2}$ and $\textsf{N}_{V_1V_2}$ respectively.  $P_2 = \{P_{2,\mu n}\}$ and $V_2 = \{V_{2,\mu n}\}$ are interactive machines that depend on the security parameter $\mu$ and an input size $n$. $V_2$ takes as input the output of $V_1$, in register $\textsf{N}_{V_1V_2}$, as well as an instance $x$ such that $|x|=n$. $P_2$ takes as input the output of $P_1$, in register $\textsf{N}_{P_1P_2}$, an instance $x$ such that $|x|=n$, and a quantum state $\rho$. $V_2$ returns a single bit $b\in\{0,1\}$ as output, and $P_2$ returns no output. If $b=1$ then we say that $V$ accepts, and otherwise we say that it rejects. 
\end{enumerate}
\end{definition}
We refer to the first phase of $P$ and $V$ as the \textit{preprocessing phase}, and to the second phase as the \textit{instance-dependent phase}.

\begin{definition}[argument system with completeness $c$ and soundness $s$]
\label{def: arg qma-rel}
Let  $(Q,\alpha,\beta)$ be a $\QMA$ relation and $s,c:\N\to [0,1]$. An \emph{argument system} (with setup) for $(Q,\alpha,\beta)$, with completeness $c$ and soundness $s$, is a protocol with setup $(S,P,V)$ such that $S,P,V$ are quantum polynomial-time and, in addition, the following hold: 
\begin{enumerate}
\item (Completeness) For all $(x,\rho) \in R_{Q,\alpha}$, for all integer $\mu$, the execution $( S , P(x,\rho), V(x))$ returns $1$ with probability at least $c(\mu)$.
\item (Soundness) For all $x\in N_{Q,\beta}$, all integer $\mu$ and all polynomial-time $P^*$ the execution $( S , P^*(x), V(x))$ returns $1$ with probability at most $s(\mu)$.
\end{enumerate}
\end{definition}

When the second phase of a protocol with setup $(S,P,V)$ consists of a single message from $P$ to $V$ we refer to it as a \textit{non-interactive} protocol with setup. If it is a an argument system with setup, we refer to it as a \textit{non-interactive} argument system with setup. When the first phase involves some communication between $P$ and $V$, we specify that it is a non-interactive argument system with setup \textit{and preprocessing}. When $S$ outputs a common reference string (as defined in \ref{def: protocol with setup}), we refer to it as an argument system \textit{with CRS setup} (possibly with preprocessing).

Note that Definition \ref{def: arg qma-rel} requires that the execution $( S , P(x,\rho), V(x))$ returns $1$ with probability at least $c(\mu)$. In the case of sequential or parallel repetition of a protocol, it may not be possible for the prover to succeed with a single copy of the witness $\rho$ as input. In this case we may considering relaxing the definition as follows. 

\begin{definition}[Completeness of argument system with setup --- alternative definition]\label{def:compl-alt}
There exists a polynomial $q>0$, such that for all $(x,\rho) \in R_{Q,\alpha}$, for all integers $\mu$, the execution $( S , P(x,\rho^{\otimes q(\mu)}), V(x))$ returns $1$ with probability at least $c(\mu)$.
\end{definition}

We will clarify, whenever we refer to an argument system with setup, which definition we refer to. 

Finally, we define the notion of \emph{adaptive soundness}, which captures security against adversaries that are allowed to choose the common instance $x$ \emph{after} having carried out the preprocessing phase.

\begin{definition}[Adaptive soundness]\label{def:adaptive-soundness}
An argument with setup $(S,P,V)$ for a $\QMA$ relation $(Q,\alpha,\beta)$ has \emph{adaptive soundness} $s(\mu)$ if for every QPT algorithm $P^* = \{(P^*_{1,\mu n}, P^*_{2,\mu n})\}$, for all $\mu$,
\[ \Pr_{\substack{(\sigma_{PV}) \leftarrow (S_{\mu n}, P_{1,\mu n}^*,V_{1, \mu n}), \\ (x, \tau) \leftarrow P_{2,\mu n}^*(\sigma_P)}} \big( x \in N_{Q, \beta} \wedge (P^*_{2, \mu n} (x, \tau),V_{2, \mu n} (x, \sigma_V))=1\big)\leq s(\mu)\;.\]
\end{definition}

The terminology that follows Definition \ref{def: arg qma-rel} is modified in the natural way in the case of adaptive soundness.

\subsection{Zero-knowledge arguments}

\subsubsection{Zero-knowledge}
\label{sec: zk}

The notion of Zero-Knowledge in the presence of quantum verifiers and quantum auxiliary information was first formulated by Watrous~\cite{watrous2009zero}. Here, we adapt Unruh's version of this definition. We first give a definition for argument systems with setup for $\NP$ relations, and then we give the generalization to $\QMA$ relations.

\begin{definition}[computational zero-knowledge argument system for NP relation]
\label{def: zk np}
An interactive argument system with setup $(S,P,V)$ for an $\NP$ relation $R$ is quantum computational zero-knowledge if for every quantum polynomial-time verifier $V^*$ there is a quantum polynomial-time simulator $Sim$ such that, for any quantum polynomial-time distinguisher $D$ and any polynomial $l>0$, there is a negligible $\nu$ such that for any $(x, w) \in R$ with $|x|, |w| \leq l(\mu)$, and for any quantum state $\ket{\Psi}$, we have:
\begin{align*}
\Pr\big(&b = 1 : \textsf{ZE} \leftarrow \ket{\Psi}, (S, P(x,w),V^*(\textsf{Z})), b \leftarrow D(\textsf{Z,E})\big) \\
 & - \Pr\big(b = 1 : \textsf{ZE} \leftarrow \ket{\Psi}, S(x,\textsf{Z}), b \leftarrow D(\textsf{Z,E})\big)  \leq \nu(\mu)\;.
\end{align*}
Here $\textsf{ZE} \leftarrow \ket{\Psi}$ denotes that the quantum registers \textsf{Z,E} are initialized jointly in state $\ket{\Psi}$. And $(S, P(x,\sigma),V^*(\textsf{Z}))$ denotes an execution where $V^*$ gets access to the quantum register \textsf{Z}. Note that after that execution $V^*$ may have changed the state of \textsf{Z}. $S(x,\textsf{Z})$ also gets access to and may change \textsf{Z}.
\end{definition}

The only difference in the definition for $\QMA$ relations is that the prover receives a state $\sigma$ such that $(x, \sigma) \in R_{Q,\alpha}$. We report the full definition for completeness.
\begin{definition}[computational zero-knowledge argument system for QMA relation]
\label{def: zk qma}
An interactive argument system with setup $(S,P,V)$ for a $\QMA$ relation $(Q,\alpha, \beta)$ is quantum computational zero-knowledge if for every quantum polynomial-time verifier $V^*$ there is a quantum polynomial-time simulator $Sim$ such that, for any quantum polynomial-time distinguisher $D$ and any polynomial $l>0$, there is a negligible $\nu$ such that for any $(x, \sigma) \in R_{Q,\alpha}$ with $|x| \leq l(\mu)$, and for any quantum state $\ket{\Psi}$, 
\begin{align*}
\Pr\big(&b = 1 : \textsf{ZE} \leftarrow \ket{\Psi}, (S, P(x,\sigma),V^*(\textsf{Z})), b \leftarrow D(\textsf{Z,E})\big) \\
 & - \Pr\big(b = 1 : \textsf{ZE} \leftarrow \ket{\Psi}, S(x,\textsf{Z}), b \leftarrow D(\textsf{Z,E})\big) \leq \nu(\mu)\;.
\end{align*}
\end{definition}

\begin{definition}[adaptive computational zero-knowledge]

In the non-interactive setting, we may also wish to consider the notion of \emph{adaptive zero-knowledge}. We refer the reader to Definition \ref{def:adaptive-zk} and to Remark \ref{remark:adaptive-zk} for the definition of adaptive zero-knowledge. NB: Definition \ref{def:adaptive-zk} uses notation from Definition \ref{def:nizk-argument-for-np}.

\end{definition}

\subsubsection{Non-interactive zero-knowledge arguments for QMA}
\label{sec:nizk-qma}

In this section we formally define the notion of NIZK arguments for QMA. This is the main object we aim to construct in this paper. 

\begin{definition}[NIZK argument system for QMA relation]
\label{def: nizk-for-qma-relation}
An argument system with setup $\Pi$ for a $\QMA$ relation $(Q,\alpha, \beta)$ is said to be an (adaptively sound) NIZK argument system with CRS setup for $(Q,\alpha, \beta)$ if it is a non-interactive protocol with CRS setup, and in addition it satisfies the following:
\begin{enumerate}
\item $\Pi$ is quantum computational zero-knowledge, as in Definition \ref{def: zk qma}.
\item An execution of $\Pi$ involves a single message from $P$ to $V$.
\item $\Pi$ has $1-\negl$ completeness and $\negl$ (adaptive) soundness.
\end{enumerate}
If $P$ and $V$ exchange additional messages in the preprocessing phase (but not in the instance-dependent phase), we refer to $\Pi$ as an NIZK argument system with CRS setup and preprocessing.
\end{definition}

\begin{definition}[Argument system for a language in QMA]
\label{def: nizk-for-qma-language}
We say that a language $L \in \text{QMA}$ has a (adaptively sound) NIZK argument system with CRS setup (and preprocessing) if there exists a $\text{QMA}$ relation $(Q,1-\negl(n),\negl(n))$ which specifies $L$ and such that there exists an (adaptively sound) NIZK argument system with setup (and preprocessing) $(S,P,V)$ for $(Q,1-\negl(n),\negl(n))$.
\end{definition}

Using the amplification technique from~\cite{marriott2005quantum} it follows that any $\QMA$ relation that specifies a $\QMA$ language $L$ can be amplified to a $\QMA$ relation of the form $(Q,1-\negl(n),\negl(n))$, such that moreover the reduction preserves the honest quantum witness for positive instances of $L$.

\subsubsection{Non-interactive zero-knowledge for NP}
\label{sec:nizk-np}

A building block in our construction of NIZK argument systems for $\QMA$ are non-interactive argument systems for $\NP$ in the CRS model (or with CRS setup). Based on the work of~\cite{canetti2019fiat,peikert2019noninteractive} it is possible to construct such argument systems satisfying both adaptive soundness and adaptive zero-knowledge assuming only the hardness of the LWE problem. Here, we require such arguments that maintain the soundness and zero-knowledge properties in case the adversarial party may be quantum. We give the definitions and sketch how quantum security follows almost immediately by adapting the same arguments from~\cite{canetti2019fiat} and using quantum-security of the LWE-based correlation intractable hash family from~\cite{peikert2019noninteractive}. 

For an $\NP$ relation $R$ we let $\mathcal{L}(R) = \{x: \exists w, (x,w)\in R\}$. 

\begin{definition}[NIZK argument system for NP relation]
\label{def:nizk-argument-for-np}
An NIZK argument system for an $\NP$ relation $R$ is a non-interactive protocol with setup $\Pi=(S,P,V)$ such that in addition:
\begin{enumerate}
\item $S = \{S_{\mu n}\}$ is a classical polynomial-time family of circuits that depend on the security parameter $\mu$ and an instance size parameter $n$ and return a common reference string $\crs$. 
\item $P = \{P_{\mu n}\}$ is a classical polynomial-time family of circuits that take as input $\crs$, an instance $x$ and a witness $w$ and return a proof $\pi$.
\item $V = \{V_{\mu n}\}$ is a classical polynomial-time family of circuits that take as input $\crs$, an instance $x$ and a proof $\pi$ and return a decision $b\in\{0,1\}$. If $b=1$ then we say that $V$ accepts, and otherwise we say that it rejects. 
\item $P$ and $V$ do not require a setup phase (other than receiving the common reference string $\crs$ from $S$).
\item $\Pi$ is quantum computational zero-knowledge, as in Definition \ref{def: zk np}.
\end{enumerate}
The argument system should satisfy the completeness requirement: for every integer $\mu,n$, for every $(x,w)\in R$ such that $|x|=n$, $V_{\mu n}(\crs,x,\pi)=1$ with probability $1$ when $\crs\leftarrow S_{\mu n}()$ and $\pi\leftarrow P_{\mu n}(\crs,x,w)$. 
\end{definition}

We define adaptive soundness and adaptive ZK for the case of quantum adversaries. 

\begin{definition}[Adaptive soundness]\label{def:adaptive-soundness}
A NIZK argument $\Pi$ for $R$ has \emph{adaptive soundness} $s$ if for every quantum polynomial time  $P^* = \{P^*_{\mu n}\}$ and all $\mu$,
\[ \Pr_{\crs\leftarrow S_{\mu n}(),\;(x,\pi)\leftarrow P_{\mu n}^*(\crs)}\big( x\notin \mathcal{L}(R) \wedge V(\crs,x,\pi)=1\big)\leq s(\mu)\;.\]
\end{definition}

\begin{definition}[Adaptive zero-knowledge]\label{def:adaptive-zk}
Let $\Pi = (S, P, V)$ be an NIZK argument system for the relation $R$. Let $V^* = \{V^*_{\mu n}\}$ be a quantum polynomial time (possibly dishonest) verifier for $\Pi$ which takes quantum auxiliary input $\sigma$. Let $Sim = \{Sim_{\mu n}\}$ be a zero-knowledge simulator for $V^*$ which consists of two stages, $Sim_1$ and $Sim_2$, the first of which (roughly speaking) takes the place of $S$ in $\Pi$, and the second of which takes the place of $P$. For notational convenience, we define distributions $D_1$ and $D_2$ through the games \textbf{Game 1} and \textbf{Game 2}:

\textbf{Game 1 (honest prover)}
\begin{enumerate}
\item $S_{\mu n}$, the setup phase of $\Pi$, is run in order to generate a CRS $\crs$.
\item $V^*_{\mu n}$ is given access to $\crs$ and to $\sigma$, and produces an instance $x$ and a witness $w$ such that $(x, w) \in R$, in addition to quantum auxiliary output $\tau$ (which may be correlated in an arbitrary way with $(x, w)$).
\item $P_{\mu n}$, the honest prover, is run on $(\crs, x, w)$.
\end{enumerate}

$D_1$ is defined to be the distribution of the following quantities, produced through an execution of \textbf{Game 1}:
\begin{equation*}
(\crs, P_{\mu n}(\crs, x, w), \tau)
\end{equation*}

\textbf{Game 2 (simulator)}
\begin{enumerate}
\item $Sim_{1, \mu n}$, the setup phase of $Sim$, is run in order to generate a CRS $\crs$.
\item $V^*_{\mu n}$ is given access to $\crs$ and to $\sigma$, and produces an instance $x$ and a witness $w$ such that $(x, w) \in R$, in addition to quantum auxiliary output $\tau$ (which may be correlated in an arbitrary way with $(x, w)$).
\item $Sim_{2, \mu n}$, the second stage of the simulator, is run on $(\crs, x)$. (Note that $Sim_2$ is \emph{not} given access to $w$.)
\end{enumerate}

$D_2$ is defined to be the distribution of
\begin{equation*}
\big(\crs, Sim_{2, \mu n} (\crs, x), \tau \big)
\end{equation*}

produced through an execution of \textbf{Game 2}.

An NIZK argument $\Pi$ for $R$ is \emph{adaptive (computational) zero-knowledge} if there is a quantum polynomial time simulator $Sim = \{Sim_{\mu n}\}$ such that, for every quantum polynomial time  $V^* = \{V^*_{\mu n}\}$, the distribution ensembles $\{D_1\}_{\mu n}$ and $\{D_2\}_{\mu n}$ are, for all quantum auxiliary inputs $\sigma$ to $V^*_{\mu n}$, indistinguishable with respect to quantum polynomial-time distinguishers.
\end{definition}

\begin{remark}
\label{remark:adaptive-zk}
Definition \ref{def:adaptive-zk} extends in a straightforward way to statistical zero-knowledge, to QMA-relations (see Definition \ref{def:qma-relation}), and to argument systems with preprocessing as well as CRS setup. Specifically, in order to get statistical zero-knowledge, one would require that $\{D_1\}_{\mu n}$ and $\{D_2\}_{\mu n}$ be statistically instead of computationally indistinguishable; in order to extend the definition to QMA-relations, one would simply allow the relation $R$ to be a QMA-relation; and, in order to accommodate argument systems with preprocessing, one would allow $V^*$ to choose an instance $x$, together with a witness $w$ for $x$, in a way which depends on all the parts of the transcript of the preprocessing phase to which it has access. These parts of the transcript, $T_V$, would then take the place of $\crs$ in \textbf{Game 2}. \textbf{Game 1}, meanwhile, would remain a full execution of the real protocol with the honest prover. For more details and motivation, see~\cite{feige1999multiple}.
\end{remark}

\begin{lem}\label{lem:nizk-np}
For any $\NP$ relation $R$ there is a NIZK argument system with negligible adaptive soundness, which is adaptively ZK against quantum adversaries, assuming the hardness of LWE and the existence of statistically binding and quantum computationally concealing commitment schemes.

Moreover, there is a quantum polynomial-time simulator $\{Sim_{\mu n}\}$ that satisfies Definition~\ref{def:adaptive-zk} and such that $\{Sim_{\mu n}\}$ is \emph{straight-line} in the following sense: 
\begin{enumerate}
\item $Sim_{\mu n}$ first generates a CRS $\crs$ without having access to $V^*$ or $\sigma$;
\item $Sim_{\mu n}$ then executes $V^*_{\mu n}(\crs,\sigma)$ to obtain $(x,\tau)$;
\item Finally, $Sim_{\mu n}$ (classically) computes $\pi$ directly from $(\crs,x)$ and returns $(\crs,\pi,\tau)$.
\end{enumerate}
\end{lem}

\begin{proof}
Fix $\mu$ and $n$ and let $\textsf{Com}$ be a statistically binding and quantum computationally concealing commitment scheme that can be used to commit to $n\times n$ Boolean matrices. We write a key for $\textsf{Com}$ as $pk$. 

Our starting point is the $3$-message protocol for graph Hamiltonicity from~\cite{feige1999multiple}, as described in~\cite[Section 5.1]{canetti2019fiat}, and the NIZK argument $\tilde{\Pi}$ for it given in~\cite[Construction 5.2]{canetti2019fiat}. 

\paragraph{Argument system.}
We recall the protocol to establish notation, referring to~\cite[Section 5.1]{canetti2019fiat} for more details. 
In the $3$-message protocol, the prover first sends a commitment $a=\textsf{Com}.\textsf{commit}(pk, H, s)$ to a graph $H$ that is a uniformly random permutation $\pi(C_n)$, for $C_n$ the adjacency matrix of a canonically fixed cycle. The verifier sends a random challenge $e\in\{0,1\}$. If $e=0$, the prover reveals $H$. If $e=1$, the prover, given $(x,w)\in R$ where $x$ is a graph and $w$ a Hamiltonian cycle, selects a permutation $\sigma$ such that $\sigma^{-1}(C_n)$ is the cycle $w$. The prover sends $\pi\circ\sigma$ as well as decommitments to entries of $H$ that correspond to non-edges of $\pi\circ\sigma(G)$. (In the actual protocol, the above steps are repeated $t$ times in parallel. For simplicity, we describe the protocol for a single run.)

Next we recall the non-interactive variant, protocol $\tilde{\Pi}$ in~\cite[Construction 5.2]{canetti2019fiat}. Here $\crs = (pk,k)$ where $pk$ is a key for the commitment scheme and $k$ is a key for a family $\mathcal{H}$ of hash functions satisfying certain properties (correlation intractability and programmability, see~\cite{canetti2019fiat}). The prover then computes $a$ as above, sets $e=h_k(a)$, and computes the third message $z$ as above. The prover returns $\pi=(a,e,z)$. 

The only difference with~\cite[Construction 5.2]{canetti2019fiat} is that we use a statistically binding commitment scheme instead of their PKE, as this is needed for soundness.  We instantiate $\mathcal{H}$ using the correlation intractable family of hash functions from~\cite{peikert2019noninteractive}.

In~\cite[Theorem 5.5]{canetti2019fiat} it is shown that $\tilde{\Pi}$ is a NIZK argument system for $\NP$ in the common reference string model satisfying both adaptive soundness and adaptive zero-knowledge against classical adversaries. We briefly sketch how their argument extends to the quantum case. 

\paragraph{Adaptive soundness.}
First we consider adaptive soundness. Here the argument is essentially identical; we follow the proof of~\cite[Lemma 5.6]{canetti2019fiat}. Let $( \crs,x,\pi=(a,e,z))$ be such that $V(\crs,x,\pi)=1$. Then for every $i\in\{1,\ldots,t\}$ (recall that the basic $3$-message protocol is repeated $t$ times before applying Fiat-Shamir), using that the commitment is unconditionally binding it must be that $e_i=0$ if and only if $a_i$ is a commitment to a cycle. This means that any $P^*$ which breaks the adaptive soundness condition from Definition~\ref{def:adaptive-soundness}, given only $\crs$, returns $(a,e)\in R$, where is the set of pairs where for all $i$, $a_i$ is a commitment to a cycle if and only if $e_i=0$, with probability at least $\nu(\mu)$. By an averaging argument, there exists a fixed $\crs$ for which the same holds. This contradicts the correlation intractability of $\mathcal{H}$ which by~\cite{peikert2019noninteractive} reduces to the quantum hardness of LWE. 

\paragraph{Adaptive zero-knowledge.}
Next we consider adaptive zero-knowledge. For this we construct a simulator, mostly following the proof of~\cite[Proposition 7.6]{canetti2018fiat}. Let $V^*$ be quantum polynomial time. First, $Sim_{\mu n}(\sigma)$ generates a key $pk$ for $\textsf{Com}$. Then, it selects uniformly random challenges $\{e_i\}$ and permutations $\{\pi_i\}$. If $e_i=0$, it sets $a_i\leftarrow \textsf{Com}.\textsf{commit}(pk, \pi_i(C_n), s)$. If $e_i=1$, it sets $a_i\leftarrow \textsf{Com}.\textsf{commit}(pk, 0, s)$ (for some randomness $s$).

The simulator then samples a key $k^*$ such that $h_{k^*}(a_i)=e_i$ for all $i\in\{1,\ldots,t\}$ and creates $\crs = (pk,k^*)$.\footnote{The possibility for selecting $k^*$ in a way that is indistinguishable from a random key comes from the \emph{programmability} of $\mathcal{H}$; see~\cite{canetti2019fiat}.}  Then, the simulator executes $V^*(\crs,\sigma)$ to obtain $(x,w,\tau)$. Finally, if $e_i=0$ the simulator sets $z_i=\pi_i(C_n)$. If $e_i=1$ it selects a uniformly random permutation $\sigma_i$ and sets $z_i$ to be $\sigma_i$ together with decommitments to the entries of $a_i$ that correspond to non-edges of $\sigma_i(G)$.

This simulator is ``straight-line simulator'' and does not involve rewinding $V^*$. Therefore, the computational zero-knowledge property follows from the computationally hiding property of $\textsf{Com}$ and the programmability of $\mathcal{H}$ (to argue that $k^*$ is indistinguishable from a random key).  
\end{proof}

\subsection{Proofs and arguments of quantum knowledge}
\label{sec: poqk}

The content of this subsection, as it pertains to \emph{proofs of quantum knowledge}, was written in collaboration with Broadbent and Grilo, and appears with slight differences in~\cite{broadbent2019zeroknowledge}.

A \emph{Proof of Knowledge (PoK)} is an interactive proof system
for some relation $R$ such that if
the verifier accepts some input $x$ with high enough
probability, then she is ``convinced'' that the prover ``knows'' some witness $w$ such
that $(x,w) \in R$. 
This notion is formalized by requiring the existence of an efficient \emph{extractor}~$K$
that is able to return a witness for $x$ when given oracle access to
the prover (including the ability to rewind its actions, in the classical case).

\begin{definition}[Classical Proof of Knowledge]
  \label{def:pok}
  Let $R \subseteq \mathcal{X} \times \mathcal{Y}$ be a relation.  A proof
  system $(P,V)$ for $R$ is a Proof of Knowledge for $R$ with knowledge error $\kappa$ if there
  exists a polynomial $p > 0$ and a polynomial-time machine $K$
  such that for any classical interactive machine $P^*$, any $\mu \in \mathbb{N}$, any polynomial $l>0$, any instance $x \in \{0,1\}^{n}$ for $n=\poly(\mu)$ and any string $y$: 
  if the execution $(P^*(x, y), V(x))$ returns $1$ with probability
  $\eps > \kappa(\mu)$, we have
  \[\Pr\left(\left(x, K^{P^*(x,y)}(x)\right) \in R \right) \geq p\left(\eps - \kappa(\mu),
  \frac{1}{\mu}\right). \]
\end{definition}
In this definition, $y$ corresponds to the side information that $P^*$ has,
possibly including some $w$ such that $(x,w) \in R$.

\medskip

PoKs were originally defined only considering classical adversaries, and 
this notion was first studied in the quantum setting by Unruh~\cite{unruh2012quantum}.
The first issue that arises in the quantum setting is to formalize the type of query that the extractor $K$ is able
to make. In order to do so, we
assume that $P^*$ always performs a fixed unitary operation $U$ when invoked. Notice that this can
be assumed without loss of generality since $(i)$ we can always consider a
purification of $P^*$, $(ii)$ all measurements can be
performed coherently, and $(iii)$ $P^*$ can keep track of the round of communication
in some internal register and $U$ can implicitly control on this value.
Then, the quantum extractor $K$ has oracle access to $P^*$ in the sense that it may
perform $U$ and $U^\dagger$ on the message register and private register
of $P^*$, but has no direct access to the latter.
We denote the extractor
$K$ with such oracle access to $P^*$ by $K^{\ket{P^*(x,\rho)}}$, where
$\rho$ is some (quantum) side information held by $P^*$. 

\begin{definition}[Quantum Proof of (Classical) Knowledge]
  \label{def:qpok}
  Let $R \subseteq \mathcal{X} \times \mathcal{Y}$ be a relation.  A proof
  system $(P,V)$ for $R$ is a Quantum Proof of Knowledge for $R$ with knowledge error $\kappa$ 
  if there
  exists a polynomial $p > 0$ and a quantum polynomial-time machine $K$
  such that for any quantum interactive machine  $P^*$, any $\mu \in \mathbb{N}$, any polynomial $l>0$, any instance $x \in \{0,1\}^{n}$ for $n=\poly(\mu)$ and any state $\rho$: 
  if the execution $(P^*(x, \rho), V(x))$ returns $1$ with probability
  $\eps > \kappa(\mu)$, we have
  \[\Pr\left(\left(x, K^{\ket{P^*(x,\rho)}}(x)\right) \in R \right) \geq p\left(\eps - \kappa(\mu),
  \frac{1}{\mu}\right). \]
\end{definition}

\begin{remark}\label{rem:repetition-extractor}
  In the fully classical case of \ref{def:pok}, the extractor could repeat the procedure in sequence
  polynomially many times in order to increase the probability of a successful extraction (which, in Definitions \ref{def:pok} and \ref{def:qpok}, is allowed to be inverse-polynomially small in the security parameter).
  This is not known to be possible for a general quantum $P^*$, since the final measurement to extract the witness could possibly
  disturb the internal state of $P^*$, making it impossible to simulate the
  side information that $P^*$ had originally in the subsequent simulations.
\end{remark}

We finally move on to the full quantum setting, where we want a {\em Proof of
Quantum Knowledge} (PoQK). Intuitively, at the end of the protocol, we would like the verifier
to be `convinced' that the prover `has' a {\em quantum witness} for the input $x$. The main difference from Quantum Proofs of (classical) Knowledge is that in the case of $\QMA$ relations, as defined in section \ref{sec: arguments}, the notion of a witness is not as unambiguous as in the case of $\NP$ relations. We introduce a parameter $q$ which quantifies the probability that the witness returned by the extractor makes the verifying circuit accept. We refer to this parameter as the ``quality'' of the PoQK. We also allow the extractor $K$ to return a special symbol ``$\perp$'' in a designated portion of the output register, and we require that either the extractor returns ``$\perp$'' or it returns a witness of a certain quality. Formally, we assume that the output of the extractor is measured according to $\{\ket{\perp}\bra{\perp}, I-\ket{\perp}\bra{\perp}\}$. We ask that the outcome of this measurement be the latter with at least inverse-polynomial probability, and that, conditioned on the latter outcome, the post-measurement state be a witness (of a certain quality).

\begin{definition}[Proof of Quantum Knowledge]
  \label{def:poq-single}
  Let $(Q,\alpha,\beta)$ be a \QMA{} relation.
  A proof
  system $(P,V)$ is a Proof of Quantum Knowledge for $(Q,\alpha,\beta)$ with knowledge
  error 
  $\kappa$ and quality $q>\beta$, if there exists a polynomial $p > 0$ and a quantum polynomial-time machine $K$ such that for any quantum interactive machine $P^*$, any $\mu \in \mathbb{N}$, any polynomial $l>0$, any instance $x \in \{0,1\}^{n}$ for $n=\poly(\mu)$ and any state $\rho$: 
  if the execution $(P^*(x, \rho), V(x))$ returns $1$ with probability
  $\eps > \kappa(\mu)$, we have
 \[\Pr\left( K^{\ket{P^*(x,\rho)}}(x) \neq \text{``$\perp$''} \text{ and }\left(x, \frac{(I-\ket{\perp}\bra{\perp})K^{\ket{P^*(x,\rho)}}(x)}{\Tr[(I-\ket{\perp}\bra{\perp})K^{\ket{P^*(x,\rho)}}(x)]}\right) \in R_{Q,q(|x|,\eps)}\right)  \geq
  p\left(\eps-\kappa(\mu), \frac{1}{\mu}\right).
  \]
\end{definition}

We also define \textit{arguments} of quantum knowledge (with a setup). The main difference is that the proof system is replaced by an argument system with setup. Moreover, the extractor is allowed to create the setup as they wish (they can ``impersonate'' the setup procedure $S$).

\begin{definition}[Quantum Argument of (Classical) Knowledge]
  \label{def:qaok}
  Let $R \subseteq \mathcal{X} \times \mathcal{Y}$ be a relation.  An argument
  system with setup $\Pi=(S,P,V)$ for $R$ is a Quantum Argument of Knowledge with setup for $R$ with knowledge error $\kappa$ 
  if there
  exists a polynomial $p > 0$ and a quantum polynomial-time machine $K$
  such that for any quantum polynomial-time interactive machine  $P^*$, any $\mu \in \mathbb{N}$, any polynomial $l>0$, any instance $x \in \{0,1\}^{n}$ for $n=\poly(\mu)$ and any state $\rho$: 
  if the execution $(S, P^*(x, \rho), V(x))$ returns $1$ with probability
  $\eps > \kappa(\mu)$, we have
  \[\Pr\left(\left(x, K^{\ket{P^*(x,\rho)}}(x)\right) \in R \right) \geq p\left(\eps - \kappa(\mu),
  \frac{1}{\mu}\right)\;. \]
\end{definition}

\begin{definition}[Argument of Quantum Knowledge]
  \label{def:aoq-single}
  Let $(Q,\alpha,\beta)$ be a \QMA{} relation.
  An argument 
  system with setup $\Pi=(S,P,V)$ is an Argument of Quantum Knowledge with setup for $(Q,\alpha,\beta)$ with knowledge
  error 
  $\kappa$ and quality $q>\beta$ if there exists a polynomial $p>0$ and a quantum polynomial-time interactive machine $K$ such that for any quantum polynomial-time interactive machine $P^*$, any $\mu \in \mathbb{N}$, any polynomial $l>0$, any instance $x \in \{0,1\}^{n}$ for $n=\poly(\mu)$ and any state $\rho$: 
  if the execution $(S, P^*(x, \rho), V(x))$ returns $1$ with probability
  $\eps > \kappa(\mu)$, we have
\[\Pr\left( K^{\ket{P^*(x,\rho)}}(x) \neq \text{``$\perp$''} \text{ and }\left(x, \frac{(I-\ket{\perp}\bra{\perp})K^{\ket{P^*(x,\rho)}}(x)}{\Tr[(I-\ket{\perp}\bra{\perp})K^{\ket{P^*(x,\rho)}}(x)]}\right) \in R_{Q,q(|x|,\eps)}\right)  \geq
  p\left(\eps-\kappa(\mu), \frac{1}{\mu}\right)\;.
  \]
 \end{definition}

As for the several possible specializations to the definition of Argument of Quantum Knowledge with setup based on the properties of the underlying argument system (NIZK, CRS setup, preprocessing etc.), we naturally apply the terminology introduced in sections \ref{sec: arguments}, \ref{sec: zk} and \ref{sec:nizk-qma}.

\subsubsection{Reducing the knowledge error sequentially}
\label{sec:seq-rep}

One of the most natural properties of Proofs of Knowledge that one investigates in the classical setting is reducing the knowledge error by sequential repetition. Classically, it is well-known that the knowledge error drops exponentially fast in the number of sequential repetitions \cite{bellare1992defining}. 
Just like in the classical case, sequential repetition of a proof of quantum knowledge reduces the knowledge error exponentially fast. This is an immediate consequence of the proof of a lemma from Unruh \cite{unruh2012quantum} for the case of quantum Proofs of (classical) Knowledge.

\begin{lem}
\label{lem: seq-rep}
Let $n = n(\mu)$ be a polynomially bounded and efficiently computable function. Let $(P, V)$ be a Proof of Quantum Knowledge for a $\QMA$ relation $(Q,\alpha,\beta)$ with knowledge error $\kappa$. Let $(P', V')$ be the proof system consisting of $n$ sequential executions of $(P,V)$ (where $V'$ accepts iff all executions accept). Then $(P',V')$ is a Proof of Quantum Knowledge for $(Q,\alpha,\beta)$ with knowledge error $\kappa^n$.
\end{lem}

\begin{proof}
Unruh's argument applies to interactive machines with quantum auxiliary input, which is our setting, and does not make use at all of the fact that the extractor's output is a classical string (the argument proceeds entirely at the level of the success probability of the extractor for the sequentially composed proof and for the atomic proof). 
\end{proof}

The only difference from \cite{unruh2012quantum} in the Quantum Knowledge case is that if the original PoQK had the additional property of existence of an efficient prover which is accepted with probability close to $1$ when given a witness as auxiliary input, the corresponding efficient prover for the $n$-sequentially repeated protocol requires $n$ copies of a witness. Formally, this does not affect the completeness property of the proof system since this is formulated with respect to computationally unbounded provers, who can generate as many witnesses as they wish on their own.

An analogous lemma holds for arguments of quantum knowledge. However, since in this case the prover is computationally bounded, the completeness parameter will, in general, drop below soundness, unless we give the prover multiple copies of the witness, instead of a single one. In this case we say that the argument has completeness with respect to the alternative definition, Definition~\ref{def:compl-alt}. 

\section{The protocol}
\label{sec: protocol}

\subsection{Notation and predicates}

For a circuit $Q_n$, we denote by $H(Q_n)$ the local Clifford Hamiltonian obtained by performing the circuit-to-Clifford-Hamiltonian reduction from~\cite[Section 2]{broadbent2016zero}. In the rest of this section, $Q_n$ will always be taken from a family $Q = \{Q_n\}_{n \in \mathbb{N}}$, where $Q$ specifies a $\QMA$ relation $(Q,\alpha,\beta)$, and we will let the $r$-th term of the Clifford Hamiltonian $H(Q_n)$ be $C_r^*\ket{0^k}\bra{0^k} C_r$. So, 
\begin{equation}
\label{eq:clifford-hamiltonian}
H(Q_n) =  \sum_{r=1}^m C_r^*\ket{0^k}\bra{0^k} C_r\;,
\end{equation}
where each $C_r$ is a $k$-local Clifford unitary.
(Following \cite{broadbent2016zero}, we use the short-hand $\ket{0^k}\bra{0^k}$ to denote a projector which is $\ket{0}\bra{0}$ on at most $k$ qubits and identity everywhere else. As shown in~\cite{broadbent2016zero}, we can take $k=5$ without loss of generality.)

We denote by $\mathcal{H}_{\textsf{clock}} \otimes \mathcal{H}_{\textsf{instance}} \otimes \mathcal{H}_{\textsf{witness}}$ the Hilbert space that $H(Q_n)$ acts on. For notational convenience, we assume in the rest of this section that $\mathcal{H}_{\textsf{instance}}$ is $n$ qubits, that is, $\mathcal{H}_{\textsf{instance}} = \mathbb{C}^{2^n}$.

For clarity and notational convenience, we define predicates $R_r$ and $Q$ below (along the lines of Definitions \ref{def:bjsw-R_r} and \ref{def:bjsw-Q}) which we will refer to in our description of our protocol.

\begin{remark}

Predicates $Q$ and $R_r$ are defined with respect to a fixed problem instance $x$ and a fixed Clifford Hamiltonian $H$, where

\begin{equation*}
H = \sum_{r = 1}^m C_r^*\ket{0^k} \bra{0^k} C_r
\end{equation*}

for some $m$ that is polynomial in $n$.

\end{remark}

\begin{definition}[Definition of $R_r$]
\label{def:new-R_r}

As in section \ref{sec:bjsw}, we write $\mathcal{D}_N$ to represent the set of all valid (classical) $N$-bit codewords of a particular error-correcting code. We will generally refer to this error-correcting code as `the concatenated Steane code'. (This code is the same concatenated Steane code which is outlined in \cite[Appendix A.6]{broadbent2016zero}.) We may write $\mathcal{D}_N = \mathcal{D}^0_N \cup \mathcal{D}^1_N$, where $\mathcal{D}^0_N$ is the set of all codewords that encode 0, and $\mathcal{D}^1_N$ is defined analogously.

We assume that $r$ takes values in $[m+1]$, where $m$ is the number of terms in the Clifford Hamiltonian $H$. Our $R_r$ is defined differently when $r \in [m]$ and when $r = m+1$.

If $r \in [m]$: we define $R_r(t, \pi, u) = \tilde{R}_r(t, \pi, u)$, where $\tilde{R}_r$ is the predicate defined in Definition \ref{def:bjsw-R_r}.

If $r = m+1$, then we set $R_r = R_{m+1}$, where $R_{m+1}$ is defined below (Definition \ref{def:R_{m+1}}).

\end{definition}

\begin{definition}[Definition of $R_{m+1}$]

\label{def:R_{m+1}}

Let $u = u_{\textsf{clock}_1}, u_{\textsf{instance}_1}, \dots, u_{\textsf{instance}_{n}}$ be a string in  $\{0,1\}^{2N(n + 1)}$.

\begin{remark}

Each $u_{\textsf{label}}$, for $\textsf{label} \in \{\textsf{clock}_1, \textsf{instance}_1, \dots, \textsf{instance}_{n}\}$, is a $2N$-bit string, and intuitively represents the result of measuring the \emph{logical} qubit with an index specified by \textsf{label}. (For notational convenience in the exposition below, we replace the iterator \textsf{label} by the iterator $i$.) For example, $u_{\textsf{clock}_1}$ is the string that results from measuring the first \emph{logical} qubit of the clock register. The logical clock register consists of many logical qubits, and each logical qubit is encoded in $2N$ physical qubits as a result of applying the authentication code described in Figure \ref{fig: authentication code}.
\end{remark}

For $\pi \in S_{2N}$, and for each $i \in \{\textsf{clock}_1, \textsf{instance}_1, \ldots, \textsf{instance}_{n}\}$, define strings $p_i, q_i$ in $\{0,1\}^N$ such that $\pi(p_i \| q_i) = u_i$ (alternatively: $\pi^{-1}(u_i) = p_i \| q_i)$. The predicate $R_{m+1}(t, \pi, u)$ takes the value 1 if and only if the following two conditions (1. and 2.) are met:
\begin{enumerate}
\item \emph{Either} \\[-3em]
\begin{quote}
\item $p_{\textsf{clock}_1} \in \mathcal{D}_N^1$ (this corresponds to the \textit{first} qubit of the clock register, expressed in unary, being in state $1$, i.e. the clock register is \emph{not} at time $0$),
\end{quote}
\emph{or} \\[-2em]
\begin{quote}
For every $i \in \{\textsf{instance}_1, \ldots, \textsf{instance}_n \}$, $p_i \in \mathcal{D}^{x_i}_N$.
\end{quote}
\item $\braket{q_{\textsf{clock}_1}\, q_{\textsf{instance}_1} \cdots q_{\textsf{instance}_n} | t_{\textsf{clock}_1} \, t_{\textsf{instance}_1}\cdots t_{\textsf{instance}_n} }\neq 0$.
\end{enumerate}
\end{definition}

We now define our predicate $Q$ in terms of the $R_r$ defined in Definition \ref{def:new-R_r}.

\begin{definition}[Definition of $Q$]
\label{def:new-Q}
Let $d = (x_1,\ldots ,x_{2Np(n)}, y_1,\ldots, y_{2Np(n)})$ be a string in $\{0,1\}^{4Np(n)}$, for some polynomial $p(n)$ of $n$.
Define
\begin{align*}
\mathbb{P}_{m+1} = &\ket{0}\bra{0}_{\textsf{clock}_1} \otimes \big(I - \ket{x} \bra{x}\big)_{\textsf{instance}} \otimes I_{\textsf{witness}} \\
+ &(I - \ket{0}\bra{0})_{\textsf{clock}_1} \otimes I_{\textsf{instance}} \otimes I_{\textsf{witness}}
\end{align*}

where $\ket{x}\bra{x}$ is a shorthand for the projector onto the standard-basis bitstring $\braket{x}$, and
\begin{equation*}
C_{m+1} = I_{\textsf{clock}} \otimes I_{\textsf{instance}} \otimes I_{\textsf{witness}}.
\end{equation*}

For $r \in [m+1]$, define
\begin{equation*}
\mathbb{P}_r = 
\begin{cases}
	C_r^*\ket{0^k} \bra{0^k} C_r & r \in [m] \\
	\mathbb{P}_{m+1} & r = m+1\\
\end{cases}
\end{equation*}

Let $i_1,..,i_k$ be the indices of the qubits on which $\mathbb{P}_r$ acts non-trivially. Let $d' = (a',b') = (a'_{i_1}, \dots, a'_{i_k}, b'_{i_1}, \dots, b'_{i_k}) = (x_{2Ni_1 + 1}, \ldots, x_{2Ni_1 + 2N}, \ldots, x_{2Ni_k + 1}, \ldots, x_{2Ni_k + 2N}, y_{2Ni_1 + 1}, \ldots, y_{2Ni_1 + 2N}, \ldots, y_{2Ni_k + 1}, \ldots, y_{2Ni_k + 2N})$ be a string in $\{0,1\}^{4Nk}$. (The example below, wherein $k=2, N=2, i_1 = 1, i_2 = 3,$ and $d' = (a', b') = 01001000$, may clarify the notation.)

\includegraphics[width=\textwidth]{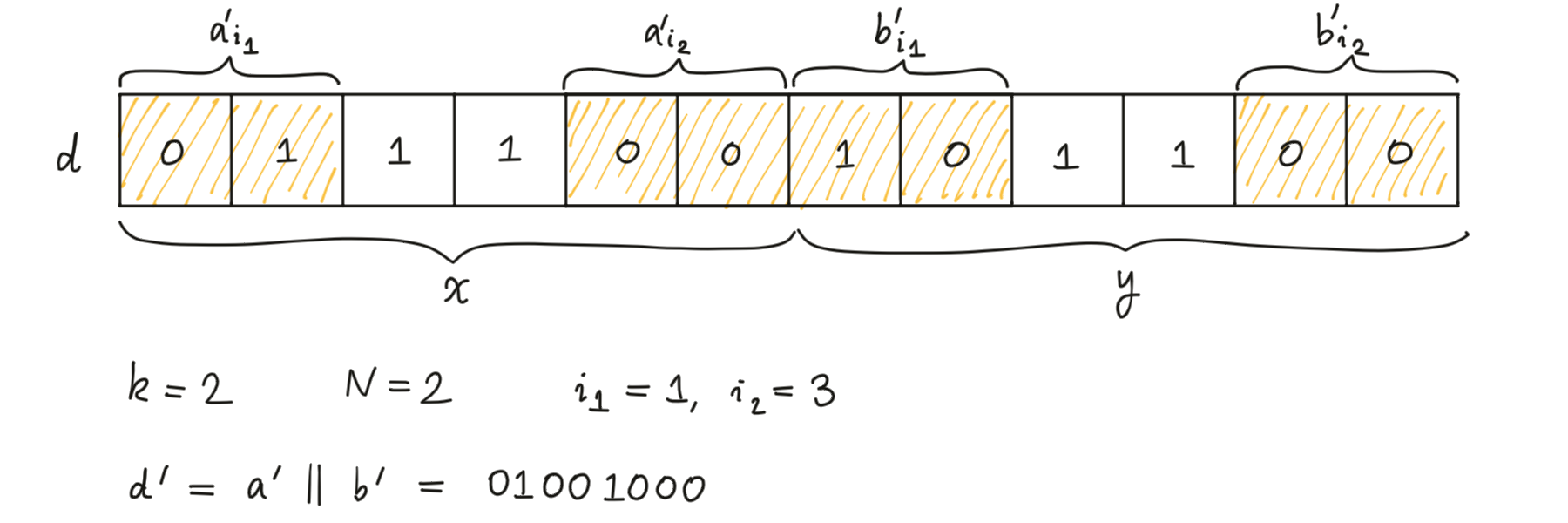}

Let $e_{i_1}, \ldots, e_{i_k}$ be the unique strings such that
\begin{gather}
C_r^{\otimes 2N}(X^{(a \oplus a')_{i_1}} Z^{(b \oplus b')_{i_1}} \otimes \cdots \otimes X^{(a \oplus a')_{i_k}} Z^{(b \oplus b')_{i_k}}) = \alpha(X^{e_{i_1}} Z^{f_{i_1}} \otimes \cdots \otimes X^{e_{i_k}} Z^{f_{i_k}}) C_r^{\otimes 2N}\label{eq:clifford-com}
\end{gather}

for some $\alpha \in \{1, i, -1, -i\}$ and some $f_{i_1}, \ldots, f_{i_k} \in \{0,1\}^{2N}$. (It is possible to efficiently compute $e = e_{i_1},\dots,e_{i_k}$ and $f = f_{i_1},\dots,f_{i_k}$ given $a, b$ and $C_r$.)

Predicate $Q$ is defined as follows:
$$Q(t, \pi, a, b, r, z, d) = R_r(t, \pi, z \oplus e_{i_1} \cdots e_{i_k}).$$

\end{definition}

\subsection{The protocol}
\label{sec:prot-decs}

\emph{Parties.}
The argument system involves 
\begin{enumerate}
\item A (QPT) \emph{verifier} $V$,
\item A (QPT) \emph{prover} $P$, and
\item A (classical PPT) \emph{setup machine} $S$.
\end{enumerate}
The verifier sends a single quantum message to the prover in the preprocessing phase of the protocol, and the prover sends the verifier a single classical message in the instance-dependent phase of the protocol. $S$ sends an identical classical message to both the prover and the verifier during the preprocessing phase.

\emph{Inputs.}
(Unless otherwise stated, all inputs are common to all three parties.)
\begin{enumerate}
\item Preprocessing stage:
\begin{enumerate}
\item An instance size parameter $n$ and a security parameter $\mu$.
\item A $\QMA$ relation $(Q,\alpha,\beta)$.
\item The Clifford Hamiltonian $H(Q_n)$ (see equation \eqref{eq:clifford-hamiltonian}).
\item Other parameters:
\begin{enumerate}
\item $c(n)$, an upper bound on the number of qubits in a witness state;
\item $p(n)$, an upper bound on the number of qubits in a history state corresponding to an execution of $Q_n$ on a witness state of length $c(n)$ and an instance of size $n$;
\item $m = \mathrm{poly}(n)$, the number of terms in the Clifford Hamiltonian~(equation \eqref{eq:clifford-hamiltonian});
\item $N = \mathrm{poly}(n)$, the number of physical qubits per logical qubit in the Steane code introduced in section \ref{sec:bjsw}.
\end{enumerate}
\item A perfectly binding, quantum computationally concealing (classical) commitment scheme \textsf{Com} = $(\mathsf{Com.gen, Com.commit, Com.reveal, Com.verify, Com.recover})$, of the form described in section \ref{sec: commitment}.
\item A NIZK argument system with setup for NP, such as the one described in section \ref{sec:nizk-np}. We denote this argument system by a triple $(\textsf{NIZK}.S, \textsf{NIZK}.P, \textsf{NIZK}.V)$.
\item A levelled fully homomorphic encryption scheme with circuit privacy, \\$\mathsf{FHE = (FHE.Gen, FHE.Enc, FHE.Dec, FHE.Eval, FHE.Refresh)}$.
\end{enumerate}
\item Instance-dependent stage:
\begin{enumerate}
\item The instance $x$.
\item {Auxiliary (private) input to the prover:} a witness $\rho$ for the claim that $x \in R_{Q, \alpha}$.
\end{enumerate}
\end{enumerate} \clearpage

\emph{Protocol.}

Refer to Figure \ref{fig:original} for a diagram charting the essential structure of the protocol below.

\begin{figure}[H]
\rule[1ex]{16.5cm}{0.5pt}
\raggedright
\textbf{Preprocessing step:}

\vspace{2mm}

\begin{enumerate}
    \item[(S)] $S$ samples $\gamma \leftarrow \textsf{NIZK}.S(1^{q(n)}, 1^{\mu})$, for a sufficiently large polynomial $q>0$, in addition to $(pk_P,sk_P), (pk_V,sk_V) \leftarrow \textsf{Com.gen}(1^{\mu})$ (the two pairs are sampled independently). $S$ then outputs $(\gamma, pk_P, pk_V)$.
    \item[(i)] $V$ samples $r' \xleftarrow{\$} [m+n]$. Sets $r = r'$ if $r \leq m$, and $r = m+1$ otherwise. Computes $\sigma \leftarrow \textsf{Com.commit}(pk_V; r)$. Let $s_V$ be the randomness used.
    \item[(ii)] $V$ creates $2Np(n)$ EPR pairs ($2N$ for each qubit of the history state).
    The second qubit of each EPR pair will be sent to $P$ in step (iv). We refer to the remaining qubits (i.e. the first qubit in each EPR pair) as ``$V$'s qubits''.
    \item[(iii)] \label{step: measurement}
    $V$ interprets her $2Np(n)$ qubits as the encoding (according to the authentication code of Fig. \ref{fig: authentication code}) of a history state for $H(Q_n)$. 
    \begin{itemize}
        \item If $r \in [m]$: $V$ applies $C_r$ transversally to the subset of her $2Np(n)$ qubits which encode the $k$ logical qubits on which $C_r$ acts. Measures those qubits in the computational basis, obtaining an outcome string $z$.
        \item If $r=m+1$: $V$ measures the subset of the $2Np(n)$ qubits corresponding to the first qubit of $\mathcal{H}_{\textsf{clock}}$ and all the qubits of $\mathcal{H}_{\textsf{instance}}$ in the computational basis, obtaining an outcome string $z$.
    \end{itemize}
    \item[(iv)] \label{step: sending} $V$ samples $(pk_E,sk_E) \leftarrow \textsf{FHE.Gen}(1^{n})$. $V$ sends to $P$:
    \begin{itemize}
        \item $\alpha \leftarrow \textsf{FHE.Enc}(pk_E, (r, s_V, z))$.
        \item $pk_E$ and $\sigma$.
    \end{itemize}
\end{enumerate}
    
\rule[2ex]{16.5cm}{0.5pt}\vspace{-.5cm}
\caption{The preprocessing step of our NIZK argument for QMA}
  \label{fig: preprocessing step}
\end{figure}

\begin{figure}[H]
\rule[1ex]{16.5cm}{0.5pt}
\raggedright

\textbf{Instance-dependent step:}\\
\vspace{2mm}
\begin{itemize}
    \item  \textbf{Prover's message:}
    \begin{enumerate}
    \item $P$ computes the history state corresponding to an evaluation of the circuit $Q_n$ on the input $\ket{x} \otimes \ket{\psi}$. This is the state $\ket{\Psi} = \sum_{t=0}^{T}\ket{t}_{\textsf{clock}}  \otimes \Pi_{j=1}^{t} U_j \big( \ket{x} \otimes \ket{\psi} \otimes \ket{0}^{\otimes n} \big)$ for some unitary $U_j$, which can be computed efficiently. $P$ computes $\ket{\tilde{\Psi}} \leftarrow \textsf{Auth.Enc}(\ket{\Psi})$ according to the authentication scheme of Fig. \ref{fig: authentication code}. Let the sampled authentication keys be:
\begin{enumerate}
\item $a = a_1,..,a_{p(n)}$, $b= b_1,..,b_{p(n)}$ for $a_1,..,a_{p(n)}, b_1,..,b_{p(n)} \in \{0,1\}^{2N},$
\item $\pi \in S_{2N}$,
\item $t = t_1,..,t_{p(n)}$ where $t_1,..,t_{p(n)} \in \{0, +, +_{y}\}^N$.
\end{enumerate}

$P$ samples commitment randomness $s_P$, and computes $\sigma_{\textsf{keys}} \leftarrow \textsf{Com.commit}(pk, (t,\pi,a,b), s_P)$.

        \item $P$ teleports the state $\rho$ to $V$ using his halves of the $2Np(n)$ shared EPR pairs received in step (iv) of the preprocessing step. Let $d = (x_1,\ldots ,x_{2Np(n)}, y_1,\ldots, y_{2Np(n)}) \in \{0,1\}^{4Np(n)}$ be the Bell basis measurement outcomes obtained during the teleportation. 
        \item $P$ computes $\beta \leftarrow \textsf{FHE.Enc}\big(pk_E, ( d, \sigma, \sigma_{\textsf{keys}}, (t,\pi,a,b), s_P ) \big)$, where $\sigma$ is the commitment received in step (iv) of the preprocessing step.
       $P$ homomorphically evaluates the following circuit $C$ using $\beta$ and the ciphertext $\alpha$ that it received from the verifier. (Recall that $\alpha$ is an encryption of $(r, s_V, z)$.)
        \begin{quote}
            $C$ takes as input $d,\sigma, r, s_V, z, \sigma_{\textsf{keys}}, t, \pi,a,b,s_P$. It checks that $(r,s_V)$ is a valid opening for $\sigma$, and that $Q(t,\pi,a,b,r,z,d) = 1$, where $Q$ is defined in Definition~\ref{def:new-Q}. If its checks pass, using $\gamma$ it computes an NIZK argument for the existence of an opening to $\sigma_{\textsf{keys}}$ such that the opened value $(t,\pi,a,b)$ satisfies $Q(t,\pi,a,b,r,z,d) = 1$. If its checks do not pass, it outputs ``$\perp$''.
        \end{quote}
        \item Let $\tilde{\pi}$ be the encrypted proof that $P$ obtains in step 4. $P$ computes $\tilde{\pi}' \leftarrow \textsf{FHE.Refresh}(\tilde{\pi})$. Sends $d$, $\sigma_{\textsf{keys}}$ and $\tilde{\pi}'$ to $V$.
    \end{enumerate}
    \item \textbf{Verifier's check:} $V$ decrypts $\tilde{\pi}'$, and executes $\textsf{NIZK.V}$ to check the decrypted proof. It checks that the $d$ received from $P$ is the same $d$ that appears in the statement being proven.
\end{itemize}

\rule[2ex]{16.5cm}{0.5pt}\vspace{-.5cm}
\caption{The instance-dependent step of our NIZK argument for QMA}
\label{fig: instance-dependent step}
\end{figure}

\begin{theorem}\label{thm:main}
Assuming that LWE is intractable for quantum polynomial-time (QPT) adversaries, every language in \textsf{QMA} has an adaptively zero-knowledge non-interactive argument system with CRS setup and preprocessing (where completeness is according to Definition \ref{def:compl-alt}) with $\negl$ adaptive soundness. Moreover, the preprocessing phase consists of a single quantum message from the verifier to the prover. 
\end{theorem}

We refer to the combination of the protocols of Figures \ref{fig: preprocessing step} and \ref{fig: instance-dependent step}  as ``the protocol''. 

To show Theorem~\ref{thm:main} we start with an arbitrary language $L\in \QMA$. Using standard amplification techniques, for any polynomial $t$ there is a family of polynomial-size verification circuits $Q$ such that $L$ is the language associated with the $\QMA$ relation $(Q,1-2^{-t},2^{-t})$ as in Definition~\ref{def:qma-relation}. We show that the protocol associated to this relation is an NIZK argument with setup for $(Q,1-2^{-t},2^{-t})$. Completeness is easy to verify, as for any $(x,\rho)\in R_{Q,1-2^{-t}}$ the prover described in Figure~\ref{fig: instance-dependent step} is accepted with probability negligibly close to $1$, given access to $\rho$. In Section~\ref{sec: soundness} we prove soundness inverse polynomially close to $1$, and in Section~\ref{sec:parallel-ampl} we show how soundness can be amplified in parallel to any $2^{-p}$ for polynomial $p$ (provided $t$ is taken large enough compared to $p$). After parallel amplification, completeness holds only if we allow the prover to receive polynomially many copies of the witness (as in definition \ref{def:compl-alt}). Finally, in Section \ref{sec: zk property} we prove the zero-knowledge property. 

\section{Soundness}
\label{sec: soundness}
In this section we prove soundness of our protocol from Section~\ref{sec:prot-decs}. This is captured by the following lemma.

\begin{lem}\label{lem:soundness}
Assume that LWE is intractable for quantum polynomial-time (QPT) adversaries. Let $(Q,\alpha,\beta)$ be a $\QMA$ relation. Then the non-interactive protocol with setup and preprocessing for $(Q,\alpha,\beta)$ described in Section~\ref{sec:prot-decs} has negligible adaptive soundness.
\end{lem}

We give an overview of the proof of Lemma~\ref{lem:soundness} in the next subsection.

\subsection{Overview}
\label{sec:soundness-overview}

The structure of the proof is as follows. We show through a sequence of hybrids that it is possible to transform an execution of our protocol on some instance $x$, into an execution of the protocol from \cite{broadbent2016zero} on a specific local Clifford Hamiltonian derived from $x$. We show that this transformation can at most negligibly decrease the optimal acceptance probability of the prover. Thus, soundness of our protocol reduces to soundness of the protocol from \cite{broadbent2016zero}. The main steps in our sequence of hybrids are the following: 
\begin{itemize}
\item Remove the encryption of $V$'s choice of $r$, randomness $s_V$ and measurement outcomes $z$ sent in step (iv) of the preprocessing step.
\item Replace the step where $P$ teleports the encoded witness to $V$ through shared EPR pairs (step 2 in Fig. \ref{fig: instance-dependent step}) with one where $P$ directly sends the qubits of the encoded witness to $V$.
\item Remove the portion of the CRS corresponding to the NIZK argument, and replace the NIZK argument sent by the prover in step 4 of Fig. \ref{fig: instance-dependent step} with a ZK proof.
\end{itemize}

In more detail, fix a $\QMA$ relation $(Q,\alpha,\beta)$ and an instance $x$ of size $n$. We consider the following sequence of hybrid experiments. Each experiment describes a modified verifier, and thus a modified protocol, where the first hybrid $H_0$ is the protocol from Section~\ref{sec:prot-decs} and the last hybrid $H_{10}$ coincides with the protocol from \cite{broadbent2016zero}. In subsection \ref{sec:hybrids}, we will argue that the optimal probability of a prover being accepted in any of the protocols can only increase (or at most negligibly decrease) across hybrids. Since the protocol from \cite{broadbent2016zero} has soundness $1-1/\text{poly}$, this shows that the protocol from Section~\ref{sec:prot-decs} also has  $1-1/\text{poly}$ soundness. In subsection \ref{sec:parallel-ampl}, we amplify soundness by repeating the protocol in parallel. One can check that our proof goes through unchanged for the case of adaptive soundness as well. In particular, the key is that the NIZK proof system for NP employed in our protocol is adaptively sound.

The following are the hybrids.
\vspace{2mm}

$H_0$: This is the real experiment in which $V$ behaves as described in Figure~\ref{fig: preprocessing step} and Figure~\ref{fig: instance-dependent step}.
\vspace{2mm}

$H_1$: Same as $H_0$, except $V_1$ does not send to $P$ an encryption of $r$, $s_V$ and her outcomes $z$. Instead, she sends an encryption of the zero string. The instance-dependent step is now interactive, and proceeds as follows:
\begin{itemize}
    \item[(a)] $V_2$ expects $d$ and $\sigma$ from $P$.
    \item[(b)] Let $z$ be the outcome of $V_1$'s measurements from step (iii) of the preprocessing step. $V_1$ Sends $z, r$ and $s_V$ to $P$ (in the clear).
    \item[(c)] $V_1$ expects an NIZK argument (in the clear) of existence of an opening of $\sigma$ such that the opened keys $t,\pi,a,b$ satisfy ${Q}(t,\pi,a,b,r,z,d) = 1$.
\end{itemize}
\vspace{2mm}

$H_2$: Same as $H_1$ except $V_2$ does not send the encryption of the zero string. 
\vspace{2mm}

$H_3$: Same as $H_2$, except $V_3$'s measurements in step (iii) of the preprocessing step are postponed. $V_3$ proceeds directly to the instance-dependent step as in $H_2$. Only after receiving $d, \sigma$, $V_3$ measures according to step (iii), and sends the outcome $z$ of these measurements to $P$. The rest is the same.
\vspace{2mm}

$H_4$: We modify $H_3$ as follows: $V_4$ does not send any EPR pairs to $P$. In step $(a)$ of the instance-dependent step (using the notation of $H_1$), $V_4$ expects $P$ to send a $2Np(n)$-qubit state directly, together with a string $d$ and a commitment $\sigma$. $V_4$ obtains $z$ by measuring the received state. The rest is the same.
\vspace{2mm}

$H_5$: Same as $H_5$, except $V_5$ does not expect any string $d$. The relation for which $V_5$ expects an NIZK argument is modified so that $d$ (which was previously part of the instance) is now fixed to $0$ (or equivalently $\widetilde{Q}$ is replaced with $Q$, from Definition~\ref{def:bjsw-Q}).
\vspace{2mm}

$H_6:$ Same as $H_5$, except $V_6$ does the following: in part (i) of the preprocessing step, $V_6$ samples $r' \xleftarrow{\$} \in [m+n]$. $V_6$ computes 
 $\sigma \leftarrow \textsf{Com.commit}(pk_V; r')$. In part (iii) of Fig. \ref{fig: preprocessing step}, 
\begin{itemize}
    \item If $r\in\{1,\ldots,m\}$: same as before.
    \item If $r = m+i, \, i\in [n]$: $V_6$ measures the subset of the $2Np(n)$ qubits corresponding to the first qubit of $\mathcal{H}_{\textsf{clock}}$ and to the $i$-th qubit of $\mathcal{H}_{\textsf{circuit-desc}}$ in the computational basis, obtaining an outcome string $z$.
\end{itemize}
In the final step, $V_6$ expects from $P$ an NIZK argument for the existence of an opening to the committed keys such that the opened keys $t,\pi,a,b$ satisfy $Q'(t,\pi,a,b,r,z) =1$, where $Q'$ is defined as follows from $R_r'$. 

\begin{definition}[Definition of $R'_{r}$, $r\in\{m+1,\ldots,m+n\}$]
\label{def:new-soundness R_r}
Let $u_{\textsf{clock}_1}, u_{\textsf{instance}_1}, .., u_{\textsf{instance}_{n}} \in \{0,1\}^{2N}$ (these represent the measurement results that the verifier sends to the prover in step (iii) of Fig \ref{fig: preprocessing step} when $r=m+1$). For $\pi \in S_{2N}$, for each $i \in \{\textsf{clock}_1,\textsf{instance}_1, \ldots, \textsf{instance}_{n}\}$, define strings $p_i, q_i$ in $\{0,1\}^N$ such that $\pi(p_i \| q_i) = u_i$ (alternatively: $\pi^{-1}(u_i) = p_i \| q_i)$. Let $u = u_{\textsf{clock}_1}, u_{\textsf{instance}_1}, .., u_{\textsf{instance}_n}$. The predicate $R'_{m+i}(t, \pi, u)$ takes the value 1 if and only if the following two conditions are met:
\begin{enumerate}
\item $p_{\textsf{clock}_1} \in \mathcal{D}_N^1$ (this corresponds to the \textit{first} qubit of the clock register, expressed in unary, being in state $1$, i.e. the clock register is \textit{not} at time $0$), \emph{or} $p_i \in \mathcal{D}^{x_i}_N$.
\item $\braket{q_{\textsf{clock}_1}\, q_{\textsf{instance}_1} \cdots q_{\textsf{instance}_n} | t_{\textsf{clock}_1} \, t_{\textsf{instance}_1}\cdots t_{\textsf{instance}_n} }\neq 0$
\end{enumerate}
\end{definition}
\begin{definition}[Definition of Q']
\label{def: Q'}
\begin{itemize}
\item If $r\in\{1,\ldots,m\}$, $$Q'(t, \pi, a, b, r, u) = R'_r(t, \pi, u \oplus e_{i_1} \cdots e_{i_k}),$$ 
where $i_1,..,i_k$ are the indices of the qubits on which the $C_r$ acts non-trivially and $e_{i_1}, \ldots, e_{i_k}$ are the unique strings such that
\begin{gather*}
C_r^{\otimes 2N}(X^{a_{i_1}} Z^{b_{i_1}} \otimes \cdots \otimes X^{a_{i_k}} Z^{b_{i_k}}) = \alpha(X^{e_{i_1}} Z^{f_{i_1}} \otimes \cdots \otimes X^{e_{i_k}} Z^{f_{i_k}}) C_r^{\otimes 2N}
\end{gather*}
for some $\alpha \in \{1, i, -1, -i\}$ and some $f_{i_1}, \ldots, f_{i_k} \in \{0,1\}^{2N}$.
    \item If $r=m+i$, for $i \in [n]$: $$Q'(t, \pi, a, b, r, u) = R'_r(t, \pi, u \oplus a_{\textsf{clock}_1} a_{\textsf{instance}_1} \cdots a_{\textsf{instance}_n}).$$
\end{itemize}
\end{definition}

$H_7$: Same as $H_6$, except $V_7$ does not send a commitment to the randomness $r$. We also drop the associated public key $pk_V$ from the CRS. The first message in the protocol is from $P$ who is expected to send a state and a commitment. After that, $V_7$ sends $z,r$ to $P$. The rest is the same.
\vspace{2mm}

$H_8$: Same as $H_7$, except we remove $\gamma$ from the CRS, and the NIZK proof from $P$ to $V_8$ in the last round is replaced by a (interactive) ZK proof.  
\vspace{2mm}

$H_9$: Same as $H_8$, except that after the first message of the prover, $V_9$ and $P$ engage in a coin flipping protocol. This determines the randomness $r$. In the next round $V_9$ sends her measurement outcome $z$ to $P$, and the rest is the same.
\vspace{2mm}

$H_{10}$: Same as $H_9$, except we remove $pk_P$ from the CRS, and we replace the public-key commitment scheme $\textsf{Com}$ with the commitment scheme $\textsf{commit}$ used by the prover in the $\cite{broadbent2016zero}$ protocol.

In Section~\ref{sec:hybrids} we show that, up to negligible quantity, the maximum success probability of a prover in the protocol described in $H_{10}$ can only be higher than the maximum success probability of a prover in the protocol described in $H_0$. In Section~\ref{sec:sound-conc} we show that this implies soundness of the protocol. Finally, in Section~\ref{sec:parallel-ampl} we show that soundness can be amplified in parallel. 

\subsection{Success probability in hybrids}
\label{sec:hybrids}

\begin{lem}
\label{lem: hybrids}
The optimal probability of a prover being accepted in $H_0$ is at most negligibly higher than in $H_{10}$.
\end{lem}

\begin{proof}
We argue that the optimal probability of the prover can decrease at most negligibly across hybrids.

\paragraph{$H_0 \approx H_1$.} 
We start with the following claim.
\begin{claim}
\label{claim: partial}
For any prover $P^*$ that plays $H_0$ or $H_1$ up to step (a) of $H_1$, the probability that $P^*$ outputs a pair $(d, \sigma)$ such that, letting $t,\pi,a,b$ be the value committed in $\sigma$ and $z,r$ be $V$'s message in step (b), we have $\widetilde{Q}(t, \pi, a, b, r, z, d) = 1$ is negligibly close for $H_0$ and $H_1$.
\end{claim}

\begin{proof}
Suppose for a contradiction that there is a $P^*$ such that the probability that $P^*$ outputs a good $d$ in $H_0$ and in $H_1$ differ non-negligibly. Then, there exists a polynomial $q>0$ such that the difference between the probabilities is at least $1/q(\mu)$ for infinitely many values of the security parameter $\mu$. Suppose it is the first probability that is at least at least $1/q(\mu)$ larger infinitely often. (the reverse case is similar).

We construct an adversary $\mathcal{A}$ that breaks CPA security of \textsf{FHE}. $\mathcal{A}$ proceeds as follows: 
\begin{itemize}
    \item $\mathcal{A}$ receives a public key $pk_E$ for the $\textsf{FHE}$ scheme from the challenger. 
    \item $\mathcal{A}$ runs $\gamma \leftarrow \textsf{NIZK-CRS.Setup}$ and $(\tilde{pk}_P, \tilde{sk}_P), (\tilde{pk}_V, \tilde{sk}_V)\leftarrow \textsf{Com.gen}(1^{\mu})$. $\mathcal{A}$ then runs the honest verifier $V$ and $P^*$ on common input the CRS $(\gamma, \tilde{pk}_P,  \tilde{pk}_V)$, except that if forwards $pk_E$ to $V$ in step (iv) (it does not matter that $\mathcal{A}$ cannot give the secret key to $V$, since $V$ only uses this in the very final ``Verifier's check'' of Fig. \ref{fig: instance-dependent step}, which $\mathcal{A}$ will not need to run). Let $z$ be the outcomes of $V$'s measurements in step (iii). Let $r$ and $s_V$ be respectively $V$'s committed random choice of Hamiltonian term and the randomness used for committing to it. 
    $\mathcal{A}$ sets $m_0 = r, s_V, z$, and $m_1 = 0$. Sends $m_0$ and $m_1$ to the challenger.
    \item The challenger returns a ciphertext $c$. $\mathcal{A}$ forwards this to $P^*$ as the first part of the verifier's message in step (iv), and runs $P^*$ obtaining output $d, \sigma$ where $\sigma$ is some commitment.  
    \item 
    $\mathcal{A}$ decrypts $\sigma$ using $\tilde{sk}_P$, obtaining $(t,\pi,a,b)$. If the output is not of this form or $(t,\pi,a,b)$ is not consistent with $d,r,z$, then $\mathcal{A}$ guesses that the ciphertext received was an encryption of $m_1$; otherwise of $m_0$.
\end{itemize}

Denote by $b$ the bit sampled by the challenger. Notice that in the case that $b=0$, the distribution of the output $d$ that $\mathcal{A}$ obtains from $P^*$ is precisely that of $H_0$. When $b=1$, it is precisely that of $H_1$. Given the hypothesis, this implies that $\mathcal{A}$ wins the CPA-security game with $1/q$ advantage infinitely often, which contradicts CPA-security of $\textsf{FHE}$.
\end{proof}

Next, we show that the optimal probabilities of a prover being accepted in $H_0$ and in $H_1$ are negligibly close. Let $P$ be a prover for $H_0$. We  construct a prover $P'$ for $H_1$ such that 
\[\Pr\big(P' \text{ is accepted}\big) > \Pr\big(P \text{ is accepted}\big) - \negl(|x|)\;.\]
$P'$ runs as follows:
\begin{itemize}
    \item As part of $H_1$, $P'$ receives a CRS $(\gamma, pk_P, pk_V)$. It then receives a commitment $\tilde{\sigma}$ and a ciphertext $c$. $P'$ runs $(\tilde{pk}_P, \tilde{sk}_P)\leftarrow \textsf{Com.gen}(1^{\mu})$. It provides $(\gamma, \tilde{pk}_P, pk_V)$ to $P$ as CRS, and $c$ as the ciphertext. $P$ returns $d$ and $\sigma$. $P'$ decrypts $\sigma$ using $\tilde{sk}_P$. Let $t,\pi,a,b$ be the committed value (if the committed value is not of this form or if decryption fails, $P'$ aborts). $P'$ computes $\sigma' \leftarrow \textsf{Com.commit}(pk_P; t,\pi,a,b)$. Let $s'$ be the randomness used. $P'$ returns $d, \sigma'$ to the verifier of $H_1$.
    \item $P'$ receives $z, r, s_V$. If $z, r, d, t, \pi, a,b$ are consistent, $P'$ computes an NIZK proof of the existence of an opening to $\sigma'$ such that the opened keys are consistent with $z,r,d$.
\end{itemize}

Since the NIZK argument system employed has negligible soundness (we chose the NIZK argument system from Section~\ref{sec:nizk-np}), the probability that $P$ is accepted in $H_0$ conditioned on returning a bad $(d, \sigma)$ is negligible. By Claim \ref{claim: partial}, it follows that
\begin{align}
\Pr\big(P' \text{ is accepted}\big) &> \Pr\big(P \text{ returns a good } (d, \sigma)\big) - \negl(|x|) \nonumber \\
&> \Pr\big(P \text{ is accepted}\big) - \negl(|x|) \nonumber
\end{align}
where the first inequality follows from Claim \ref{claim: partial}, and the second is because the NIZK argument system employed has negligible soundness, and so the probability that $P$ is accepted in $H_0$ conditioned on returning a bad $(d, \sigma)$ is negligible.

We have thus shown that the optimal probability of acceptance in $H_1$ is at least that of $H_0$, up to a negligible quantity. An analogous argument shows the converse. Hence, the optimal probabilities of acceptance in $H_0$ and $H_1$ are negligibly close. 

\paragraph{$H_1 \approx H_2$.} The encryption of the zero string can be simulated by the prover itself, so clearly the optimal acceptance probabilities in $H_1$ and $H_2$ are exactly equal. 
\paragraph{$H_2 \approx H_3$.} For any prover $P^*$, since the action of $P^*$ to produce $d,\sigma$ and the measurements of $V$ from step (iii) of the preprocessing step commute, performing them before receiving $d, \sigma$ or after does not change the acceptance probability of $P^*$.
\paragraph{$H_3 \approx H_4$} We show that the optimal probability of a prover being accepted in $H_3$ and $H_4$ is the same. Given a prover $P$ for $H_3$ one can construct a prover $P'$ that is accepted in $H_4$ with the same probability: $P'$ starts by sharing $2Np(n)$ EPR pairs with $P$. When $P$ outputs $(d, \sigma)$, $P'$ forwards $(d,\sigma)$ as well as his half of the $2Np(n)$ EPR pairs to $V$. When $V$ returns a message of the form $(z,r,s_V)$, $P'$ forwards this to $P$. $P$' returns the final message of $P$ to $V$. It is clear that the acceptance probability of $P'$ is the same as that of $P$, as the distribution of the whole transcript is identical to that of experiment $H_3$. The reverse direction is similar. 
\paragraph{$H_4 \approx H_5$.} Again, the optimal probabilities in $H_4$ and $H_5$ are the same. For any prover $P$ for $H_4$, we construct $P'$ for $H_5$ which is accepted with the same probability. $P'$ runs $P$. The only difference is that when $P$ outputs $d = (a,b)$ and a state $\rho$ to send to $V$, $P'$ sends $X^aZ^b\rho (X^a Z^b)^\dagger$ to $V$. 

\paragraph{$H_5 \approx H_6$.} We show that the optimal acceptance probability can only increase. Let $P$ be a prover in $H_5$. We construct a prover $P'$ for $H_6$ as follows. $P'$ samples a CRS $(\gamma, pk_P, pk_V)$ such that it knows the secret key $sk_P$ corresponding to $pk_P$. Upon receiving a commmitment (to some randomness) from the verifier, it runs $P$ obtaining $\sigma, \tilde{\rho}$. It decrypts $\sigma$ using $sk_P$. Let $(t,\pi,a,b)$ be the opened value (if it is not of this form, $P'$ aborts). $P'$ sends $\sigma, \tilde{\rho}$ to the verifier. Upon receiving $z, r$ from the verifier, $P'$ does the following:
\begin{itemize}
    \item If $r\in \{1,\ldots,m\}$: $P$' obtains an NIZK proof from $P$ and forwards this to the verifier.
    \item If $r = m+i, \, i\in [l(n)]$: $P'$ computes an NIZK proof for the existence of an opening to $\sigma$ such that the opened keys $(t,\pi,a,b)$ satisfy $Q'(t,\pi,a,b, r, z) = 1$ (where $Q'$ is as defined in \ref{def: Q'}).
\end{itemize}
Notice that in $H_6$, the verifier only measures the subset of physical qubits corresponding to the $i$-th qubit of the circuit description, and the check that it performs is strictly less stringent than in $H_5$. A simple calculation shows that the probability of acceptance of $P'$ in $H_6$ is at least that of $P$ in $H_5$. 

\paragraph{$H_6 \approx H_7$.} We first show that the acceptance probability of a prover $P$ cannot change more than negligibly if the commitment to $r$ in $H_6$ is replaced by a commitment to a zero string. Suppose for a contradiction that there exists $P$ such that the probability of acceptance in $H_6$ differs non-negligibly from $H_7$. This implies that there exists a polynomial $q>0$ such that the two probabilities differ by at least $1/q(\mu)$ for infinitely many values of the security parameter $\mu$. Suppose the first probability is higher by at least $1/q$ infinitely often (the reverse case is similar). Then, we can construct an adversary $\mathcal{A}$ that breaks the hiding property of the commitment. $\mathcal{A}$ samples a random $r$ and sends $m_0=r$ and $m_1 = 0$ to the challenger. Upon receiving a commitment and a public key $pk_P$ from the challenger, $\mathcal{A}$ simulates a CRS which includes $pk_P$ and an execution of the verifier from $V_7$ (or equivalently $V_8$) except that it forwards to $P$ the commitment received. If the verifier accepts, $\mathcal{A}$ guesses that it received a commitment to $r$, otherwise that it received a commitment to $0$. Notice that when the challenger sends a commitment to $r$ the view of $P$ is exactly as in $H_6$, and when it sends a commitment to $0$ the view of $P$ is exactly as in $H_7$. Thus $\mathcal{A}$ has $1/q$ distinguishing advantage infinitely often, which gives the desired contradiction.

\paragraph{$H_7 \approx H_8$.} We show that the optimal acceptance probability decreases at most negligibly. Let $P$ be a prover for $H_7$. We construct a prover $P'$ for $H_8$ as follows. $P'$ gets as input the $CRS$ $pk_P$. Then, $P'$ samples a CRS $\tilde{pk}_P$ such that it knows the secret key $\tilde{sk}_P$ corresponding to $\tilde{pk}_P$. It forwards $\tilde{pk}$ to $P$ as the CRS, and obtains $\sigma, \tilde{\rho}$. It decrypts $\sigma$ using $\tilde{sk}_P$. Let $(t,\pi,a,b)$ be the opened value (if it is not of this form $P'$ aborts). $P'$ computes $\sigma' \leftarrow \textsf{Com.commit}(pk_P; t,\pi,a,b)$. Let $s'$ be the randomness used. Sends $\tilde{\rho}, \sigma'$ to the verifier. $P'$ receives $z,r$ from the verifier, and checks that $Q(z,r,t,\pi,a,b) = 1$. If so, it engages with the verifier in a ZK proof of existence of an opening to $\sigma'$ such that the opened keys, together with $z,r$, satisfy the predicate $Q$. Since the NIZK proof has negligible adaptive soundness, the probability that $P$ is accepted in $H_8$ conditioned on $z,r, t,\pi,a,b$ being inconsistent is negligible. This implies that 
\[ \Pr\big(P' \text{ is accepted in } H_8\big)> \Pr\big(P \text{ is accepted in } H_7\big) -negl(|x|)\;.\]

\paragraph{$H_8 \approx H_9$.} From a prover $P$ for $H_8$, we obtain a prover $P'$ for $H_9$ with the same acceptance probability as follows: $P'$ runs $P$ for the first round. Then $P'$ engages in an honest coin flipping protocol with the verifier and forwards the outcome $r$ to $P$. Finally, $P'$ runs $P$ until the end. 

\paragraph{$H_9 \approx H_{10}$.} We show that the optimal acceptance probability can only increase. Let $P$ be a prover for $H_9$. We construct a prover $P'$ for $H_{10}$ as follows. Similar to the earlier argument for $H_7 \approx H_8$, $P'$ samples a CRS $\tilde{pk}_P$ such that it knows the corresponding secret key $\tilde{sk}_P$. It forwards $\tilde{pk}_P$ as CRS to $P$ who returns $\tilde{\rho}, \sigma$. $P'$ decrypts $\sigma$ using $\tilde{sk}_P$. Let $(t,\pi,a,b)$ be the opened value. $P'$ computes $\sigma' \leftarrow \textsf{commit}(t,\pi,a,b)$. Let $s'$ be the randomness used. $P'$ Sends $ \sigma, \tilde{\rho}$ to the verifier. $P'$ engages honestly in a coin-flipping protocol with $V$. Let $r$ be the outcome. Upon receiving $z$ from the verifier, $P'$ checks that $Q(z,r,t,\pi,a,b) = 1$. If so, it engages in a ZK proof of an opening to $\sigma'$ such that the opened keys, together with $z,r$ satisfy the predicate $Q$. Since the ZK proof has negligible soundness, the probability that $P$ is accepted in $H_9$ conditioned on $Q(z,r , t,\pi,a,b) = 0$ is negligible. Moreover, since the interaction of $P$ with $V$ in the coin-flipping protocol could have been simulated by a uniformly random $r$ together with some local operation on the register of $P$, the probability that $Q(z,r,t,\pi,a,b) = 1$ is exactly the same for $P$ and $P'$. This implies that 
\[\Pr\big(P' \text{ is accepted in } H_{10}\big) > \Pr\big(P \text{ is accepted in } H_9\big) -\negl(|x|)\;.\] 
\end{proof}

\subsection{Reduction to~\cite{broadbent2016zero}}
\label{sec:sound-conc}

Consider the following Hamiltonian:
\begin{equation}
    H = H(Q_n) +\ket{0}\bra{0}_{\textsf{clock}}\otimes \sum_{i} \ket{\overline{x_i}}\bra{\overline{x_i}}_{\textsf{instance}_i}
\end{equation}
Notice that this is a \textit{local} Clifford Hamiltonian (once the clock is expressed in unary). It is immediate to verify that the protocol described in $H_{10}$ is identical to an execution of the protocol from \cite{broadbent2016zero} on common input $H$ (see Section~\ref{sec:bjsw}). By invoking the soundness of the \cite{broadbent2016zero} protocol, we deduce that protocol $H_{10}$ is sound, and as a consequence of Lemma~\ref{lem: hybrids}, protocol $H_0$ (the original protocol) is also sound. This concludes the soundness analysis.

\subsection{Parallel amplification of soundness}
\label{sec:parallel-ampl}

%



Let $\Pi=(S,P,V)$ be a non-interactive argument with setup. Let $k\geq 1$ be an integer. We define the $k$-th parallel repetition of $\Pi$, $\Pi^{(k)} = (S^{(k)},P^{(k)},V^{(k)})$ as follows: 
\begin{itemize}
\item For every $\mu$ and $n$, $S^{(k)}_{\mu n}$ executes $k$ independent copies of $S_{\mu n}$ and returns their outputs. 
\item For every $\mu$ and $n$, $V^{(k)}_{\mu n} = (V^{(k)}_{1,\mu n}, V^{(k)}_{2,\mu n})$ and $P^{(k)}_{\mu n} = (P^{(k)}_{1,\mu n}, P^{(k)}_{2,\mu n})$. $V^{(k)}_{1,\mu n}$ and $P^{(k)}_{1,\mu n}$ take as input a classical string, interpreted as $k$ outputs of $S_{\mu n}$. They execute $k$ copies of $V_{1,\mu n}$ and $P_{1,\mu n}$ respectively on the $k$ inputs and they each produce a $k$-register quantum state as output.  $V^{(k)}_{2,\mu n}$ takes as input the output of $V^{(k)}_{1,\mu n}$ and an instance $x$. $P^{(k)}_{2,\mu n}$ takes as input the output of $P^{(k)}_{1,\mu n}$, $x$ and a $k$-register quantum state $\rho^{(k)}$. Each of them executes $k$ copies of $V_{2,\mu n}$ and $P_{2,\mu n}$ respectively, where the $i$-th copy of $P_{2,\mu n}$ is executed with the $i$-th register of $\rho^{(k)}$ as its quantum input. $V_{\mu n}^{(k)}$ returns $1$ if and only if each copy of $V_{\mu n}$ returns $1$. 
\end{itemize}

For classical non-interactive arguments (and more generally arguments of up to three rounds) it is well-known that parallel repetition amplifies soundness exponentially~\cite{bellare1997does}. In our setting the argument is made more delicate by the fact that the message sent from verifier to prover in the setup phase is quantum. We do not know if parallel repetition amplifies soundness in a black-box way for (quantum) non-interactive arguments with setup. However, using the specific structure of our protocol we can give a tailored argument.

\begin{lem}\label{lem:sound-ampl}
Let $(Q,\alpha,\beta)$ be a $\QMA$ relation and $\Pi=(S,P,V)$ the argument system for $(Q,\alpha,\beta)$ described in Section~\ref{sec: protocol}. Let $c$ and $s$ be the completeness and soundness of the protocol respectively. Then for any integer $k\geq 1$, $(S^{(k)},P^{(k)},V^{(k)})$ is an argument system for $(Q,\alpha,\beta)$ with completeness $1-k(1-c)$ and soundness $s'$, where $s'$ is negligibly close to $s^k$.
\end{lem}

\begin{proof}
Let $k\geq 1$ and for ease of notation write $\Pi' = (S',P',K')=(S^{(k)},P^{(k)},V^{(k)})$. We also fix a $\mu$ and an integer $n$ and omit the indices.

\paragraph{Completeness.} Let $(x,\rho) \in R_{Q,\alpha}$. Then by definition $Q$ accepts $\rho$ with probability at least $\alpha$. By the completeness property of $\Pi$, $(S,P(x,\rho_i),V(x))$ returns $1$ with probability at least $c$. By a union bound, $(S',P'(x,\rho^{\otimes k}),V'(x))$ returns $1$ with probability at least $1-k(1-c)$. 

\paragraph{Soundness.} 
For ease of notation we give the proof for $k=2$, as the general case is similar. Suppose for contradiction that there exists an $x\in N_{Q,\beta}$ and a quantum polynomial-time $P^*$ such that $(S',P^*(x),V'(x))$ returns $1$ with probability non-negligibly larger than $s'=s^2$. 

As a first step we show that there exists a $\tilde{P}^*$ that succeeds with probability non-negligibly larger than $s'=s^2$ in an (interactive) hybrid variant $H'_1$ of the protocol, that is identical to hybrid $H_1$ considered in the soundness proof (see Section~\ref{sec:soundness-overview}), except for the following modifications: First, the protocol described in $H_1$ is executed twice in parallel. Second, at step (a) of $H_1$, the prover is expected to report, in addition to the commitments $\sigma_i$, valid openings $(t_i,\pi_i,a_i,b_i)$ (including the randomness used) for $i\in\{1,2\}$. (The verifier checks that the openings are valid.) We let $\tilde{V}$ be the verifier in the hybrid $H'_1$, and denote by $\tilde{\Pi}'$ the resulting interactive protocol. 

\begin{claim}\label{claim:ampl-1}
 The maximum success probability of a prover in protocol $\Pi'$ and in protocol $\tilde{\Pi}'$ are negligibly close. 
\end{claim}

\begin{proof}
The fact that the maximum success probability of a prover in the two-fold parallel repetition of $H_1$ is negligibly close to that in $\Pi'$ follows from the soundness analysis in Section~\ref{sec:hybrids}, with straightforward modifications to handle the parallel repetition. 

Next, we establish that the optimal success probabilities in $H_1'$ and in the two-fold parallel repetition of $H_1$ are negligibly close. From now on, for convenience, we say $H_1$ to mean the two-fold parallel repetition of $H_1$. Clearly, since the prover is asked for more information in $H_1'$ than in $H_1$ the maximum success probability cannot increase. To show that it can only decrease negligibly, let $P^*$ be a prover that succeeds in $H_1$ with some probability, and consider the following prover $\tilde{P}$ for $H_1'$. 

$\tilde{P}$ first executes the setup phase with $\tilde{V}$ exactly as $P^*$ would. As a result, it obtains two CRS $(\gamma_1,pk_{P,1},pk_{V,1})$ and $(\gamma_2,pk_{P,2},pk_{V,2})$ as well as two $2Np(n)$-qubit states $\rho_1$ and $\rho_2$ and classical information $(pk_{E,1},\sigma_1,\alpha_1)$ and $(pk_{E,2},\sigma_2,\alpha_2)$ from the preprocessing phase $\tilde{V}_1$ of $\tilde{V}$. 

Next, for $i\in\{1,2\}$ $\tilde{P}$ generates new CRS $(\tilde{\gamma}_i,\tilde{pk}_{P,i},\tilde{pk}_{V,i})$ such that $\tilde{\gamma}_i = \gamma_i$ and $\tilde{pk}_{V,i} = pk_{V,i}$ but $\tilde{pk}_{P,i}$ is sampled jointly with a secret key $\tilde{sk}_{P,i}$ for $\textsf{Com}$.

In the instance-dependent interactive phase, after having received $\tilde{V}_2$'s first message (encryptions of $0$), $\tilde{P}$ simulates $P^*$, \emph{except} that it uses the new CRS in place of the older one. Receiving $P^*$'s messages $(d_i,\sigma_i)$ for $i\in\{1,2\}$, $\tilde{P}$ uses $\tilde{sk}_{P,i}$ to open $\sigma_i$. Let $(t_i,\pi_i,a_i,b_i)$ be the opened keys. $\tilde{P}$ samples $\tilde{\sigma}_i \leftarrow  \textsf{Com.commit}(t_i,\pi_i,a_i,b_i)$. Sends $(d_i, \tilde{\sigma}_i)$ together with the opening $(t_i,\pi_i,a_i,b_i)$ and the randomness used to the verifier. Then, proceeds to play the remaining part of $H_1'$ optimally. The optimal probability of being accepted at this point only depends on whether the committed values are consistent with the information received from the verifier in the subsequent step. Thus $\tilde{P}$ performs at least negligibly close to $P^*$ (if not better). 
\end{proof}

Let $P^*$ be a prover with optimal success probability in $H_1'$. 
Let $p_1$ be the probability that the first proof returned by $P^*$ is accepted by the first verifier, and $p_2$ the probability that the second proof is accepted, conditioned on the first one being accepted. Since the maximum success probability in $\Pi$ is $s$, it follows from (the single-repetition analogue of) Claim~\ref{claim:ampl-1} that the maximum success probability of any prover in the single-repetition analogue $\tilde{\Pi}$ of $\tilde{\Pi}'$ is at most negligibly higher than $s$. By a straightforward simulation argument it follows that $p_1$ is at most negligibly larger than $s$. Applying Bayes' rule, $p_2$ is non-negligibly higher than $s$. In the remainder of the proof we derive a contradiction by showing   that there is an (efficient) $P'$ such that $P'$ succeeds in the single-repetition hybrid $\tilde{\Pi}$ with probability negligibly close to $p_2$.

\paragraph{Definition of $P'$.}
Let $\rho_2, m_2$ be respectively the quantum state and the classical message received from $V_1$ in the preprocessing phase. $P'$ independently simulates another execution of $V_1$ and $V_2$ to obtain $\rho_1$ and $m_1$. Instead of executing directly $P^*$ on these inputs, it executes it on the state $\rho_1$, but the state $\rho_2$ is replaced by half-EPR pairs. 

Let $(d'_i,\sigma_i)$ and $(t_i,\pi_i,a_i,b_i)$, for $i\in\{1,2\}$, be the outcomes obtained from $P^*$. Note that at this point $P'$ can already verify if the outcomes associated with the first repetition $i=1$ will lead to acceptance in the protocol, simply by running the simulated verifier one step forward to obtain $(z_1,r_1,s_1)$ and checking whether $Q(t_1,\pi_1,a_1,b_1,r_1,z_1,d_1)=1$. If this is not the case, $P'$ repeats the first part with a new execution of the setup phase $V_1$ and fresh half-EPR pairs. 

Suppose now that $P'$ has obtained $(d_i,\sigma_i)$ and $(t_i,\pi_i,a_i,b_i)$, for $i\in\{1,2\}$ from $P^*$ such that the execution of the first copy of the verifier leads to acceptance. 

To complete the protocol, $P'$ now performs a Bell basis measurement between the half-EPR pairs not given to $P^*$ and the state $\rho_2$, that until now has not been used. Let $d'_2 = (x'_i,z'_i)$ denote the teleportation outcomes obtained. $P'$ then pursues the interaction with $V_2$, except that it replaces the outcome $d_2 = (x_i,z_i)$ obtained from $P^*$ by $d_2 \oplus d'_2 = (x_i\oplus x'_i, z_i \oplus z'_i)$. All other messages are the same. This completes the description of $P'$. 

\paragraph{Success probability of $P'$.}
To conclude we claim that $P'$ succeeds with probability non-negligibly larger than $s$ in $\tilde{\Pi}$. Note that the result of the ``teleportation'' operation performed by $P'$ is that the prover $P^*$ is executed on the state $(X^{x'} Z^{z'}) \rho_2 (X^{x'} Z^{z'})^\dagger$ instead of the state $\rho_2$. Since the state $\rho_2$ is obtained from half-EPR pairs that were prepared and measured by the verifier $V_1$, the effect of this one-time pad is identical to not modifying $\rho_2$, but instead modifying the verifier's measurement outcomes $z$ by XORing them with a string $e'_i$ defined as in~\eqref{eq:clifford-com} with $(a,b)$ on the left-hand side replaced by $(x',z')$. Let $\tilde{z}$ be the modified outcomes. Using that by definition for any $t,\pi,a,b,r,d$ it holds that 
\[\tilde{Q}(t,\pi,a,b,r,\tilde{z},d) = \tilde{Q}(t,\pi,a,b,r,z,\tilde{d})\;,\]
where $\tilde{d}$ is computed from $d$ as done by $P'$ (described above), it follows that the simulation of $P^*$ that is performed by $P'$ is indistinguishable (from the point of view of $P^*$) of an actual execution of protocol $\tilde{\Pi}'$. Furthermore, since $P'$ repeatedly executes $P^*$ until the first repetition of the protocol accepts, the probability that $P'$ succeeds in its interaction with $V_2$ is identical to the probability that  $P^*$ succeeds in the second repetition in $\tilde{\Pi}'$, conditioned on the first repetition succeeding. This probability is exactly $p_2$, which according to the discussion above is non-negligibly larger than $s$. This is a contradiction and concludes the proof. 

%
\end{proof}

Lemma \ref{lem:sound-ampl}, which we just proved, establishes soundness amplification for the parallel repetition of the protocol from Section~\ref{sec:prot-decs}. Next, we prove that such parallel repetition preserves the computational zero-knowledge property. We will later prove in Section \ref{sec:zk property} that the protocol from Section~\ref{sec:prot-decs} is computational zero-knowledge.

\begin{lem}\label{lem:zk-ampl}
Let $(Q,\alpha,\beta)$ be a $\QMA$ relation, and let $\Pi=(S,P,V)$ be the argument system for $(Q,\alpha,\beta)$ which is described in Section~\ref{sec: protocol}. Suppose that $\Pi$ is computationally zero-knowledge with respect to any quantum polynomial time verifier with (quantum) auxiliary input.\footnote{This will be proven in section \ref{sec:zk property}.} Then $\Pi^{(k)} = (S^{(k)}, P^{(k)}, V^{(k)})$ is also computationally zero-knowledge with respect to any quantum polynomial time verifier with quantum auxiliary input.
\end{lem}

\begin{proof}
Consider a single execution of $\Pi^{(k)}$, in which $V^{*(k)}$, a potentially cheating verifier, outputs $k$ quantum messages, and receives $k$ proofs in return:

\begin{figure}[H]
\begin{center}
\includegraphics{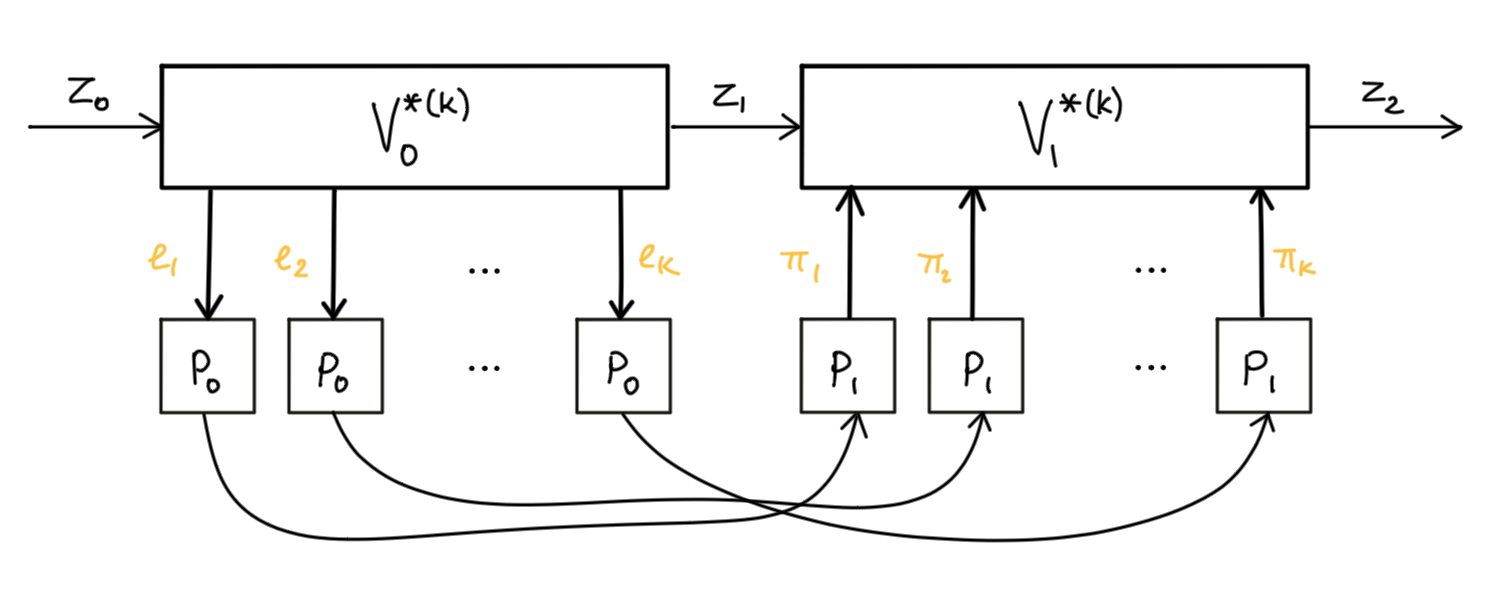}
\caption{The first stage of the cheating verifier, $V_0^{*(k)}$, takes in an auxiliary input $Z_0$ and outputs $k$ (quantum) setup messages $e_1, \ldots, e_k$, each of which goes into one copy of $P_0$. $V_0^{*(k)}$ then communicates with $V_1^{*(k)}$, the second stage of the cheating verifier, by sending an arbitrary quantum state $Z_1$. $V_1^{*(k)}$ takes proofs $\pi_1, \ldots, \pi_k$ from the $k$ copies of $P_1$, along with $Z_1$ from $V_0^{*(k)}$, and outputs an arbitrary quantum state $Z_2$.}
\label{fig:ampl-1}
\end{center}
\end{figure}

Notice that $Z_1$ is entirely independent of the behaviour of $P^{(k)}$. As such, in order to guarantee that $Z_2$ remains computationally indistinguishable from what it would be in a real execution of the protocol, we need only simulate the following: 

\begin{figure}[H]
\begin{center}
\includegraphics{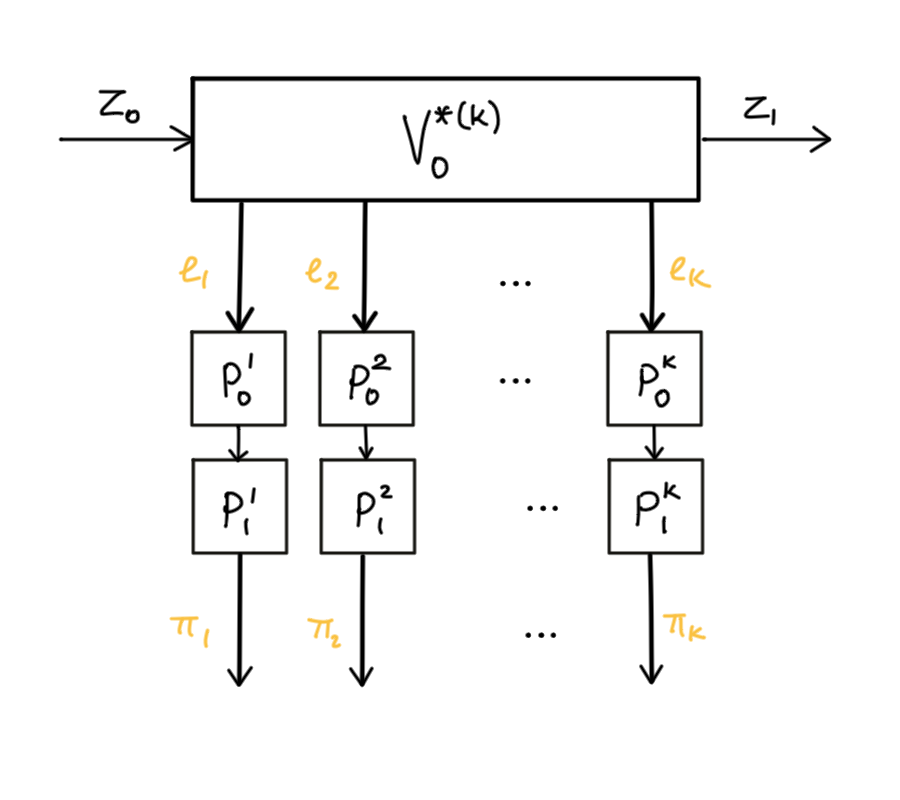}
\caption{If we can create a simulation which takes in $Z_0$ and outputs $\pi'_1, \ldots, \pi'_k$ which are indistinguishable from $\pi_1, \ldots, \pi_k$ that arise from a real execution of the protocol, then we can also simulate Figure \ref{fig:ampl-1}.}
\label{fig:ampl-2}
\end{center}
\end{figure}

That Figure \ref{fig:ampl-2} is efficiently simulable follows directly from the fact that $\Pi$ is zero-knowledge. The simulator $S_{\Pi^{(k)}}$ would run $V_0^{*(k)}$ on $Z_0$ to obtain $Z_1, e_1, \dots, e_k$, and then replace each of the $k$ pairs $(P_0^i, P_1^i), i \in \{1, \dots, k\}$ in the diagram above with a simulator $S^i$ (which is guaranteed to exist by the auxiliary-input zero-knowledge property of $\Pi$).

\begin{remark}
Note that it is in general not possible to extend this simple argument to a multi-round protocol. The intuitive reason is that, in a multi-round protocol, the parallel instances of the protocol may depend on each other, because the cheating verifier may introduce cross-correlations. More formally, the assumption we made that the internal communications of the cheating verifier do not depend on the behaviour of $P^{(k)}$ is false if the protocol is multi-round.
\end{remark}

\end{proof}

\section{Zero-knowledge property}

\label{sec:zk property}

\begin{lem}
Assume that LWE is intractable for quantum polynomial-time (QPT) adversaries. Let $L$ be a language in \QMA, let $x \in \{0,1\}^*$ be a problem instance, and let $V^* = \{V^*_{\mu n}\}$ be an arbitrary QPT verifier for the protocol of Section \ref{sec: protocol}. There exists a QPT simulator $S = \{S_{\mu n}\}$ such that, for any $\mu,n$ and yes-instance $x \in L$ with $|x|=n$, and for any auxiliary quantum input $Z_0$ to the verifier, the distribution of $V^*$'s final output after its interaction with the honest prover $P$ in the protocol is quantum computationally indistinguishable from $S$'s output distribution on auxiliary input $Z_0$.

Furthermore, the simulator $S$ only requires knowledge of the instance $x$ after the preprocessing phase has been executed (simulated) with $V^*$. As such, the zero-knowledge property holds in the adaptive setting.
\end{lem}

We show that our protocol is (adaptively) zero-knowledge by replacing the components of the honest prover with components of a simulator one at a time, and demonstrating that, when the input is a yes-instance, the dishonest verifier's output after each replacement is made is at the least computationally indistinguishable from its output before. The end result, after all replacements have been made, is an efficient simulation which does not require access to the witness $\rho$. (Our simulator is what is known as a \emph{straight-line} simulator, meaning that the simulator never needs to rewind the dishonest verifier.) The proof follows an outline that is similar to that of the proof of zero-knowledge in \cite{broadbent2016zero}. We emphasise those aspects of our proof which do not appear in the proof of zero-knowledge in \cite{broadbent2016zero}, and refer the reader to \cite[Section 5]{broadbent2016zero} for a more detailed exposition of the steps that the two proofs share.

\label{sec: zk property}
\subsection{The original protocol}

\begin{figure}[H]
\includegraphics[width=\textwidth]{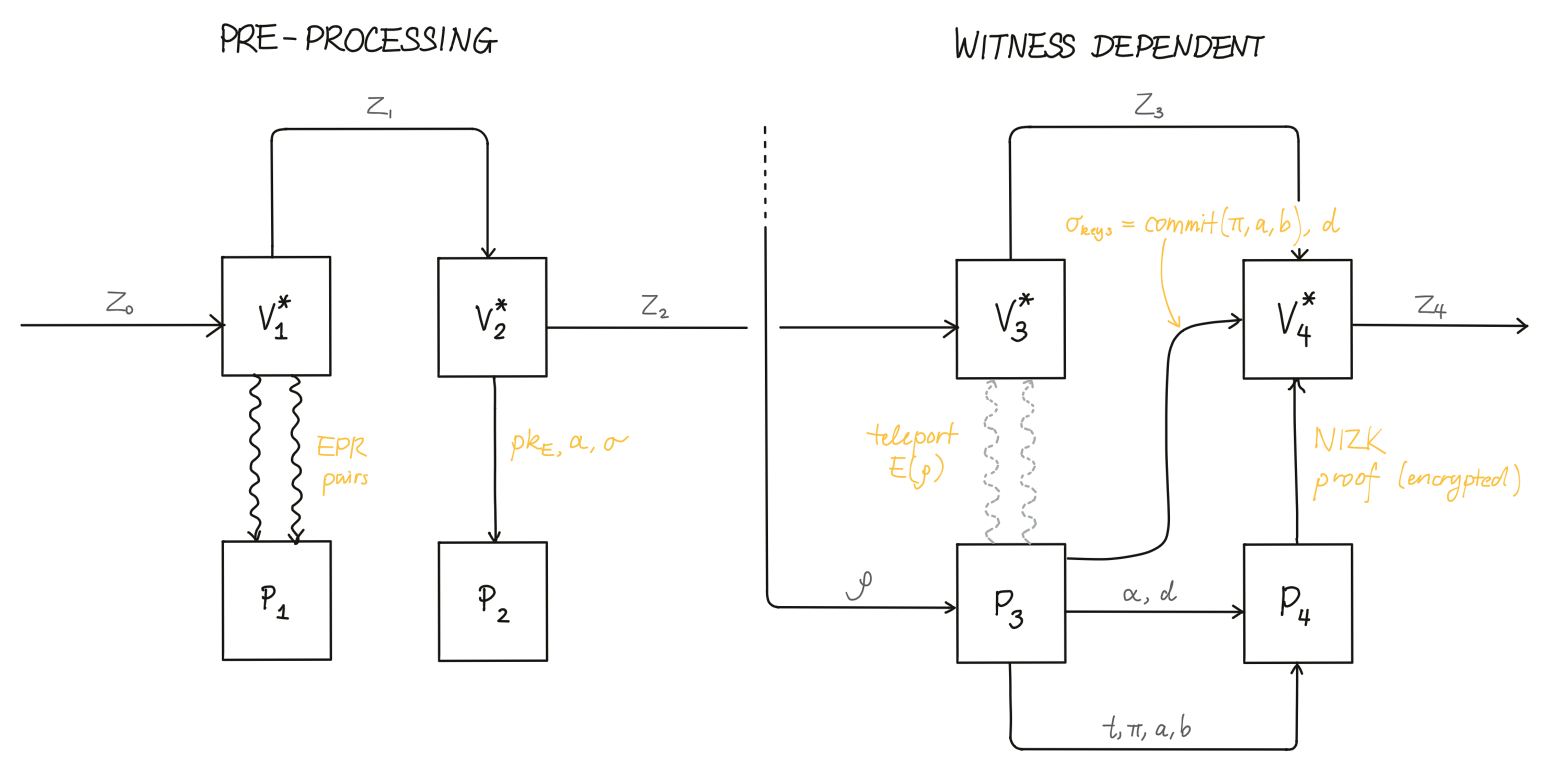}
\caption{Diagram representing the original protocol execution between the honest prover $P$ and a cheating verifier $V^*$. For visual clarity, the prover and the (cheating) verifier have been split into parts $\{P_i\}$ and $\{V^*_i\}$ with $i \in \{1, 2, 3, 4\}$, respectively, where parts 1 and 2 execute the preprocessing phase of the protocol, and parts 3 and 4 execute the instance-dependent phase of the protocol. Communications between verifier and prover are labelled in orange; internal communications on either side are labelled in grey. In the two subsequent diagrams, we will omit the auxiliary input $Z_0$ that the cheating verifier receives, as well as the internal communications $Z_1, Z_2, Z_3$ between the different parts of the cheating verifier. (Note, however, that the auxiliary input $Z_0$ and these internal communications make a reappearance in Figure \ref{fig:simplification-1}.) In the two subsequent diagrams, we will highlight replacements in green.}

\label{fig:original}

\end{figure}

\subsection{Replacing the NP proof and the commitment}

\begin{figure}[H]
\includegraphics[width=\textwidth]{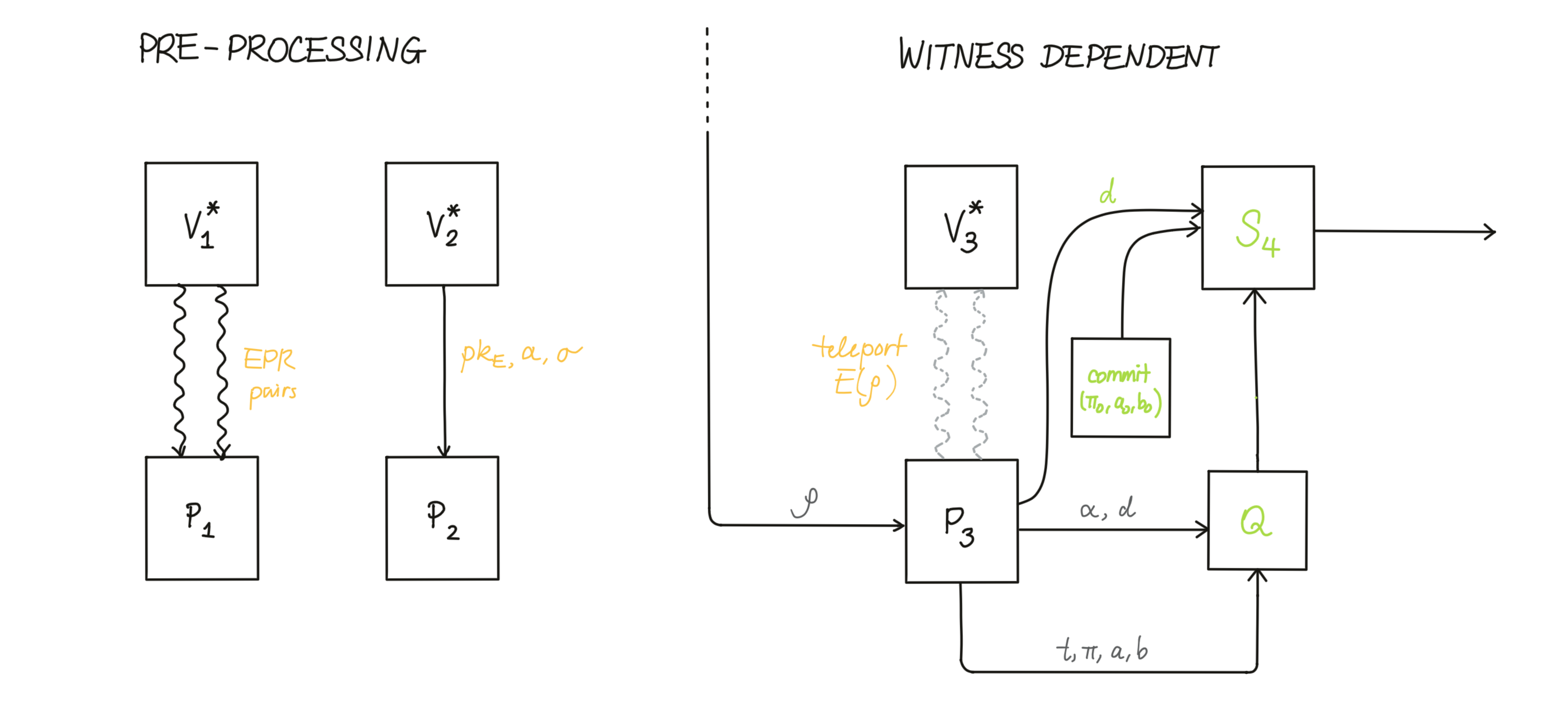}
\caption{Figure \ref{fig:original} after the NIZK NP proof has been replaced by a simulation, and the commitment $z$ has been replaced by a commitment to a fixed string.}

\label{fig:replacement-1}

\end{figure}

\label{sec:s_4'}

We replace the interaction between $P_4$ and $V_4^*$ with the entities $Q$ and $S_4$. (Recall that $P_4$ is the part of the prover which executes steps 4 and 5 in Figure \ref{fig: instance-dependent step}, and that $V_4^*$ is the part of the verifier which is meant to perform the `verifier's check' described in the same figure.) $Q$ is an entity which evaluates the predicate $Q$ (see Definition \ref{def:new-Q} for the definition of the predicate $Q$) homomorphically, and outputs an encryption of the result, which is either 0 or 1. $S_4$ takes in $Q$'s output, and homomorphically does one of two things. If the output of $Q$ is (homomorphically) 0, $S_4$ runs $V_4^*$ on an encryption of the abort symbol ``$\perp$" (as well as the other inputs which $V_4^*$ expects, namely $d, \sigma_{keys}$, and $Z_3$). If the output of $Q$ is homomorphically 1, then $S_4$ homomorphically runs the simulation $Sim_{\mu n}$ from Lemma \ref{lem:nizk-np} that is guaranteed to exist by the zero-knowledge property of the NIZK NP proof system that $V_4^*$ and $P_4$ use. More precisely: in the (incompletely simulated) version of the interaction in which $P_4$ and $V_4^*$ are replaced by $S_4$, the following occurs.

\begin{enumerate}
\item The first stage \footnote{We refer to the stage of $Sim$ which produces a CRS as the `first stage', and the stage which produces a simulated proof $\pi'$ as the `second stage'.} of $Sim$ is run in order to obtain a CRS $\gamma_P$ for the prover's NIZK NP proof. The rest of the CRS (i.e. $\gamma_V, pk_V, pk_P$) is chosen as before.
\item $V_1^*$, $P_1$, $V_2^*$, $P_2$, $V_3^*$ and $P_3$ are executed normally.
\item $S_4$ is executed (instead of $V_4^*$ and $P_4$). $S_4$ homomorphically evaluates a circuit $C'$ which, firstly, computes $Q$, and secondly does one of the following, depending on whether the output of $Q$ is 0 or 1:
\begin{enumerate}
\item If $Q = 0$, $C'$ outputs $\perp$.
\item If $Q = 1$, $C'$ runs the second stage of $Sim$. Because the NP proof system is non-interactive, the second stage of $Sim$ (consisting of steps 2 and 3 from Lemma \ref{lem:nizk-np}) is particularly simple: the only reason $Sim$ requires access to $V^*$ is to run it forward, once, in order to produce an instance $x$, from which (and no other input) it can produce a simulated proof $\pi$. In this case, $S_4$ wishes to use $Sim$ to produce a simulated proof for the NP statement `There exists a commitment string $s_P$ and an encoding key $(t, \pi, a, b)$ such that $\sigma_{keys} = \textsf{Com.commit}(pk_P, (\pi, a, b), s_P)$ and $Q(t, \pi, a, b, r, z, d) = 1$.' Toward this end, we let $x$ be some bitstring representation of this NP statement. From $P_3$, $S_4$ receives $(d, \sigma_{keys})$, which together define $x$. $S_4$ can then run $Sim$ homomorphically (using the ability which \textsf{FHE.Eval} has to evaluate \emph{classical} circuits) on instance $x$, and obtain from it an encryption of a bogus (simulated) proof $\pi'$.
\end{enumerate}

We shall refer to the ciphertext which $S_4$ obtains from either (a) or (b) above as $c$.

\item $S_4$ then runs $V_4^*$ with $\mathsf{FHE.Refresh}(pk_E, c)$ as input (as well as the other inputs which $V_4^*$ expects, namely $d, \sigma_{keys}$, and $Z_3$). (See section \ref{sec:hom-enc} for a definition of $\mathsf{FHE.Refresh}$ and its properties.)
\end{enumerate}

We argue that $V_4^*$'s output after this replacement occurs is computationally indistinguishable from its output before.

In the original protocol, $P_4$ homomorphically evaluates the predicate $Q$ to check that it is satisfied, and if it is, $P_4$ provides an encryption of a zero-knowledge proof to $V_4^*$ of the NP statement `There exists a commitment string $s_P$ and an encoding key $(t, \pi, a, b)$ such that $\sigma_{keys} = \textsf{Com.commit}(pk, (\pi, a, b), s_P)$ and $Q(t, \pi, a, b, r, z, d) = 1$.' Let $S$ be the predicate that is 1 if and only if the NP statement `There exists a commitment string $s_P$...' is true. Then, because $P_4$ is honest, $S$ is satisfied whenever $Q$ is satisfied; and, by definition, $Q$ is satisfied whenever $S$ is satisfied. Therefore, if the simulation $S_4$ aborts if and only if $Q$ is not satisfied, $S_4$ will abort exactly when $P_4$ would have aborted.

When $Q$ is satisfied, meanwhile, we can argue via a series of hybrids that the verifier $V_4^*$'s output is computationally indistinguishable from its output in the original protocol. Let $C$ be the circuit described in step 4 of Figure \ref{fig: instance-dependent step}, which evaluates $Q$ and outputs a real proof $\pi$ if $Q$ is 1; and let $C'$ be the circuit evaluated homomorphically by $S_4$, which evaluates $Q$ and outputs a simulated proof $\pi'$ if $Q$ is 1. Recall that $\alpha$ is the ciphertext which $V_2^*$ provides to $P_2$. It is this ciphertext $\alpha$ on which $P_4$ homomorphically evaluates $C$, and on which $S_4$ homomorphically evaluates $C'$.

\begin{enumerate}
\item In the real protocol, $V_4^*$ receives $\textsf{FHE.Refresh}(pk_E, \textsf{FHE.Eval}(pk_E, C, \alpha))$ from the prover $P_4$. (It also receives its other inputs $d, \sigma_{keys}$ and $Z_3$ from various sources. Since we will not alter these other inputs, we omit them in descriptions of subsequent hybrids.) We call this original experiment $H_0$.
\item $H_1$: $V_4^*$ receives $\textsf{Enc}^*\big(pk_E, C\big(\textsf{Dec}^*(pk_E, \alpha)\big)\big)$. (See `Circuit privacy' under section \ref{sec:hom-enc} for a definition of $\textsf{Enc}^*$.) By the malicious circuit privacy property of \textsf{FHE}, $H_0$ and $H_1$ are statistically indistinguishable, and so $V_4^*$'s output distribution differs only negligibly between the two hybrids.
\item $H_2$: $V_4^*$ receives
\begin{equation*}
\textsf{FHE.Refresh}\Big( \textsf{FHE.Eval}\big(pk_E, I, \: C(\textsf{Dec}^*(pk_E, \alpha)) \: \big) \Big)
\end{equation*}
where $I$ is the identity circuit. $H_2$ and $H_1$ are again indistinguishable by malicious circuit privacy.
\item $H_3$: $V_4^*$ receives
\begin{equation*}
\textsf{FHE.Refresh}\Big( \textsf{FHE.Eval}\big(pk_E, I, \: C'(\textsf{Dec}^*(pk_E, \alpha)) \: \big) \Big).
\end{equation*}
(The only difference between $H_3$ and $H_2$ is that $C$ has been replaced with $C'$.)

$H_3$ and $H_2$ are computationally indistinguishable from $V_4^*$'s point of view, because $C'(\textsf{Dec}^*(pk_E, \alpha))$ is computationally indistinguishable from $C(\textsf{Dec}^*(pk_E, \alpha))$ (this follows from the fact that $S$ is satisfied because we have by hypothesis in this section that $Q$ is satisfied, and that $\pi' = Sim(x)$ is computationally indistinguishable from a real proof $\pi$ for any yes-instance $x=\braket{S}$ of the verifier's choice, by the adaptive zero-knowledge property of Lemma \ref{lem:nizk-np}), and \textsf{FHE.Refresh} and \textsf{FHE.Eval} are both PPT algorithms.
\item $H_4$: $\textsf{Enc}^*\big(pk_E, C'\big(\textsf{Dec}^*(pk_E, \alpha)\big)\big)$. $H_4$ and $H_3$ are statistically indistinguishable because of malicious circuit privacy.
\item $H_5$: $\textsf{FHE.Refresh}(pk_E, \textsf{FHE.Eval}(pk_E, C', \alpha))$. $H_5$ and $H_4$ are again statistically indistinguishable by malicious circuit privacy. Note that $H_5$ is exactly the experiment that $S_4$ executes.
\end{enumerate}

We have thus argued that the interaction between $P_4$ and $V_4^*$ in Figure \ref{fig:original} can be replaced with $S_4$. Having made this first replacement (i.e. turning $P_4$ and $V_4^*$ into $Q$ and $S_4$), we can then replace the commitment $\sigma_{keys}$ with a commitment to a fixed encoding key $(\pi_0, a_0, b_0)$ that is independent of the problem instance. The reasoning that justifies this replacement is the same as in \cite[Section 5, step 3]{broadbent2016zero}: because the commitment is computationally hiding, and because it is never opened after the NP proof has been replaced by a simulation, the real commitment $\sigma_{keys}$ is computationally indistinguishable from $\sigma_{keys}' = \textsf{Com.commit}(pk_P, (\pi_0, a_0, b_0), s_P)$ to all of the entities in the diagram above, given that all of said entities are quantum polynomial-time. As such, after both replacements have been made, the output of $S_4$ is computationally indistinguishable from the output of $V_4^*$ in the original protocol.

\subsection{Replacing $\rho$ with $\rho_r$}
\label{section:replacing-rho}

We now claim that, if we replace the real witness $\rho$ with a simulated witness $\rho_r$ in Figure \ref{fig:replacement-1}, $S_4$'s output is computationally indistinguishable from what it was in Figure \ref{fig:replacement-1}. (This claim holds regardless of the verifier's choice of $r$.) The simulated witness $\rho_r$ is a state that, given $r$, can be efficiently constructed.

For the remainder of section \ref{section:replacing-rho}, we assume that $r$ is chosen and published by the verifier before the witness-dependent part of the protocol begins, so that $r$ is fixed and known to all entities which make use of it (such as $Q$). Although this is, in the real world, untrue---in the real protocol, the verifier only sends a \emph{commitment} to $r$, which hides the true value of $r$ from the prover---we will show in section \ref{section:feasibility} that the simulator can \emph{extract} $r$ from the verifier's perfectly binding commitment, so that the effect is the same (for the simulator, not the real prover) as if the verifier had published $r$.

\subsubsection{Pauli-twirling the verifier}
Recall that $P_3$ is the part of the prover which receives the witness $\rho$, applies an encoding to $\rho$, and then teleports $E(\rho)$ to the verifier. In order to facilitate the analysis, we firstly separate the action of $P_3$ from Figure \ref{fig:replacement-1} into two parts: we make applying the encoding $E$ the responsibility of a separate entity labelled $E$, and delegate the remainder of $P_3$'s actions to a new entity $P_3'$. Since $P_3$ was honest, this replacement results in no change in the output of $S_4$.

\begin{figure}[H]
\includegraphics[width=\textwidth]{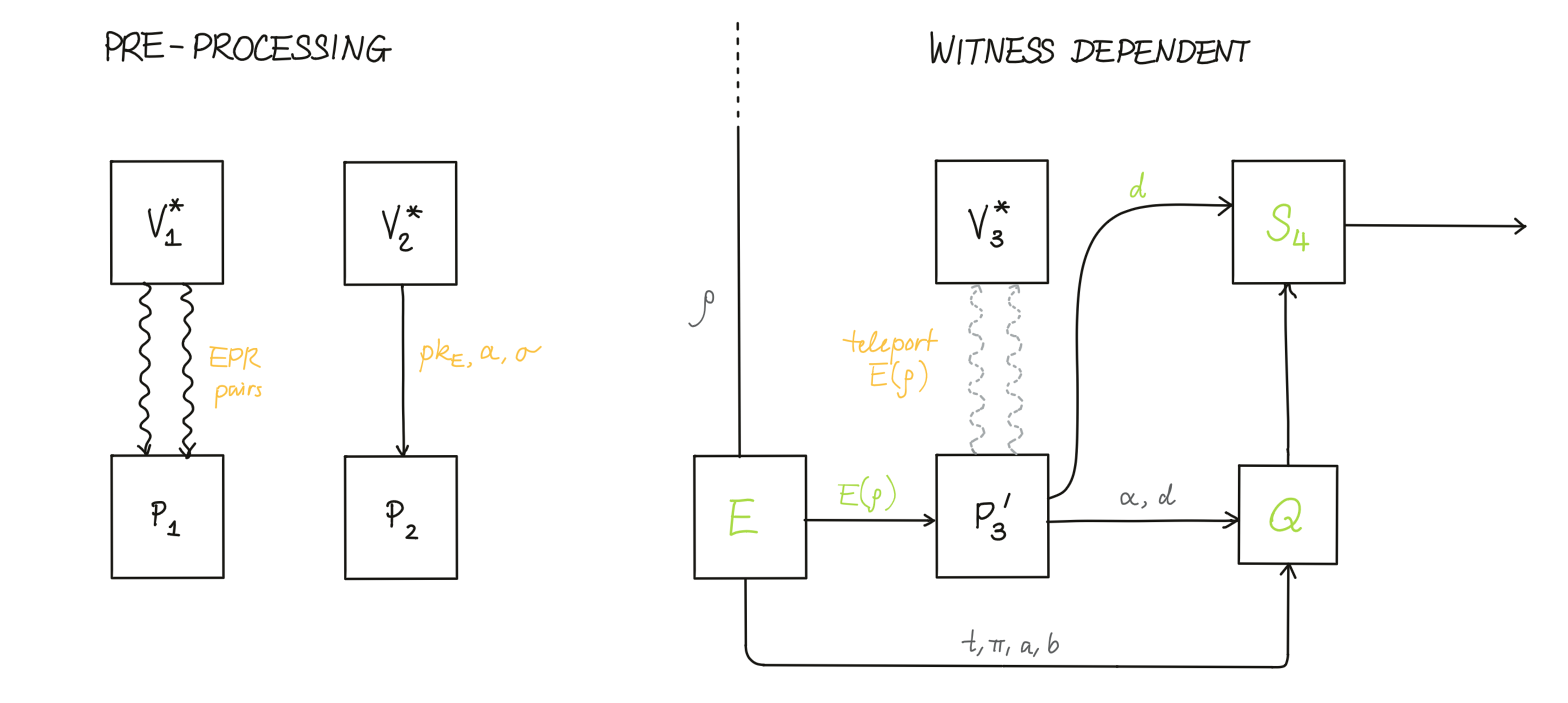}
\caption{$E$ represents a procedure that samples an encoding key from the same distribution from which $P_3$ would have sampled its encoding key in Figure \ref{fig:replacement-1}, applies the encoding represented by that key to $\rho$, and sends the encoded $\rho$ to $P_3$ and the key it chose to $Q$. Because $P_3$ was honest, the output distribution of $S_4$ in this diagram is exactly the same as the output distribution of $S_4$ in Figure \ref{fig:replacement-1}.}

\label{fig:replacement-2}

\end{figure}

Now, we consider $(V^*_1, P_1), (V^*_2, P_2), (V^*_3, P_3')$ taken together to be a cheating verifier $V'$:
\begin{figure}[H]
\includegraphics[width=\textwidth]{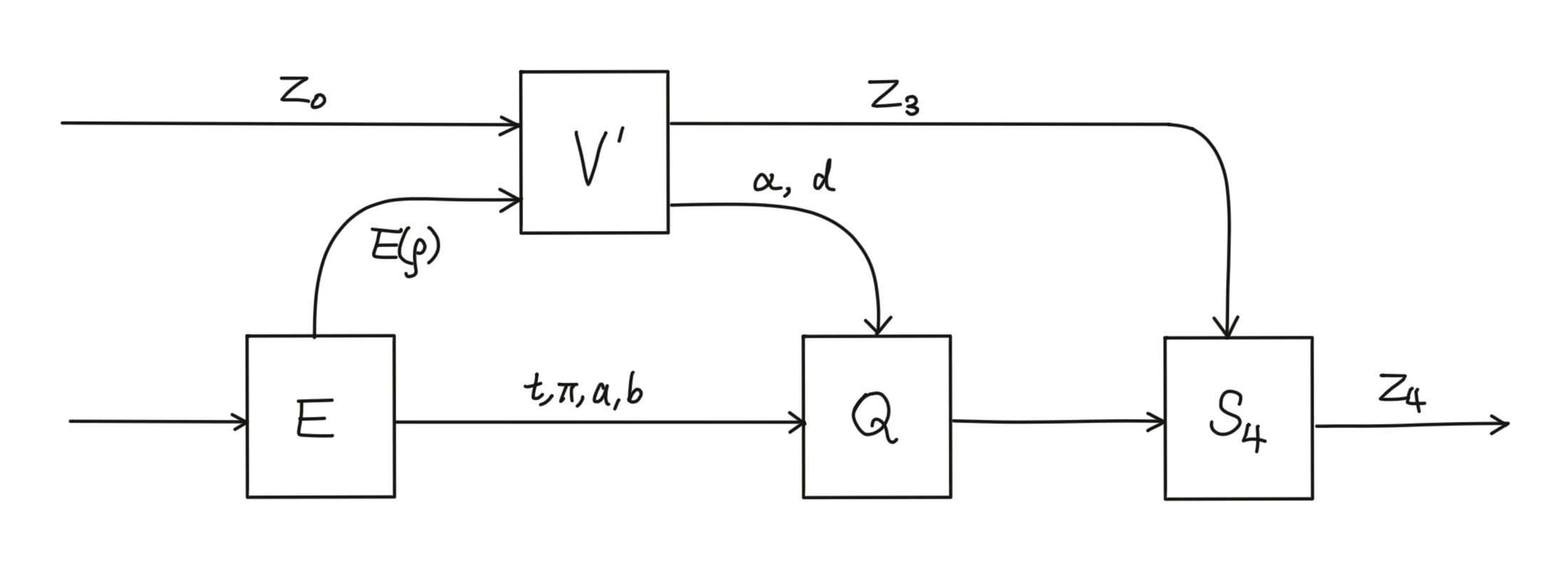}
\caption{$V'$ receives the encoded witness, $E(\rho)$, along with auxiliary information $Z_0$. It sends some arbitrary auxiliary $Z_3$ to $S_4$, and also sends $d = (x,y)$ and the ciphertext $\alpha$ (which, in the honest verifier's case, contains the homomorphic encryptions of $(r, s_V, z)$) to $Q$. $Q$ receives its other inputs from $E$.}

\label{fig:simplification-1}

\end{figure}

\begin{remark}
Although the string $d$ in Figure \ref{fig:simplification-1} now appears to be originating from what we have called the cheating verifier $V'$, we know that it is honestly produced, because we know that $P_3'$ is honest. We include $P_3'$ as part of the cheating verifier because we would like to consider the ciphertext $\alpha$ a potentially dishonest input to $Q$.
\end{remark}

We will now make a series of simplifications to Figure \ref{fig:simplification-1} such that, after each simplification, the distribution of $Z_4$ is statistically indistinguishable from what it was before the simplification. Our first step is to visually reorganise a number of the features of Figure \ref{fig:simplification-1}.
\begin{enumerate}
\item We split the process of applying the encoding $E$ into two stages. The first stage, which we call $E_1$, comprises steps (a) and (b) in Figure \ref{fig: authentication code}. The second stage is step (c) from Figure \ref{fig: authentication code}, which consists of the application of the Pauli one-time-pad $X^a Z^b$.
\item We explicitly show the application of $X^a Z^b$, as well as the application of the additional one-time-pad which arises from the prover's teleportation measurements that resulted in $d$. (This latter is $X^{a'}Z^{b'}$, using the $a', b'$ notation from Definition \ref{def:new-Q}.) These two one-time-pads taken together result in a total one-time-pad of $X^{a \oplus a'} Z^{b \oplus b'}$ which is applied to the state $E_1(\rho)$.
\item We let $e, f$ be the unique strings such that $X^{a \oplus a'} Z^{b \oplus b'} = C_r X^e Z^f C_r^*$, and we rewrite $X^{a \oplus a'} Z^{b \oplus b'}$ as $C_r X^e Z^f C_r^*$. Recall that evaluating $Q$ on $(t, \pi, a, b, r, z, d)$ is equivalent to evaluating $R_r(t, \pi, z \oplus e_{i_1} \dots e_{i_k})$. We make this explicit by rewriting $Q$ as $X^e$ followed by $R_r$.
\item We separate $S_4$ into two entities, one of which homomorphically evaluates the circuit $C'$ which is defined in section \ref{sec:s_4'} (we call this entity $S_4'$ to avoid confusion with the Cliffords $C_r$), and the other of which evaluates \textsf{FHE.Refresh} on the ciphertext which results from homomorphically evaluating $C'$.
\end{enumerate}

Redrawing Figure \ref{fig:simplification-1} according to steps 1, 2, 3 and 4 above (and omitting the encoding keys $(t, \pi, a, b)$ that the honest $E$ sends to $Q$) results in Figure \ref{fig:simplification-2}:

\begin{figure}[H]
\includegraphics[width=\textwidth]{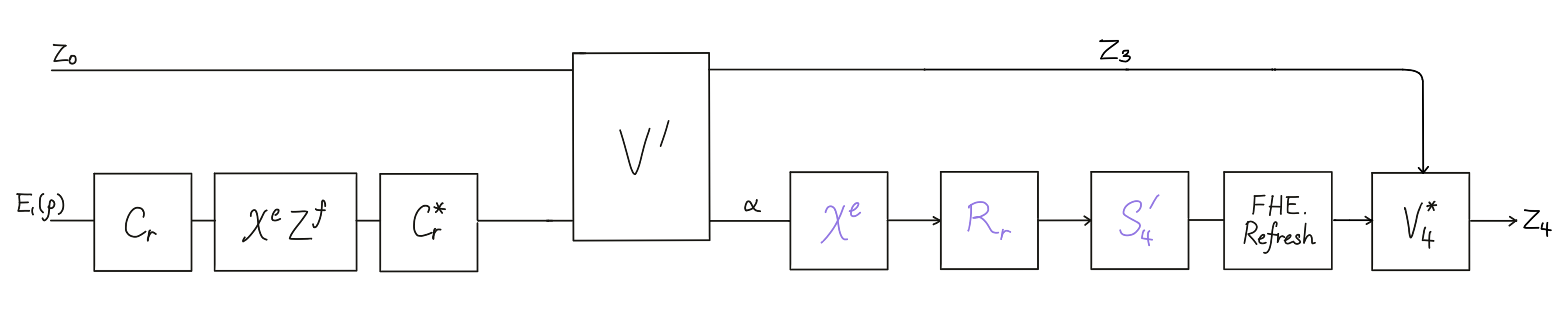}
\caption{Operations shown in purple are operations which are applied \emph{homomorphically} to an encrypted input (in this case, the encrypted input is $\alpha$).}

\label{fig:simplification-2}

\end{figure}

Using the malicious circuit privacy property of \textsf{FHE} which is described in section \ref{sec:hom-enc}, the output of $V^*_4$ in Figure \ref{fig:simp-3} is statistically indistinguishable from its output in Figure \ref{fig:simplification-2}:

\begin{figure}[H]
\includegraphics[width=\textwidth]{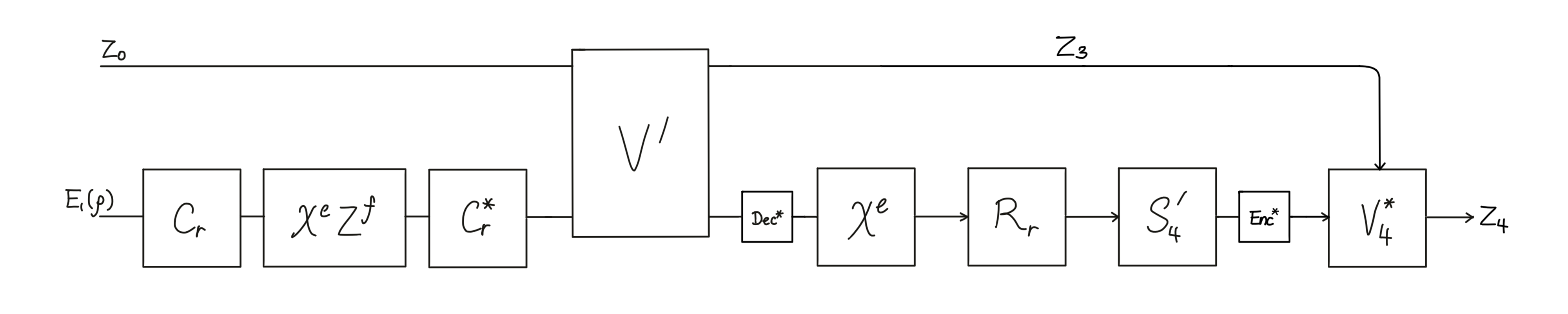}
\caption{Note that the applications of $X^e$, $R_r$ and $S_4'$ are no longer homomorphic.}

\label{fig:simp-3}

\end{figure}

Following \cite{broadbent2016zero}, we merge the operations $C_r^*$ and $\textsf{Dec}^*$ into $V'$, and apply a Pauli twirl in order to limit the action of $V'$ to an XOR attack which does not depend on the state in the lower wire.

\begin{remark}

The reader may observe that there is an $X^e Z^f$ on the left-hand-side of $V'$ and only an $X^e$ on the right-hand-side. We can add a $Z^f$ to the right-hand-side without consequence, since the state passing out of $V'$ on the lower wire is treated as a classical string by all entities which follow it. Adding a $Z^f$ allows us to perform a Pauli twirl on $V'$.

\end{remark}

\begin{figure}[H]
\includegraphics[width=\textwidth]{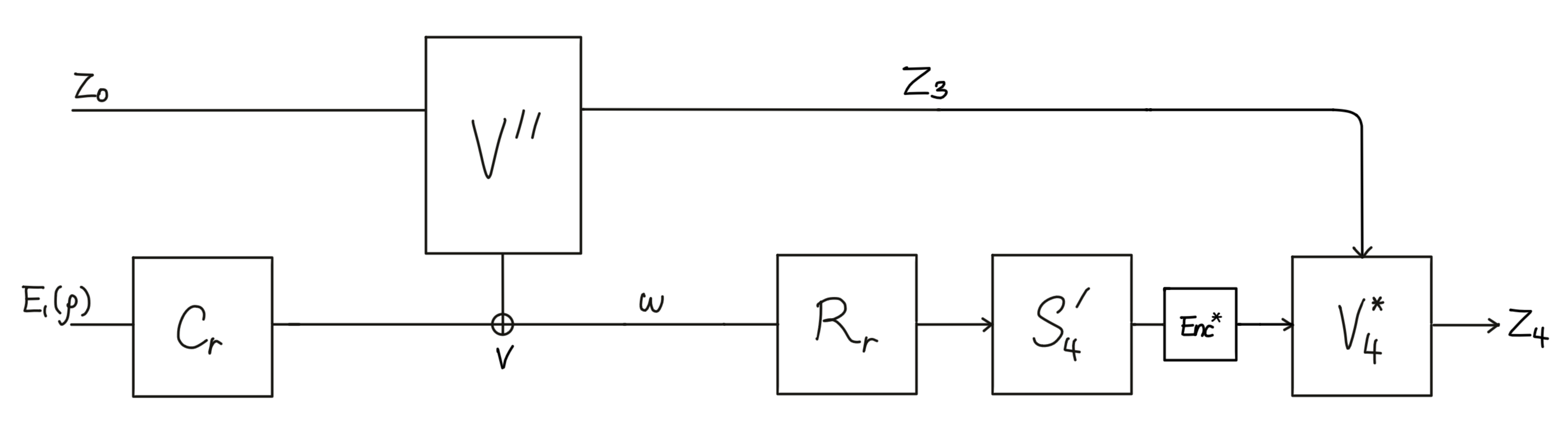}
\caption{Using the Pauli twirl, we have replaced the cheating verifier $V'$ with the box $V''$, shown on the left. $V''$ acts only by XOR attack on the lower wire.}

\label{fig:simp-4}

\end{figure}

We began section \ref{section:replacing-rho} with the stated intention of showing that $V_4^*$'s output in Figure \ref{fig:replacement-1} when $\rho$ is provided to the prover is statistically indistinguishable from its output when $\rho_r$ is provided to the prover, where $\rho_r$ is a state designed particularly to pass the challenge determined by the verifier's choice of $r$ with probability 1. Since Figure \ref{fig:simp-4} is statistically indistinguishable from (a subset of) Figure \ref{fig:simplification-1}, which is in turn equivalent to Figure \ref{fig:replacement-1}, it suffices to show that the output of $\mathsf{Enc}^*$ in Figure \ref{fig:simp-4} when $E_1(\rho_r)$ is provided as input on the lower wire is (statistically) indistinguishable from its output when $E_1(\rho)$ is provided as input. This statement is implied by the statement that the output of $R_r$ in Figure \ref{fig:simp-4} when $E_1(\rho_r)$ is provided as input is (statistically) indistinguishable from its output when $E_1(\rho)$ is provided as input. If we can show the latter, then we can substitute Figure \ref{fig:simp-4} back into Figure \ref{fig:simplification-1}, and conclude that---by the statistical indistinguishability which the preceding argument has proven---the final output $Z_4$ in Figure \ref{fig:simplification-1} when $\rho_r$ is used is statistically indistinguishable from its output when the real witness $\rho$ is used.

As such, it suffices to show that the output of $R_r$ in Figure \ref{fig:bjsw-11} when $\rho_r$ is given as input is (statistically) indistinguishable from its output when $\rho$ is given as input.

\begin{figure}[H]
\includegraphics[width=\textwidth]{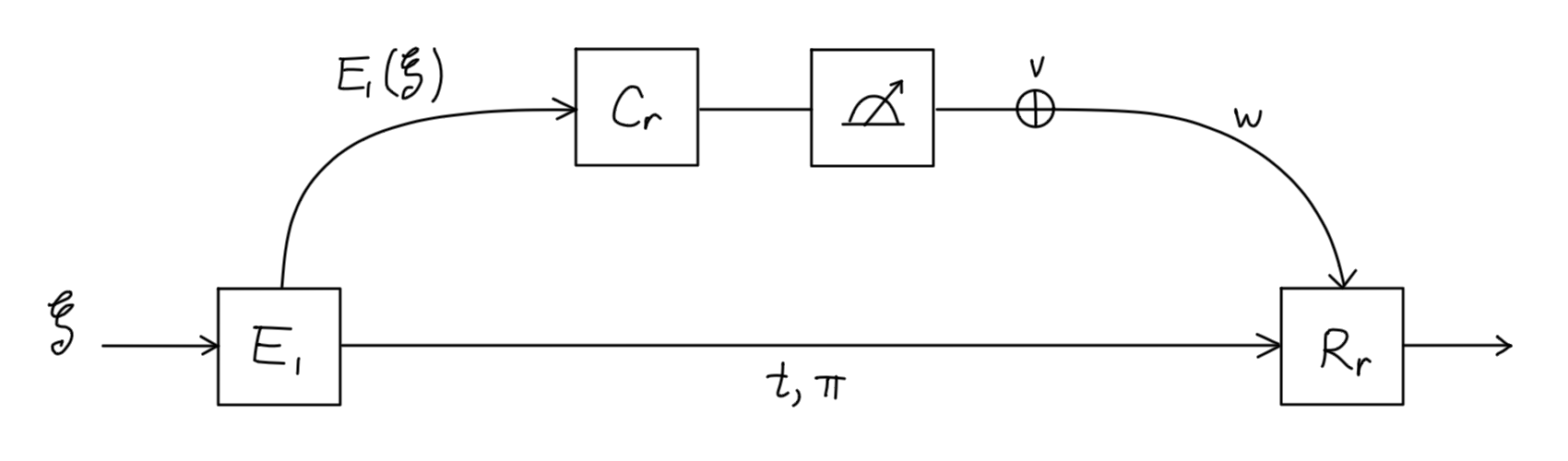}
\caption{It is sufficient to argue that the output of $R_r$ when $\xi = \rho_r$ and when $\xi = \rho$ are statistically indistinguishable. Note that the upper wire in this diagram is the \emph{lower} wire in Figure \ref{fig:simp-4}.}
\label{fig:bjsw-11}
\end{figure}

\subsubsection{Reducing to \cite{broadbent2016zero}}

Figure \ref{fig:bjsw-11} is identical to Figure 11 from \cite{broadbent2016zero}. As such, we can apply step 6 of their zero-knowledge proof, almost in a black-box fashion, in order to argue that the output of $R_r$ when $\xi = \rho_r$ is statistically indistinguishable from its output when $\xi = \rho$.

One difference between their situation and ours of which we must be mindful is that our Hamiltonian (effectively) includes one extra term designed to check the registers of the witness state in which the prover is supposed to put a classical description of the instance. When any other term in our Hamiltonian is selected according to $r$, our predicate $R_r$ can be defined as $R_r(t, \pi, u) = \tilde{R}_r(t, \pi, u)$, where $\tilde{R}_r$ is the predicate known as $R_r$ in \cite{broadbent2016zero}. (Note that, in Figure \ref{fig:bjsw-11}, the third argument of $R_r$ is called $w$. This is to distinguish the potentially dishonest $w$ from the honest measurement outcomes $u$.) As such, in these cases, the argument in step 6 of \cite[Section 5]{broadbent2016zero} applies verbatim. In the case when the special instance-check term is selected, the translation between our $R_r$ and theirs is not so direct, but the substance of their argument is not affected. It is still true that, in order to change the \emph{logical} meaning of $u$, the verifier must apply an XOR attack of large Hamming weight, with the result that the probability its XOR attack goes undetected by the traps $t$ becomes negligibly small. As such, even though the specific form of the correctness check condition (condition (a) in Definition \ref{def:bjsw-R_r} and condition 1. in Definition \ref{def:new-R_r}) is different when the instance-check term from our Hamiltonian is selected, the argument made in step 6 of \cite{broadbent2016zero} still holds.

It follows from the argument in \cite{broadbent2016zero}, then, that when $\rho$ in Figure \ref{fig:replacement-2} is replaced by $\rho_r$, where $\rho_r$ is a state that passes the particular challenge determined by $r$ with probability 1, the output of $S_4$ is computationally indistinguishable from what it was before the replacement occurred. $\rho_r$ in our case can be constructed efficiently, just as it can be in the \cite{broadbent2016zero} protocol: for those $r$s that do not correspond with the instance check term, $\rho_r$ is the same that it would be in the \cite{broadbent2016zero} protocol, and for those $r$s that do correspond with the instance check term, a state which had the correct instance description in the relevant registers and the state $\ket{0}$ on all other qubits would pass with probability 1.

\subsection{Feasibility of the $\rho \rightarrow \rho_r$ replacement}

\label{section:feasibility}

It remains only to argue that, as we claimed at the start of section \ref{section:replacing-rho}, the simulator can efficiently extract $r$ from the verifier's commitment to $r$. (We assume the latter to be perfectly binding, so that, conditioned on the prover accepting when it checks the verifier's commitment homomorphically, $r$ is well-defined.) The simulator does this by taking advantage of the CRS. The verifier expects as input a CRS which contains in it the public key $pk_V$ that it is to use for the commitment scheme \textsf{Com}. We assume, however, that \textsf{Com} is \emph{extractable} (see section \ref{sec: commitment}), meaning that every valid public key for \textsf{Com} has a corresponding secret key $sk_V$ such that the efficient algorithm \textsf{Com.recover} can, given $sk_V$ and any commitment $z = \textsf{Com.commit}(pk_V, b, s)$, recover $(b, s)$ from $z$. Then, in order to extract $r$, the simulator simply samples a public key, secret key \emph{pair} $(pk_V^*, sk_V^*) \leftarrow \textsf{Com.gen}$; inserts $pk_V^*$ into the CRS when it gives the verifier the CRS as input; and recovers $r$ from the verifier's commitment to $r$ using $\textsf{Com.recover}$ and $sk_V^*$.

The final simulation $S$ is given by Figure \ref{fig:replacement-1}, except with $\rho$ replaced with $\rho_r$. Since all the entities in this simulation are efficient, and none have access to the real witness, we have shown that the protocol described in section \ref{sec: protocol} is zero-knowledge.

\section{NIZK argument of quantum knowledge with setup for QMA}
\label{sec: aoqk}

In this section we show that for any $\QMA$ relation the NIZK argument system with CRS setup and preprocessing described in section \ref{sec: protocol} is also a NIZK Argument of Quantum Knowledge with CRS setup and preprocessing (as defined in section \ref{sec: poqk}). The intuition for this is simple. From the proof of soundness of the protocol from \cite{broadbent2016zero}, to which soundness of our argument system reduces, we are able to infer that any prover which is accepted in our protocol with high probability must be teleporting to the verifier \textit{an encoding} of a low-energy witness state for the given instance of the 5-local Clifford Hamiltonian problem. Then, all that an extractor (given oracle access to such a prover) has to do in order to output a good witness is: 
\begin{itemize}
\item Simulate an honest verifier so as to receive (by teleportation) such an encoded witness from the prover,
\item Find a way to recover the committed encoding keys and use them to decode the received state.
\end{itemize}
The extractor recovers the keys by using the same technique used multiple times in the soundness proof of section \ref{sec: soundness}, and also in the zero-knowledge proof of section \ref{sec: zk property}. Namely, the extractor samples a CRS which includes a public key for the prover's commitment scheme for which it knows the associated secret key. The distribution of such CRS is identical to the true distribution, but the additional side information allows the extractor to recover the authentication keys once the prover sends his commitment. In the rest of this section, we formalize this sketch.

The following lemma, that is a direct consequence of the soundness proof in \cite{broadbent2016zero}, will be useful in the analysis. 

\begin{lem}
\label{lem: technical bjsw}
For any $N \in \mathbb{N}$ there exists an efficiently implementable quantum channel $\Phi$ from $2N$ qubits to $1$ such that the following holds for any 5-local Clifford Hamiltonian $H$ on $N$ qubits: let $P$ be a prover for the protocol of \cite{broadbent2016zero} (where $N$ is the number of physical qubits per logical qubit in the concatenated Steane code), and suppose $P$ is accepted with probability $p$ on common input $H$. Let $\sum_{\sigma} p_\sigma \rho_\sigma \otimes \ket{\sigma}\bra{\sigma}$ be the state of the message register of $P$ after the first round, where the $\sigma$'s are commitments to authentication keys, and the $\rho_\sigma$'s are the associated (possibly invalidly) encoded witnesses. Assume $\sum_{\sigma \text{ is a valid commitment}} p_{\sigma} =1$. Let $\xi_{\sigma}$ be the state obtained from $\rho_{\sigma}$ after applying the decoding map $Dec_{\pi,a,b}$ (from Definition \ref{def: dec auth}), where $\pi,a,b$ are values committed in $\sigma$. Let $\xi = \sum_{\sigma} p_{\sigma}\xi_{\sigma}$. Then
$$\Tr[H \Phi^{\otimes n}(\xi)] < 1-p + \negl(n)\;. $$ 
\end{lem}

Consider a $\QMA$ relation $(Q, \alpha, \beta)$ with $\alpha- \beta > 1/poly$ and a language $L \in \QMA$ specified by  $(Q, \alpha, \beta)$. Given an instance $x$ for the associated decision promise problem, let $H(x)$ be an instance of the ($\QMA$-complete) 5-local Clifford Hamiltonian problem that is associated to it. We assume without loss of generality that this reduction has the property that for all $x \in L$, there exists a state $\sigma$ such that $\Tr(H(x)\sigma)< \negl(|x|)$. Moreover, there exists $\gamma = 1/\poly$ such that for all $x\notin L$, for all $\sigma$, $\Tr(H(x) \sigma)> \gamma(|x|)$. Moreover, such a reduction has the property that there exists an efficiently implementable unitary $U$ which maps a witness state for $x$, tensored with some auxiliary qubits in the state $\ket{0^t}$ for some $t$, to a witness state for $H(x)$. Let $\tau$ be the efficiently implementable map obtained from item 3 in Lemma~\ref{lem:5local}. Then by the lemma $\tau$ has the property that there exists a universal polynomial $t$ such that whenever 
 $\Tr(H(x) \sigma)\leq 1-\eps$ for some state $\sigma$ and $\eps \leq 1$, it holds that $(x,\tau(\sigma)) \in R_{Q,1-\poly((1-\eps)t)}$.
%
%

We are ready to state the main result of this section. For simplicity we only consider relations $(Q,\alpha,\beta)$ such that $\alpha = 1-\negl(n)$ and $\beta=\negl(n)$; by~\cite{marriott2005quantum} any relation can be amplified to such parameters without changing the witness for the case of a yes instance. 

\begin{theorem}
\label{thm: nizk aoqk}
Let $(Q,\alpha, \beta)$ be a $\QMA$ relation with $\alpha = 1-\negl(n)$ and $\beta=\negl(n)$. There is a polynomial $t$ such that the following holds. There is a NIZK Argument of Quantum Knowledge with CRS setup and single-message preprocessing for $(Q,\alpha,\beta)$ with completeness $1-\negl$, soundness $s$ (where $s$ is the soundness from Thereom \ref{thm:main}), knowledge error $1- \frac{1}{2t}$ and quality $q(n,\eps)=1-\poly((1-\eps) t(n))$.
\end{theorem}

\begin{proof}
We prove that the argument system with CRS setup for $(Q,\alpha, \beta)$ from section \ref{sec: protocol} is also an argument of Quantum Knowledge with CRS setup for $(Q,\alpha, \beta)$. We describe an extractor $K$. Fix an instance $x$ and a prover $P'$ that is accepted with probability $p>1-\frac{1}{2t}$ in the protocol of section \ref{sec: protocol} on common input $x$. At a high level, $K$ does the following:
\begin{enumerate}
    \item $K$ uses $P'$ to construct a prover $P''$ for the protocol of \cite{broadbent2016zero}. The crucial part is that $K$ is able to know the authentication keys that $P''$ commits to.
    \item $K$ runs $P''$ and obtains a state $\rho$ of $2nN$ qubits and a commitment $\sigma$. Let $(t,\pi,a,b)$ be the committed keys. $K$ outputs $\Phi^{\otimes n}(Dec_{\pi,a,b} (\rho))$.
\end{enumerate}
In more detail, the hybrid argument in section \ref{sec: soundness} implies that there exists an efficient $A$ such that $A^{\ket{P'}}$ is accepted with probability at least $p-\negl(|x|)$ in the protocol from hybrid $H_{10}$, i.e.\ the protocol from \cite{broadbent2016zero}. In particular, the prover reduction in the last replacement $H_{9} \approx H_{10}$ implies that we can choose $A$ such that after the first step of the instance-dependent part of the protocol the state of $A^{\ket{P'}}$ is (where the first two registers are message registers and the third register is an internal register):
\begin{equation}
\label{eq: state}
    \sum_{\sigma} p_{\sigma}\, \rho_{\sigma} \otimes \ket{\sigma} \bra{\sigma} \otimes \ket{(t,\pi,a,b)_{\sigma}}\bra{(t,\pi,a,b)_{\sigma}}\;,
\end{equation}
where $(t,\pi,a,b)_{\sigma}$ is the value committed in $\sigma$, and all $\sigma$'s that appear in the sum are valid commitments. The extractor $K$ runs $A^{\ket{P'}}$ to obtain \eqref{eq: state}, and efficiently maps this to $\Xi = \sum_{\sigma} p_{\sigma} \Phi^{\otimes n} \left(Dec_{\pi,a,b}(\rho_{\sigma})\right)$. Since $A^{\ket{P'}}$ is accepted with probability at least $p-\negl(|x|)$ in the protocol of \cite{broadbent2016zero}, Lemma \ref{lem: technical bjsw} implies that $\Xi$ satisfies $\Tr(H(x) \Xi) < 1-p + \negl(|x|)$. The final output of the extractor is $\tau(\Xi)$, where $\tau$ is the efficiently implementable channel defined at the start of this section. The claim on the quality of $K$'s output follows from the property of $\tau$ stated below Lemma~\ref{lem: technical bjsw}.

\end{proof}

\begin{remark} Notice that $K$ succeeds at extracting with probability $1$. This is perhaps surprising given that the only known zero-knowledge proofs or arguments of knowledge for $\NP$ in the quantum settting \cite{unruh2012quantum} have an extractor that only succeeds with inverse-polynomial probability. Here, our argument is non-interactive and it works for $\QMA$! The main difference, which grants the additional power to the extractor, is the presence of the CRS setup.
\end{remark}

\begin{remark}
The knowledge error in Theorem \ref{thm: nizk aoqk} is $1-1/poly$. This can reduced to negligible by sequential repetition according to Lemma \ref{lem: seq-rep}, but this would forego non-interactiveness. We believe that starting from the parallel amplified version of our protocol from section \ref{sec:parallel-ampl}, it should be possible to obtain a non-interactive argument of knowledge with negligible knowledge error, but we leave a more thorough investigation of this for future work. 
\end{remark}

\section{Proofs of quantum knowledge for QMA}
\label{sec:bjsw-poqk}

The interactive protocol that we show is a proof of quantum knowledge for languages in $\QMA$ is identical to the protocol from~\cite{broadbent2016zero}, as recalled in Section~\ref{sec:bjsw}, except for one modification: at the same time as the prover sends the encoded state $E(\rho)$ and the commitment $\sigma$ to the verifier (end of step 1 of the protocol), the prover also sends a classical zero-knowledge PoK of an opening to the commitment. More precisely, define a relation $R$ such that $R(\sigma, z) = 1$ if $z$ is a valid opening for the commitment $\sigma$. $V$ and $P$ engage in a ZK PoK protocol for the relation $R$ on common input $\sigma$, as defined in Definition~\ref{def:pok}. If the verifier rejects in this protocol, then the verifier outputs ``reject'' for the whole protocol; otherwise the verifier proceeds to the next phase.

We denote by $\Pi_{PoK}$ the ZK PoK protocol for relation $R$. Using~\cite[Corollary 21]{unruh2012quantum}, we may assume that this protocol has negligible knowledge error.

The main result of this section is that with this modification, the protocol from~\cite{broadbent2016zero} is a quantum proof of quantum knowledge for any language in $\QMA$. As in Section~\ref{sec: aoqk}, for simplicity we only consider relations $(Q,\alpha,\beta)$ such that $\alpha = 1-\negl(n)$ and $\beta=\negl(n)$. 
\begin{theorem}
\label{thm: main}
There is a polynomial $t$ such that the following holds. Let $(Q,1-\negl(n),\negl(n))$ be a $\QMA$ relation. There is a zero-knowledge PoQK for $(Q,1-\negl(n),\negl(n))$ with completeness $1-\negl(n)$, negligible soundness and knowledge error, and quality $q(n,\eps)=1-\poly((1-\eps) t(n))$. 
\end{theorem}

\begin{proof}
Given a $\QMA$ relation $(Q,1-\negl(n),\negl(n))$ and an instance $x$, as in Section~\ref{sec: aoqk} by applying Lemma~\ref{lem:5local} we can efficiently construct a $5$-local Clifford Hamiltonian $H(x)$ and an efficient quantum channel $\tau$ such that there is a polynomial $t$ such that for any state $\rho$ such that $\Tr(H\rho) \leq \eps$ it holds that $(x,\tau(\rho)) \in R_{Q,1-\poly(\eps t)}$. We then consider the protocol from~\cite{broadbent2016zero} for the Hamiltonian $H(x)$, modified as described at the start of the section. Completeness holds trivially. Provided $t$ is chosen large enough, soundness is at most $1-1/t$. We first show that the knowledge error is at most $1-1/(2t)$. Using Lemma \ref{sec:seq-rep}, both soundness and knowledge error can be improved by sequential repetition.

The Zero-Knowledge property follows essentially as in~\cite{broadbent2016zero}, with the only difference being the addition of the quantum zero-knowledge PoK of an opening to the commitment, that also has the Zero-Knowledge property. 

We show that the protocol is a quantum proof of quantum knowledge by constructing an extractor $K$ as follows. Let $P^*$ be an arbitrary (quantum polynomial-time) prover in the protocol, and $(x,\rho)$ a pair of a problem instance and a quantum witness such that the protocol execution $(P^*(x,\rho),V(x))$ returns $1$ with probability at least $1-s$ for some $s\leq \min\{1/t,1/2\}$. 

Informally, the extractor $K$ first takes the quantum state $\rho^*$ sent by $P^*$ in the first message. It then executes an extractor $K'$ for an opening to the commitment sent in the first message, that must exist by the quantum proof of knowledge property for the sub-protocol. If $K'$ succeeds in recovering the committed keys, $K$ decodes the state received in the first message using these keys and returns the decoded state. Otherwise, $K$ returns an abort symbol ``$\perp$''. 

We provide more detail. 
Denote by $\textsf{N}$ and $\textsf{S}_P$ the message and state registers of $P$ respectively. Initially, $\textsf{S}_P$ contains the witness $\rho$ and $\textsf{N}$ is initialized to a zero string. Let $U_P$ be the unitary applied by the prover at the first step of the protocol. (For simplicity, we omit dependence on the instance $x$.)

The extractor $K$ operates as follows. 
(In the following description, we sometimes write for simplicity that ``$K$ measures a certain register''. Formally, however, $K$ does not perform any measurement, and what we mean is that $K$ uses its auxiliary register $\textsf{S}_{K_{aux}}$ to implement the channel corresponding to these measurements in a coherent way.) 
\begin{itemize}
    \item[(i)] $K$ executes $U_{P}$ on registers $\textsf{N}$ and $\textsf{S}_P$. Let $n$ is the number of logical qubits of a witness. $K$ swaps the first $2nN$ qubits of $\textsf{N}$ (the ones that the prover would have sent to the verifier) in an auxiliary register $\textsf{S}_{K_{q}}$, and the qubits corresponding to the ``commitment'' into an auxiliary register $\textsf{S}_{K_c}$. $K$ measures $\textsf{S}_{K_c}$ to obtain a string $\sigma$. 
    \item[(ii)] Denote by $\tilde{\rho}$ the state of $\textsf{S}_{K_{q}}$. Let $\rho'$ be the leftover state of register $\textsf{S}_P$. Let $K'$ be the knowledge extractor for $\Pi_{QPoK}$. $K$ executes $K'$ on the remaining registers of the prover to obtain an outcome $z$. If $R(\sigma, z) = 1$ then $K$ proceeds to the next step. Otherwise $K$ halts with the outcome ``$\perp$'' on its output register.
    \item[(iii)] Let the opened commitment be $z=(t, \pi, a,b)$. $K$ applies the map 
	$Dec_{\pi,a,b}$ from Definition~\ref{def: dec auth} to the state $\tilde{\rho}$ contained in register $\textsf{S}_{K_{q}}$, followed by the channel $\Phi^{\otimes n}$ from Lemma~\ref{lem: technical bjsw}, followed by the map $\tau$ described at the start of the proof. $K$ returns the resulting $n$-qubit state in its output register. 
\end{itemize}
We show that $K$ is a valid extractor for the PoQK property. First note that the state of $\textsf{S}_{K_q} \otimes \textsf{S}_{K_c}$ after step (i), conditioned on not returning ``$\perp$'', takes the form 
\begin{equation}
\label{eq: 1}
    \sum_v p_\sigma \,\rho_\sigma \otimes \ket{\sigma}\!\bra{\sigma}\;,
\end{equation}
where each $\sigma$ is a valid commitment. Let $1-s(\sigma)$ be the probability that $V$ accepts $P$ in the protocol, conditioned on $V$ having received commitment $\sigma$. Then $\sum_\sigma p_{\sigma}s(\sigma) \leq s$. Let $p,q$ be a polynomials such that $K'$ succeeds at extracting an opening for a valid commitment with probability at least $\frac{1}{q(\mu)}p(1-s(\sigma))$. Denote by $\xi_{\sigma}$ the state returned by the verifier at step (iii) prior to applying the map $\tau$, when the commitment is $\sigma$ and conditioned on the extraction performed by $K'$ having succeeded. Thus as long as $s(\sigma)<1$ (which implies that $\sigma$ is a valid commitment), the probability that the prover is eventually accepted in the protocol, conditioned on having given commitment $\sigma$ and the extraction having succeeded is at least $1-\frac{s(\sigma) q(\mu)}{p(1-s(\sigma))}$. Applying Lemma~\ref{lem: technical bjsw}, we get 
\[ \Tr\big( H \xi_\sigma \big) \leq \frac{s(\sigma) q(\mu)}{p(1-s(\sigma))} + \negl(n) \;.\]
By Markov's inequality, a fraction at least $1/2$ (under the distribution $(p_\sigma)$) of commitments $\sigma$ are such that $s_{\sigma} \leq 2s$. For any such $\sigma$, conditioned on the extraction performed by $K'$ succeeding, $K$ obtains a state such that $\Tr( H \xi_\sigma ) < 2s q(\mu)/p(1-2s) + \negl(n)$. By the discussion at the start of the proof, in any such case the extractor returns $\tau(\xi_\sigma)$ such that $(x,\tau(\xi_\sigma)) \in R_{Q,1-\poly(2s t q/p(1-2s))}$. To obtain the desired conclusion, we combine the polynomials $q$ and $t$.
\end{proof}

\bibliographystyle{alpha}
\bibliography{references}

\end{document}